\documentclass[12pt,a4paper]{article}
\usepackage{amsmath, amssymb, amsthm}
\usepackage[font={footnotesize,it}]{caption}
\usepackage[top=1.2in, bottom=1.2in, left=1.1in, right=1.1in]{geometry}
\usepackage[normalem]{ulem}

\DeclareCaptionStyle{italic}[justification=centering]{labelfont={bf},textfont={it},labelsep=colon}
\usepackage{graphicx,psfrag,epsf}
\usepackage{enumerate}
\usepackage{natbib, subcaption}
\usepackage{url, setspace, titling} 
\usepackage{booktabs, bm, paralist,mathpazo,tikz,todonotes,longtable,microtype,dsfont,rotating} 
\usepackage[pdftex,colorlinks=true]{hyperref}
\usepackage{mathtools}	      
\usepackage{appendix}
\usepackage{mathrsfs}
\usepackage{multirow}
\usepackage[ruled,linesnumbered]{algorithm2e}
\usepackage{prettyref, cleveref}
\newrefformat{hypo}{H\ref{#1}}
\newcounter{hypothesis}

\definecolor{darkblue}{rgb}{0,0,.6}
\hypersetup{citecolor=darkblue,linkcolor=darkblue,urlcolor=darkblue}

\newtheorem{theorem}{Theorem}
\newtheorem{assumption}{Assumption}
\newtheorem{lemma}{Lemma}

\newtheorem{proposition}{Proposition}
\newtheorem{remark}{Remark}
\newtheorem{condition}{Identification Condition}
\newcommand{\overbar}[1]{\mkern 1.5mu\overline{\mkern-1.5mu#1\mkern-1.5mu}\mkern 1.5mu}

\newcommand{\xdownarrow}[1]{%
  {\left\downarrow\vbox to #1{}\right.\kern-\nulldelimiterspace}
}
\newcommand{\blind}{0}

\graphicspath{{plots/}}
\DeclareMathOperator*{\argmin}{\arg\!\min}
\newsavebox\CBox

\date{\today}

\begin{document}

\def\spacingset#1{\renewcommand{\baselinestretch}
{#1}\small\normalsize} \spacingset{1}

\if0\blind
{
  \title{\bf Factor-augmented Smoothing Model for Functional Data}
\author{Yuan Gao\thanks{Postal address: Research School of Finance, Actuarial Studies and Statistics, Level 4, Building 26C, Kingsley St, Australian National University, Canberra, ACT 2601, Australia; Email: yuan.gao@anu.edu.au}
 \hspace{.2cm}\\
Research School of Finance, Actuarial Studies and Statistics\\ 
Australian National University \\ 
\\
  Han Lin Shang
  \hspace{.2cm}\\
    Department of Actuarial Studies and Business Analytics \\
    Macquarie University \\
 \\ 
 Yanrong Yang \\  
Research School of Finance, Actuarial Studies and Statistics\\ 
Australian National University \\ 
 }
  \maketitle
} \fi

\if1\blind
{
   \title{\bf Factor-augmented Smoothing Model for Functional Data}
   \author{}
   \maketitle
} \fi

\bigskip

\begin{abstract}
We propose modeling raw functional data as a mixture of a smooth function and a high-dimensional factor component. The conventional approach to retrieving the smooth function from the raw data is through various smoothing techniques. However, the smoothing model is not adequate to recover the smooth curve or capture the data variation in some situations. These include cases where there is a large amount of measurement error, the smoothing basis functions are incorrectly identified, or the step jumps in the functional mean levels are neglected. To address these challenges, a factor-augmented smoothing model is proposed, and an iterative numerical estimation approach is implemented in practice. Including the factor model component in the proposed method solves the aforementioned problems since a few common factors often drive the variation that cannot be captured by the smoothing model. Asymptotic theorems are also established to demonstrate the effects of including factor structures on the smoothing results. Specifically, we show that the smoothing coefficients projected on the complement space of the factor loading matrix is asymptotically normal. As a byproduct of independent interest, an estimator for the population covariance matrix of the raw data is presented based on the proposed model. Extensive simulation studies illustrate that these factor adjustments are essential in improving estimation accuracy and avoiding the curse of dimensionality. The superiority of our model is also shown in modeling Canadian weather data and Australian temperature data.
\\

\end{abstract}

\newpage
\spacingset{1.48}

\section{Introduction}

With the increasing capability to store data, functional data analysis (FDA) has received growing attention over the last 20 years. Functional data are considered realizations of smooth random objects in graphical representations of curves, images, and shapes. The monographs of \cite{RS02, RS05} and \cite{RH17} provide a comprehensive account of the methodology and applications of the FDA; other relevant monographs include \cite{FV06} and \cite{HK12}. More recent advances in this field can be found in many survey papers \citep[see, e.g.,][]{C14,FGG17, GV16, RGS+17, WCM16}. One main challenge in the FDA lies in the fact that we cannot observe functional curves directly, but only discrete points, which are often contaminated by measurement errors. 
To model a mixture of functional data and high-dimensional measurement error, we introduce a factor-augmented smoothing model (FASM).

We denote a random sample of $n$ functional data as $\mathcal{X}_i(u), i = 1,\dots, n$, and $u \in \mathcal{I}\subset \mathbb{R}$, where $\mathcal{I}$ is a compact interval on the real line $\mathbb{R}$. In practice, the observed data are discrete points and are often contaminated by noise or measurement error. We use $Y_{ij}$ to represent the $j$th observation on the $i$th subject; the observed data can then be expressed as a ``signal plus noise" model:
\begin{equation*}
Y_{ij} = \mathcal{X}_i(u_j) + \eta_{ij}, \quad  j = 1,\dots p, \ i=1, \ldots, n.
\end{equation*} 
We use $\mathcal{X}_i(u_j)$ to denote the realization of the $j$th discrete point on the curve $\mathcal{X}_i(\cdot)$, and $\eta_{ij}$ is the noise or measurement error. We assume that measurement error only occurs where the measurements are taken; thus, the error $\bm{\eta}_{i} = (\eta_{i1},\dots,\eta_{ip})$ is a multivariate term of dimension $p$. Though in practice, the signal function component $\bm{\mathcal{X}}_i = (\mathcal{X}_i(u_1),\dots, \mathcal{X}_i(u_j))$ is of the same $p$ dimension, it differs from $\bm{\eta}_{i}$ in nature. Although functions are potentially infinite-dimensional, we may impose smoothing assumptions on the functions, which usually implies functions possess one or more derivatives. This smoothness feature is used to separate the functions from measurement errors -- a functional smoothing procedure.

When the variance of the noise level is a tiny fraction of the variance of the function, we say the signal-to-noise ratio is high. In this case, classic smoothing tools apply to functional data, including kernel methods \citep[e.g.][]{WJ95}, local polynomial smoothing \citep[e.g.][]{FG96}, and spline smoothing \citep[e.g.][]{W90, E99, GS94}. With pre-smoothed functions, estimates, such as mean and covariance functions, can be further obtained. More recent studies on functional smoothing approaches include \cite{CY11, YL13}, and \cite{ZW16}. In this article, we apply basis smoothing to the functions $\mathcal{X}_i(u)$; that is, we represent $\mathcal{X}_i(u)$ as $\mathcal{X}_i(u) = \sum_{k=1}^K c_{ik}\phi_k(u)$, where $\{\phi_k(u), \ k = 1,\dots, K\}$ are the basis functions and $\{c_{ik}, \ i=1,\dots, n,\ k= 1,\dots, K\}$ are the smoothing coefficients. The smoothing model then becomes
\begin{equation*}
Y_{ij} =  \sum_{k=1}^K c_{ik}\phi_k(u_j) + \eta_{ij}, \quad  j = 1,\dots p, \ i=1, \ldots, n.
\end{equation*} 

When the signal-to-noise level is low, smoothing tools may not be adequate in removing the measurement error and may cause an inefficient estimation of the smoothing coefficients. Let us take a further look at the measurement error $\eta_{ij}$. In the FDA, the number of discrete points $p$ on each subject is often large compared with the sample size $n$. Hence the term $\bm{\eta}_{i}$ is a high-dimensional component. In this case, the observed data are, in fact, a mixture of functional data and high-dimensional data. The existence of the large measurement error $\eta_{ij}$ raises the curse of dimensionality problem, which naturally calls for the application of dimension reduction models to $\eta_{ij}$. Many studies have been conducted on various dimension reduction techniques for high-dimensional data; among theses, factor models are widely used \citep[e.g.][]{FFL08, LYB11}. 

We propose using a factor model for the measurement error term. Without further information on the measurement error, factor model is appropriate since the estimation of latent factors does not require any observed variables. The high-dimensional measurement error is assumed to be driven by a small number of unobserved common factors.
\begin{equation*}
\eta_{ij} = \bm{a}_j^\top\bm{f}_i + \epsilon_{ij}, \quad i = 1,\dots, n, \  j = 1,\dots, p,
\end{equation*}
where $\bm{f}_i \in \mathbb{R}^r$ are the unobserved factors, $\bm{a}_j \in \mathbb{R}^r$ are the unobserved factor loadings, $r$ is the number of latent factors, and $\epsilon_{ij}$ are idiosyncratic errors with mean zero. Thus, the observed data $Y_{ij}$ can be written as the sum of two components:
\begin{equation*}
Y_{ij} =  \sum_{k=1}^K c_{ik}\phi_k(u_j) + \bm{a}_j^\top\bm{f}_i  + \epsilon_{ij}, \quad i = 1,\dots, n,\  j = 1,\dots, p.
\end{equation*}

This is a basis smoothing model with the factor-augmented form. This proposed model can be easily modified to adopt nonparametric smoothing methods. In Section~\ref{se:nonparam}, we illustrate the use of spline smoothing approaches. In Section~\ref{se:6.7}, the nonparametric smoothing model is applied to simulated data.

In this paper, we motivate the FASM in three considerations, as listed below. In these three cases, using the proposed model remedies the defects of the traditional smoothing model. Examples of the following three motivations are provided in Section~\ref{se:2}.
\begin{enumerate}
\item
In traditional smoothing models, the measurement error $\eta_{ij}$ is assumed to be non-informative and independently and identically distributed (i.i.d.) in both directions. This is an unrealistic assumption when the measurement errors contain information. With the factor model applied, we assume that a small number of unobserved factors can capture the covariance in the measurement error. This is usually reasonable in practice because a few common factors often drive the occurrence of systematic measurement error.
\item
When the smoothing basis functions are incorrectly identified, the smoothing model will lead to an erroneous coefficient estimate and large residuals. The proposed model deals with this problem since the unexplained variation resulting from the basis's misidentification can be modeled with a small number of unobserved common factors.
\item
When there are step jumps in the mean level of the functions, neglecting the mean shift in smoothing models will result in large residuals at the point where the jumps occur. The changes in the mean levels of the functions come from a universal source and can be modeled by common factors.
\end{enumerate}

Since the latent factors are unobserved, we propose an iterative approach to estimate the smooth function and the factors simultaneously. Principal component analysis (PCA) is used as a tool in estimating the factor model, and penalized least squares estimation is applied to construct the estimator for the smoothing coefficient $c_{ik}$. We establish the asymptotic theories of the smoothing coefficient estimator, where the consistency of the estimator is proved. We also provide the asymptotic distribution of the projected estimator in the orthogonal complement of the space spanned by the factors $\bm{f}_i$. The interplay between the smooth component and the factor model component is manifested.



In the remainder of this article, we elaborate on the previously mentioned three motivations in detail, with examples given in Section~\ref{se:2}. In Section~\ref{se:3}, the model is formally stated, and the iterative estimation approach is provided. We discuss the asymptotic properties of the smoothing coefficients under various assumptions in Section~\ref{se:4}. We extend the proposed model to a nonparametric smoothing approach in Section~\ref{se:nonparam}. In Section~\ref{se:5}, we consider the statistical inference aspect of the model and propose a covariance matrix estimator for the raw data. In Section~\ref{se:6}, we conduct Monte-Carlo simulations on the proposed model under different settings. A few real data examples are given in Section~\ref{se:7}, and conclusions are drawn in Section~\ref{se:8}. Last, we provide proofs of the relevant theorems and lemmas in the Appendix.

\section{Motivation}\label{se:2}

We introduce three examples to motivate the proposed model. In these cases, the smoothing model is not adequate to capture the raw data's signal information. In the first example, when large measurement error exists, the residuals after smoothing are large with some extreme values. In the second example, when the basis functions are selected incorrectly, part of the functions' variation cannot be captured by the smoothing model. In the third example, when there are step jumps in the functional data, the residuals after smoothing contain gaps. These examples demonstrate that further modeling of the residuals is needed.

\subsection{Functional data with measurement error}\label{se:2.1}

Figure~\ref{fig:2} shows the rainbow plots of the average daily temperature and log precipitation at 35 locations in Canada. Due to the nature of the two kinds of data, it is reasonable to assume that temperature and log precipitation are functions over time. The two graphs, however, display distinct features. In the temperature plot, though there are some perturbations, it is relatively easy to discern each curve's shape. In the precipitation plot, there is a tremendous amount of variability in the raw data, such that it is almost impossible to observe the underlying shape of the curves. 

Smooth temperature data can be retrieved without much difficulty using basic smoothing techniques. The residuals are small, with constant variation. On the other hand, the residuals after smoothing exhibit a high level of variation for the precipitation data and even contain some extreme values. Our model endeavors to further explain the large residuals in similar cases to the precipitation data; we will show the fitting result in Section~\ref{se:7}.

\begin{figure}[!htbp]
\center
\includegraphics[width = 8.65cm]{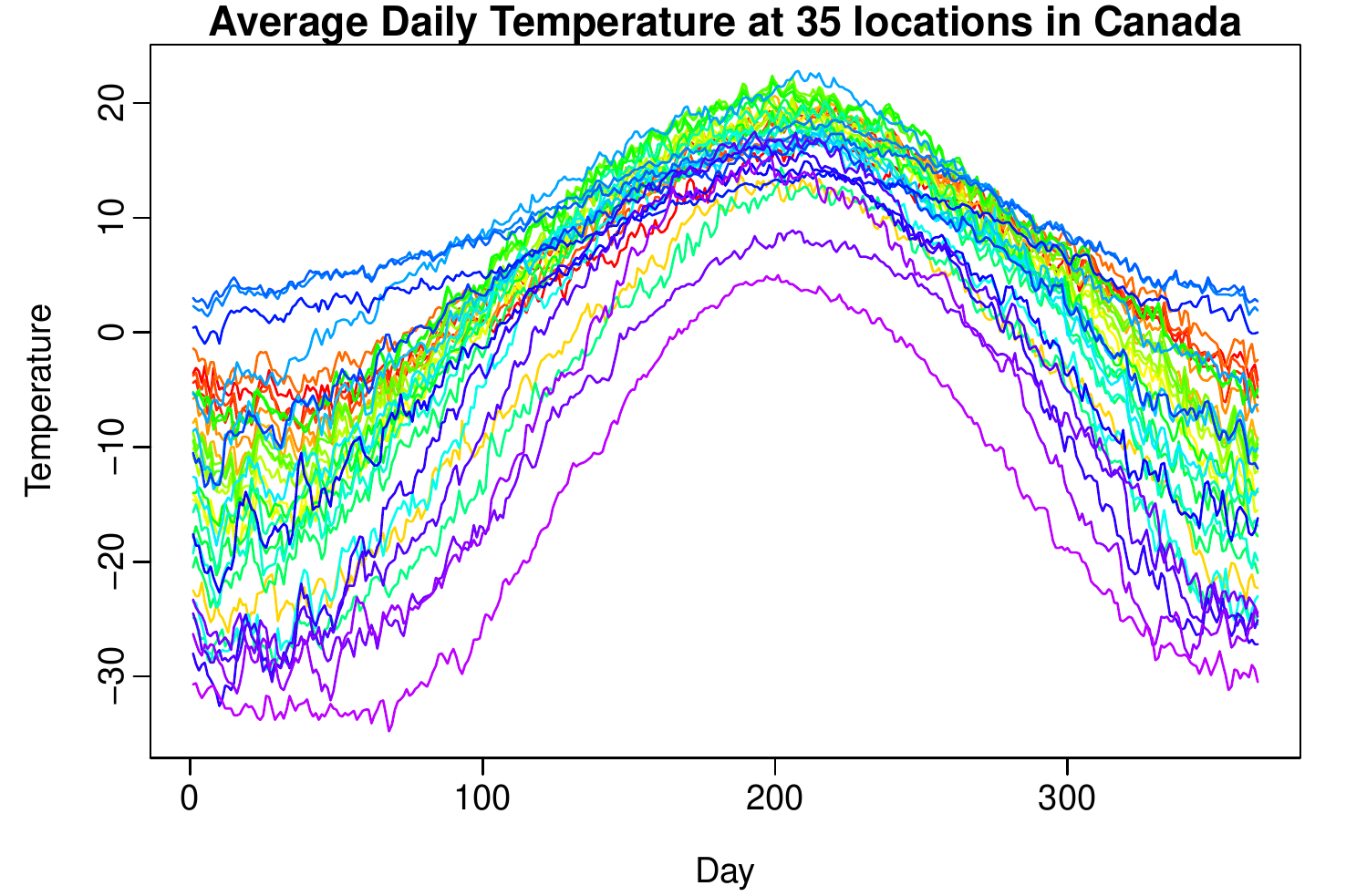}
\quad
\includegraphics[width = 8.65cm]{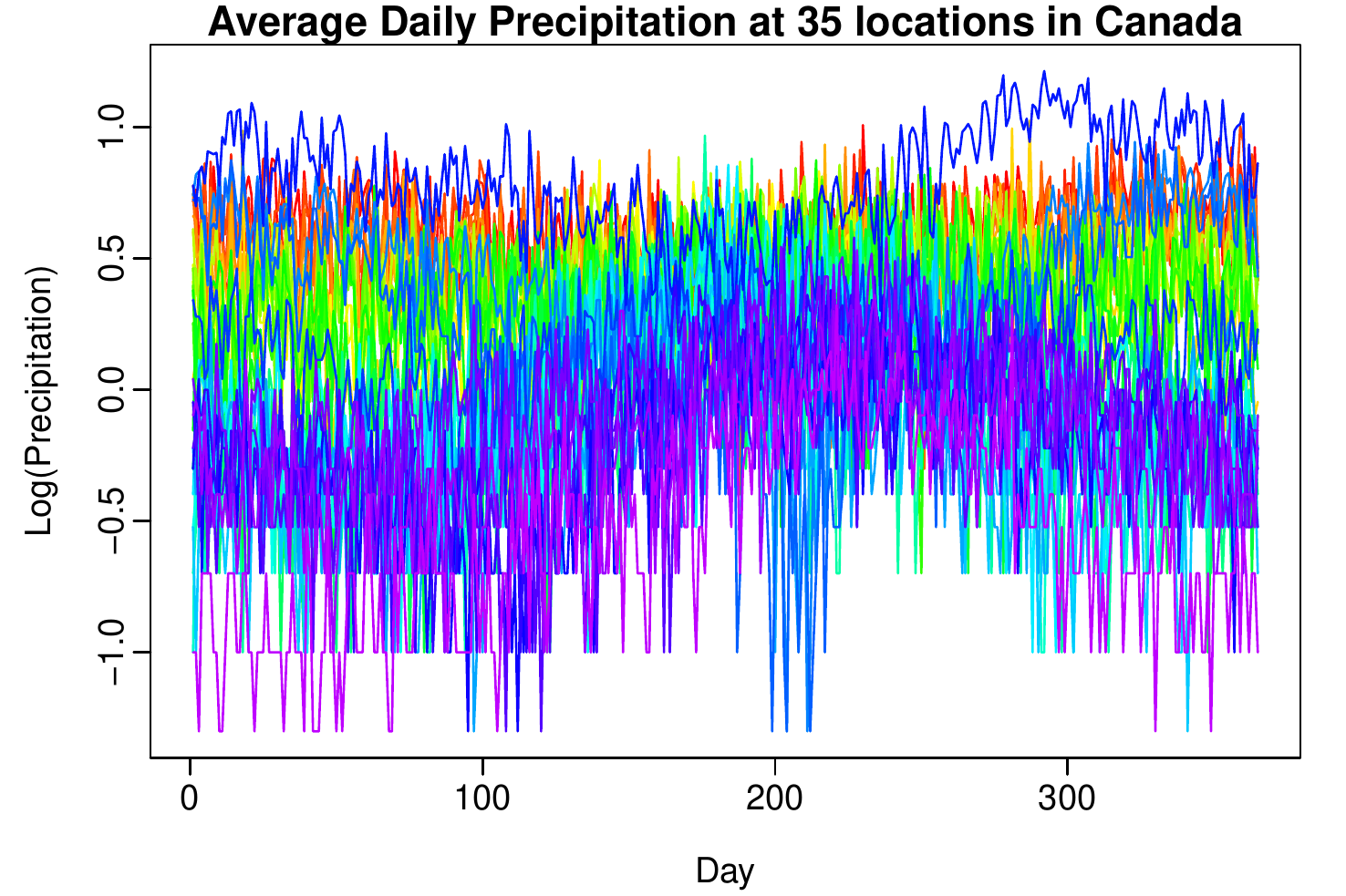}
\caption{Average daily temperature and log precipitation in 35 Canadian weather stations averaged over the year 1960 to 1994.}\label{fig:2}
\end{figure}

\subsection{Misidentification of the basis function}\label{se:2.2}

It is important to choose the appropriate basis functions in the smoothing method. In this example, we show the inadequacy of the smoothing model when the basis functions are misidentified. We generate functional data using basis functions with changing frequencies. The raw data are shown in Figure~\ref{fig:9a}. Fourier basis functions are used. In the second half of the data, the frequency of the Fourier basis functions increases, so the data set exhibits more variation toward the right end. Suppose that we were not aware of the change in the frequencies in the basis functions, and still used the basis of the first half of the data for the whole curves. The consequence of misidentifying the basis functions when a smoothing model is applied can be observed in Figure~\ref{fig:9b}. The residuals are large in the second half. The smoothing model fails to reduce the residuals; a factor model can be used to further model the signal hidden in the large residuals. The data generation process and further analysis can be found in Section~\ref{se:6.5}.

\begin{figure}[!htbp]
\center
\subcaptionbox{Raw data \label{fig:9a}}
{\includegraphics[width = 8.65cm]{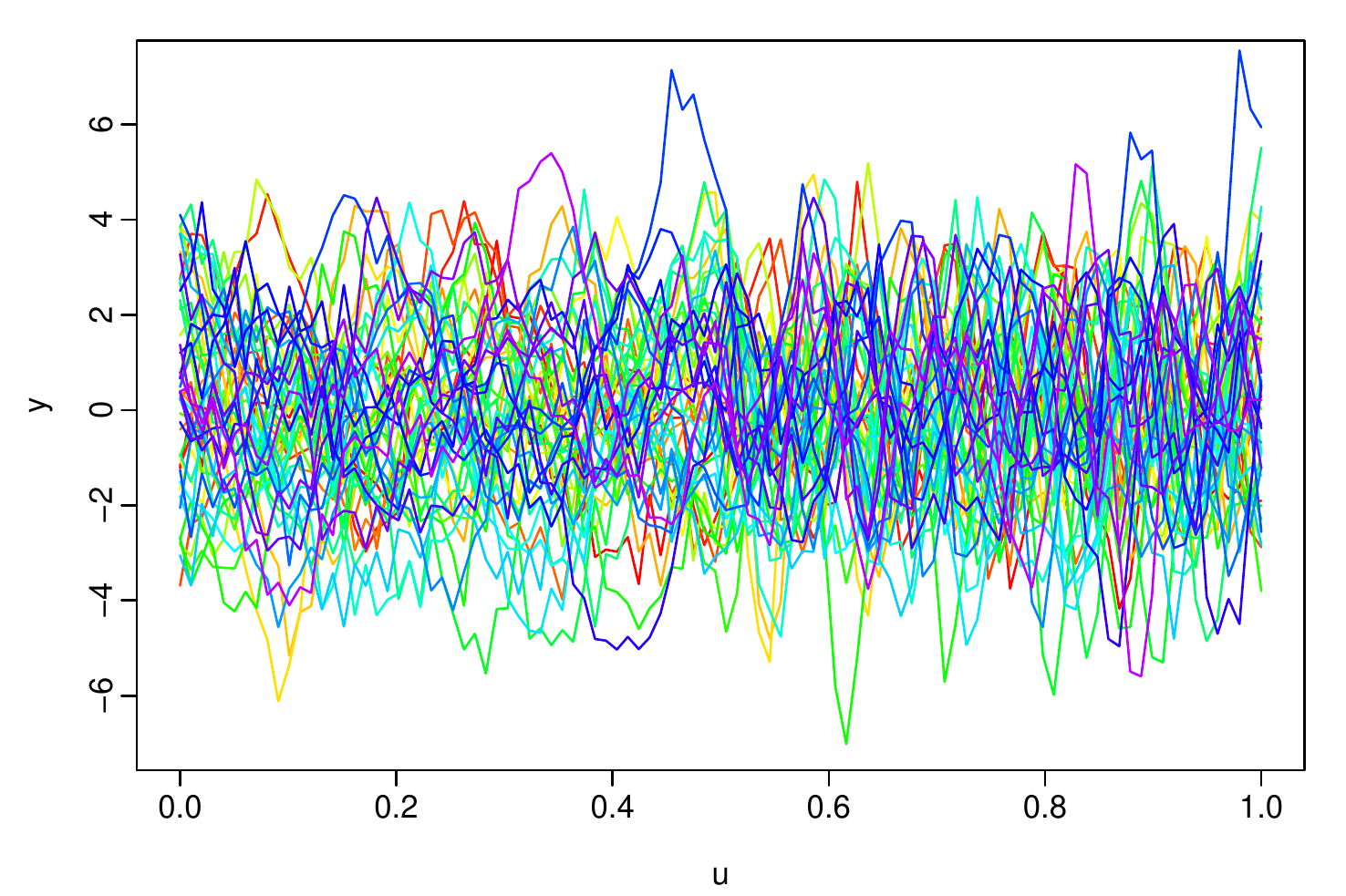}}
\subcaptionbox{Residuals\label{fig:9b}}
{\includegraphics[width = 8.65cm]{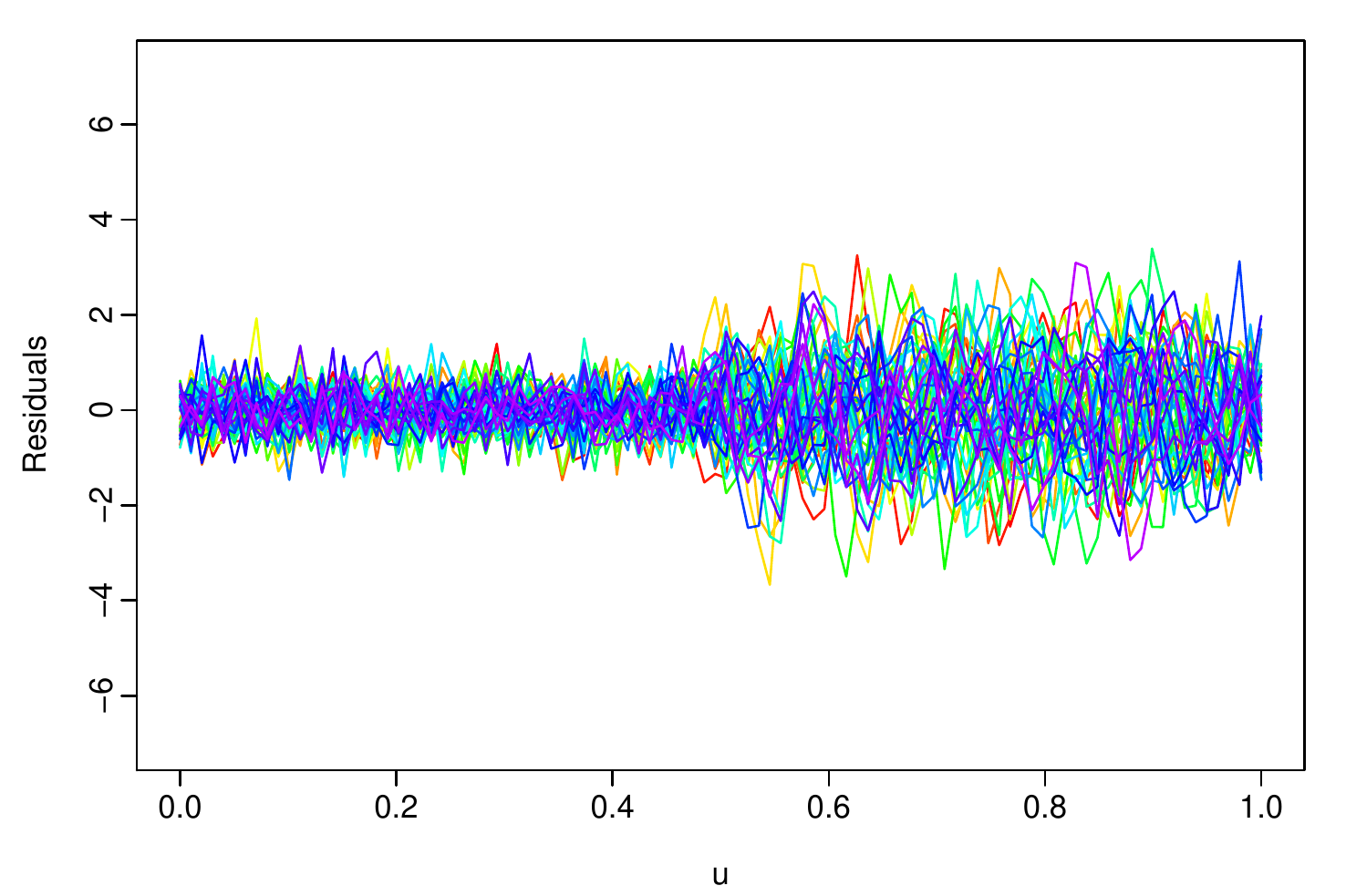}}
\caption{A simulated sample of functional data with changing basis functions.}\label{fig:9}
\end{figure}

\subsection{Functional data with step jumps in the mean level}\label{se:2.3}

We provide another example of functional data with step jumps to motivate our proposed model. Suppose we observed a sample of the raw functional data, as shown in Figure~\ref{fig:1a}. It can be seen that there is a jump at around $u = 0.5$. The jump applies to all the sample data, so this sudden shift is at the mean level. We will explain how the data are generated in Section~\ref{se:6.6}. The residuals after smoothing are presented in Figure~\ref{fig:1b}. The large residuals around the jump clarify that without measures to deal with the step jumps, smoothing itself is not enough to model these kinds of data. We show in Section~\ref{se:6.6} that the proposed model applied to the same data generates smaller residuals and has less flexibility. This is indeed one of the main goals because of model selection.

\begin{figure}[!htbp]
\center
\subcaptionbox{Raw data\label{fig:1a}}
{\includegraphics[width = 8.65cm]{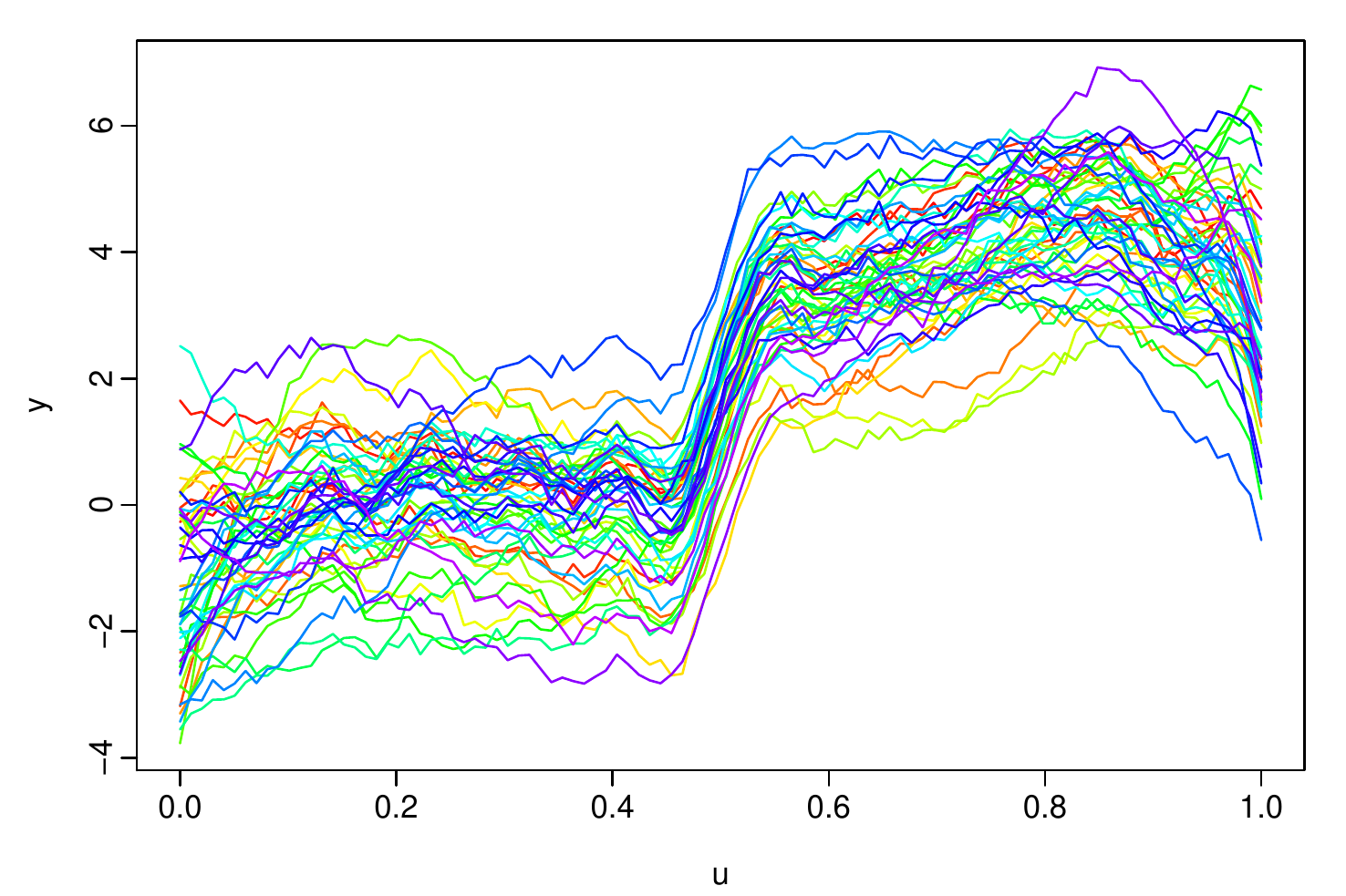}}
\subcaptionbox{Residuals\label{fig:1b}}
{\includegraphics[width = 8.65cm]{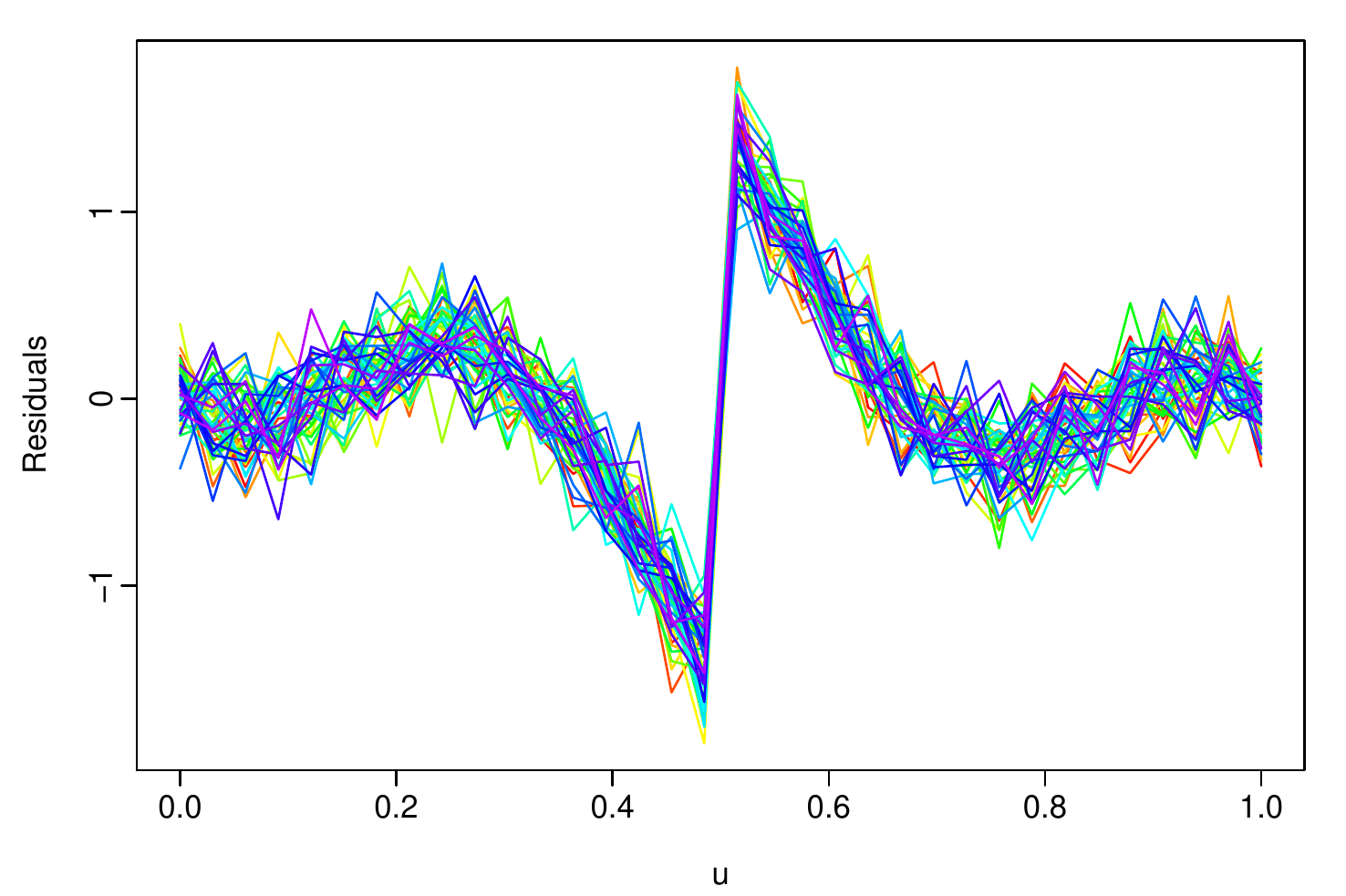}}
\caption{A simulated sample of functional data with step jump.}\label{fig:1}
\end{figure}

\section{Model specification and estimation}\label{se:3}

In this section, we formally state the proposed model in Section~\ref{se:3.1} and provide the estimation method in Section~\ref{se:3.2}. We first show how the smoothing coefficient $\bm{c}_i$ and the latent factors $\bm{f}_i$ are estimated separately and then introduce an iterative approach to simultaneously find these estimates. 

\subsection{Model statement}\label{se:3.1}

We consider a sample of functional data $\mathcal{X}_i(u)$, which takes values in the space $H:= L^2(\mathcal{I})$ of real-valued square integrable functions on $\mathcal{I}$. The space $H$ is a Hilbert space, equipped with the inner product $\langle x, y\rangle := \int x(u)y(u) du$. The function norm is defined as $ \|x\| := \langle x, x \rangle^{1/2}. $ The functional nature of $\mathcal{X}_i(u)$ allows us to represent it as a linear expansion of a set of $K$ smooth basis functions.
\begin{equation*}
\mathcal{X}_i(u) = \sum_{k=1}^K c_{ik} \phi_k(u), \quad u \in \mathcal{I},
\end{equation*}
where $\phi_k(u)$ is a set of common basis functions and $c_{ik}$ is the $k$th coefficient for the $i$th curve. Therefore, we can express the full model as
\begin{align*}
Y_{ij} &= \sum_{k=1}^K c_{ik}\phi_k(u_j) + \eta_{ij},\\
\eta_{ij} &= \bm{a}_j^\top\bm{f}_i  + \epsilon_{ij}, \quad i = 1,\dots, n, \  j = 1,\dots, p,
\end{align*}
where $\bm{f}_i\in \mathbb{R}^r$ are the unobserved common factors, $\bm{a}_j\in \mathbb{R}^r$ are the unobserved factor loadings and $r$ is the number of factors. We call this model the FASM. For the model to be identifiable, we require the following condition.

\begin{condition}\label{id}
We require
\begin{enumerate}
\item[(i)] $\mathcal{X}_i(u_j)$ is independent of $\eta_{ij}$ for $i=1,\dots, n,\ j=1,\dots, p$, and 
\item[(ii)]
$\frac{1}{p}\sum_{j=1}^p \bm{a}_j\bm{a}_j^\top \overset{p}{\to} \Sigma_{\bm{a}} > 0$ for some $r\times r$ matrix $\Sigma_{\bm{a}}$, as $p\rightarrow \infty$;

$ \frac{1}{n}\sum_{i=1}^n\bm{f}_i \bm{f}_i^\top \overset{p}{\to} \Sigma_{\bm{f}} > 0$ for some $r\times r$ matrix $\Sigma_{\bm{f}}$, as $n\rightarrow \infty$.
\end{enumerate}
\end{condition}

The first part of the identification condition ensures the signal function component and the factor model component are independent. The second part ensures the existence of $r$ factors, each of which makes a non-trivial contribution to the variance of $\eta_{ij}$, which in turn guarantees the identifiability between the factors and the error term $\epsilon_{ij}$.

We treat the basis functions $\phi_k(u)$ as known, and the number $K$ fixed. This is, of course, a simplification to accommodate for the theoretical proofs. In real data analysis, there are various choices for the basis functions, and the decision can be quite subjective. For example, Fourier bases are preferred for periodic data, while spline basis systems are most commonly used for non-periodic data. Other bases include wavelet, polynomial, and some ad-hoc basis functions. 


\subsection{Estimation}\label{se:3.2}

We can write the model for the $i$th object as
\begin{equation}\label{eq:14}
\bm{Y}_i = \bm{\Phi}\bm{c}_i + \bm{A}\bm{f}_i + \bm{\epsilon}_i,\quad i=1,\dots, n
\end{equation}
where 
\begin{align*}
\bm{Y}_i = \begin{bmatrix}
Y_{i1}\\ \vdots \\ Y_{ip}
\end{bmatrix},\quad
\bm{c}_i = \begin{bmatrix}
c_{i1}\\\vdots\\c_{iK}
\end{bmatrix},\quad
\bm{\Phi}  = \begin{bmatrix}
\phi_1(u_1) & \dots & \phi_K(u_1)\\\vdots && \vdots\\ \phi_1(u_p)&\dots& \phi_K(u_p)
\end{bmatrix}, \quad
\bm{A} = \begin{bmatrix}
\bm{a}_1^\top\\ \vdots\\ \bm{a}_p^\top 
\end{bmatrix}, \quad
\bm{\epsilon} = \begin{bmatrix}
\epsilon_{i1}\\ \vdots\\ \epsilon_{ip} 
\end{bmatrix}.
\end{align*}
Combining all the objects, we have in matrix form
\begin{equation}\label{eq:16}
\bm{Y} = \bm{\Phi}\bm{C} + \bm{A} \bm{F}^\top + \bm{E},
\end{equation}
where $\bm{Y}$ is $p\times n$ and $\bm{C} = (\bm{c}_1,\dots, \bm{c}_n)$ is a $K \times n$ matrix containing all the smoothing coefficients. The matrix $\bm{F} = (\bm{f}_1, \dots, \bm{f}_n)^\top$ is $n\times r$ and $\bm{E} = (\bm{\epsilon}_1, \dots, \bm{\epsilon}_n)$ is $p\times n.$ Since $\bm{\Phi}$ is assumed to be known, we illustrate how the parameters $\bm{C}, \bm{A}$ and $\bm{f}$ are estimated in the following.

For the latent factor estimation, there is an identification problem such that $\bm{A}\bm{F}^\top = \bm{A}\bm{U}\bm{U}^{-1}\bm{F}^\top$ for any $r\times r$ invertible matrix $\bm{U}$. Thus we impose the normalization restriction on the matrices $\bm{A}$ and $\bm{F}$
\begin{equation}\label{eq:50}
\bm{A}^\top \bm{A}/p = \bm{I}_r, \quad \text{and} \ \bm{F}^\top\bm{F}\ \text{is a diagonal matrix}.
\end{equation}

We propose to implement penalized least squares, where the objective function is defined as
\begin{equation*}
\text{SSR}(\bm{c}_i, \bm{A}, \bm{f}) = \sum_{i=1}^n \left[(\bm{Y}_i - \bm{\Phi}\bm{c}_i - \bm{A}\bm{f}_i)^\top  (\bm{Y}_i - \bm{\Phi}\bm{c}_i - \bm{A}\bm{f}_i ) + \alpha\times \text{PEN}_2(\mathcal{X}_i) \right],
\end{equation*}
where $\text{PEN}(\mathcal{X}_i)$ is a penalty term used for regularization, and $\alpha$ is the tuning parameter controlling the degree of regularization. The same $\alpha$ is used for all the functional observations $i$. This is a simplified case, where we assume a similar degree of smoothness for all curves. The tuning parameter can be chosen by cross-validation or information criteria. We intend to penalize the ``roughness" of the function term. To quantify the notion of ``roughness" in a function, we use the square of the second derivative. Define the measure of roughness as 
\begin{equation*}
\text{PEN}_2(\mathcal{X}_i) = \int_{\mathcal{I}} \left[D^2\mathcal{X}_i(s)\right]^2 ds,
\end{equation*}
where $D^2\mathcal{X}_i$ denotes taking the second derivative of the function $\mathcal{X}_i$, the larger the tuning parameter $\alpha$, the smoother the estimated functions we obtain. Further, we denote
\begin{equation}\label{eq:2}
\bm{\Phi}(u) = \left[\phi_1(u), \dots,\phi_K(u)\right]^\top.
\end{equation}
Then
\begin{equation*}
\mathcal{X}_i(u) = \bm{c}_i^\top\bm{\Phi}(u).
\end{equation*}
We can re-express the roughness penalty $\text{PEN}_2(\mathcal{X}_i)$ in matrix form as the following:
\begin{align*}
\text{PEN}_2(\mathcal{X}_i) &= \int_{\mathcal{I}} \left[D^2\mathcal{X}_i(s)\right]^2 ds\\
&= \int_{\mathcal{I}} \left[D^2\bm{c}_i^\top\bm{\Phi}(s)\right]^2 ds\\
&= \int_{\mathcal{I}} \bm{c}_i^\top D^2\bm{\Phi}(s)D^2\bm{\Phi}^\top (s) \bm{c}_i ds\\
&= \bm{c}_i^\top \left[\int_{\mathcal{I}} D^2\bm{\Phi}(s)D^2\bm{\Phi}^\top (s)ds\right] \bm{c}_i\\
&= \bm{c}_i^\top \bm{R} \bm{c}_i, \qquad i = 1,\dots, n
\end{align*}
where 
\begin{equation}\label{eq:55}
\bm{R} \equiv \int_{\mathcal{I}} D^2\bm{\Phi}(s)D^2\bm{\Phi}^\top (s)ds.
\end{equation}
Thee matrix $\bm{R}$ is the same for all subjects and the penalty term $\text{PEN}_2(\mathcal{X}_i)$ differs for each subject only by the coefficient $\bm{c}_i$. 

\begin{remark}
The number of smoothing coefficient $\bm{c}_i$ needed increases as the sample size increases. The inclusion of a penalty term not only penalizes the "roughness" of the smoothed function but also mitigates the effect of increasing number of parameters to control the model flexibility.
\end{remark}

Thus, the objective function can be written as 
\begin{equation*}
\text{SSR}(\bm{c}_i, \bm{A}, \bm{f}) = \sum_{i=1}^n \left[(\bm{Y}_i - \bm{\Phi}\bm{c}_i - \bm{A}\bm{f}_i)^\top  (\bm{Y}_i - \bm{\Phi}\bm{c}_i - \bm{A}\bm{f}_i ) + \alpha \bm{c}_i^\top \bm{R}\bm{c}_i \right],
\end{equation*}
subject to the constraint $\bm{A}^\top\bm{A}/p = \bm{I}_r$. 

We aim to estimate the smoothing coefficient of $\bm{c}_i$. We left multiply a matrix to each term in~\eqref{eq:14} to project the factor model term onto a zero matrix. Define the projection matrix 
\begin{equation}\label{eq:51}
\bm{M}_{\bm{A}} \equiv \bm{I}_p -\bm{A}(\bm{A}^\top \bm{A})^{-1}\bm{A}^\top = \bm{I}_p -\bm{A}\bm{A}^\top/p.
\end{equation}
Then
\begin{equation*}
\bm{M}_{\bm{A}}\bm{A}\bm{f}_i = \left(\bm{I}_p -\bm{A}\bm{A}^\top/p\right)\bm{A}\bm{f}_i = \left(\bm{A} - \bm{AA}^\top\bm{A}/p\right)\bm{f}_i = \bm{0}.
\end{equation*}
So we estimate $\bm{c}_i$ from the projected equation
\begin{equation*}
\bm{M}_{\bm{A}}\bm{Y}_i = \bm{M}_{\bm{A}}\bm{\Phi}\bm{c}_i + \bm{M}_{\bm{A}}\bm{\epsilon}_i.
\end{equation*}

The projected objective function becomes
\begin{equation}\label{eq:22}
\text{SSR}(\bm{c}_i, \bm{A}) = \sum_{i=1}^n \left[(\bm{M}_{\bm{A}}\bm{Y}_i - \bm{M}_{\bm{A}}\bm{\Phi}\bm{c}_i)^\top  (\bm{M}_{\bm{A}}\bm{Y}_i - \bm{M}_{\bm{A}}\bm{\Phi}\bm{c}_i ) + \alpha \bm{c}_i^\top \bm{R}\bm{c}_i \right].
\end{equation}
By taking the derivative of $\text{SSR}(\bm{c}_i, \bm{A})$ with respective to each $\bm{c}_i$, we can solve for the estimator $\widehat{\bm{c}}_i$.
\begin{equation*}
\frac{\partial \text{SSR}(\bm{c}_i\bm{A})}{\partial \bm{c}_i} = (\bm{M}_{\bm{A}}\bm{Y}_i - \bm{M}_{\bm{A}}\bm{\Phi}\bm{c}_i)^\top (\bm{M}_{\bm{A}}\bm{\Phi}) + 2\alpha \bm{c}_i^\top \bm{R}.
\end{equation*}
Setting the derivative to zero and rearranging the terms, we have
\begin{equation*}
\left(\bm{\Phi}^\top \bm{M}_{\bm{A}}^\top\bm{M}_{\bm{A}}\bm{\Phi} + \alpha \bm{R} \right)\bm{c}_i = \bm{\Phi}^\top \bm{M}_{\bm{A}}^\top\bm{M}_{\bm{A}}\bm{Y}_i.
\end{equation*}
Using the fact that 
\begin{equation*}
\bm{M}_{\bm{A}}^\top \bm{M}_{\bm{A}} = \left(\bm{I}_p -\bm{A}\bm{A}^\top/p\right)^\top \left(\bm{I}_p -\bm{A}\bm{A}^\top/p\right) = \bm{M}_{\bm{A}},
\end{equation*}
we obtain the least squares estimator for $\bm{c}_i$ given $\bm{A}$
\begin{equation*}
\widehat{\bm{c}}_i = \left(\bm{\Phi}^\top \bm{M}_{\bm{A}}\bm{\Phi} + \alpha \bm{R} \right)^{-1} \bm{\Phi}^\top \bm{M}_{\bm{A}}\bm{Y}_i.
\end{equation*}

Next, to estimate $\bm{A}$ and $\bm{f}_i$, we focus on the factor model
\begin{equation*}
\bm{\eta}_i = \bm{A}\bm{f}_i + \bm{\epsilon}_i,
\end{equation*}
and in matrix form 
\begin{equation*}
\bm{Z} = \bm{A}\bm{F}^\top + \bm{E},
\end{equation*}
where $\bm{Z} = (\bm{\eta}_1, \dots, \bm{\eta}_n)$. In high dimensions, the unknown factors and loadings are typically estimated by least squares (i.e., the principal component analysis; see, e.g., \cite{FFL08, O12}. The least squares objective function is
\begin{equation}\label{eq:49}
\text{tr}\left[(\bm{Z} - \bm{AF}^\top)(\bm{Z} - \bm{AF}^\top)^\top\right].
\end{equation}
Minimizing the objective function with respect to $\bm{F}^\top$, we have $\bm{F}^\top = (\bm{A}^\top\bm{A})^{-1}\bm{A}^\top\bm{Z} = \bm{A}^\top\bm{Z}/p$ using~\eqref{eq:50}. Substituting in~\eqref{eq:49}, we obtain the objective function
\begin{align*}
&\text{tr}\left[(\bm{Z} - \bm{AA}^\top\bm{Z}/p)(\bm{Z} - \bm{AA}^\top\bm{Z}/p)^\top\right] \\
&= 
\text{tr}\left(\bm{ZZ}^\top - \bm{ZZ}^\top\bm{AA}^\top/p - \bm{ZZ}^\top\bm{AA}^\top/p + \bm{AA}^\top\bm{ZZ}^\top\bm{AA}^\top/p^2\right)\\
&= \text{tr}(\bm{ZZ}^\top) - \text{tr}(\bm{A}^\top\bm{ZZ}^\top\bm{A})/p,
\end{align*}
where the last equality uses~\eqref{eq:50} and that $\text{tr}(\bm{ZZ}^\top\bm{AA}^\top) = \text{tr}(\bm{A}^\top\bm{ZZ}^\top\bm{A})$. Thus, minimizing the objective function is equivalent to maximizing $\text{tr}(\bm{A}^\top\bm{ZZ}^\top\bm{A})/p$. The estimator for $\bm{A}$ is obtained by finding the first $r$ eigenvectors corresponding to the $r$ largest eigenvalues of the matrix $\bm{ZZ}^\top$ in descending order, where
\begin{equation*}
\bm{Z}\bm{Z}^\top = \sum_{i=1}^n\bm{\eta}_i\bm{\eta}_i^\top = \sum_{i=1}^n(\bm{Y}_i-\bm{\Phi}\bm{c}_i)(\bm{Y}_i-\bm{\Phi}\bm{c}_i)^\top.
\end{equation*} 
Therefore, knowing $\bm{c}_i$, we solve for $\widehat{\bm{A}}$ using
\begin{equation}\label{eq:9}
\left[\frac{1}{np}\sum_{i=1}^n(\bm{Y}_i-\bm{\Phi} \bm{c}_i)(\bm{Y}_i-\bm{\Phi} \bm{c}_i)^\top\right]\widehat{\bm{A}} = \widehat{\bm{A}}\bm{V}_{np},
\end{equation}
where $\bm{V}_{np}$ is a $r\times r$ diagonal matrix containing the $r$ eigenvalues of the matrix in the square brackets in decreasing order. The coefficient $\frac{1}{np}$ is used for scaling.

\begin{remark}
The number of factors $r$ is assumed to be known in this paper. In practice, $r$ is selected based on some criteria regarding the eigenvalues. There have been many studies on this topic. Examples include \cite{BN02}, where two model selection criteria functions were proposed; and \cite{O10}, where the number of factors was estimated using differenced eigenvalues; and \cite{AH13}, where this number was selected based on the ratio of two adjacent eigenvalues.
\end{remark}

It can be seen that $\bm{A}$ is needed to find $\widehat{\bm{c}}_i$, and in turn $\bm{c}_i$ is needed to find $\widehat{\bm{A}}$. The final estimator $(\widehat{\bm{c}}_i, \widehat{\bm{A}})$ is the solution of the set of equations
\begin{align}\label{eq:1}
\begin{cases}
\widehat{\bm{c}_i} = \left(\bm{\Phi}^\top \bm{M}_{\widehat{\bm{A}}}\bm{\Phi} + \alpha \bm{R} \right)^{-1} \bm{\Phi}^\top \bm{M}_{\widehat{\bm{A}}}\bm{Y}_i, \qquad i = 1,\dots, n\\
\left[\frac{1}{np}\sum_{i=1}^n\left(\bm{Y}_i-\bm{\Phi} \widehat{\bm{c}}_i\right)\left(\bm{Y}_i-\bm{\Phi} \widehat{\bm{c}}_i\right)^\top\right]\widehat{\bm{A}} = \widehat{\bm{A}}\bm{V}_{np}.\end{cases}
\end{align}
Since there is no closed-form expression of $\widehat{\bm{A}}$ and $\widehat{\bm{c}}_i$, we propose using numerical iterations to find the estimates. The details of these iterations are as follows:

\hfill

\begin{algorithm}[H]
\SetAlgoLined
\begin{enumerate}
\item
Denote the initial value as $\widehat{\bm{A}}^{(0)}$. Using~\eqref{eq:1}, we obtain $\widehat{\bm{c}}_i^{(0)} =  \left(\bm{\Phi}^\top \bm{M}_{{\widehat{\bm{A}}}^{(0)}}\bm{\Phi} + \alpha \bm{R} \right)^{-1} \bm{\Phi}^\top \bm{M}_{{\widehat{\bm{A}}}^{(0)}}\bm{Y}_i$.
\item
With $\widehat{\bm{c}}_i^{(t)}$, we substitute into the second equation of~\eqref{eq:1}, to obtain $\widehat{\bm{A}}^{(t+1)} = (\widehat{\bm{a}}_1^{(t+1)}, \dots, \widehat{\bm{a}}_r^{(t+1)})^\top$, where $\widehat{\bm{a}}_j^{(t+1)}$ is the eigenvector of the matrix $\frac{1}{np}\sum_{i=1}^n(\bm{Y}_i-\bm{\Phi} \widehat{\bm{c}}_i^{(t+1)})(\bm{Y}_i-\bm{\Phi} \widehat{\bm{c}}_i^{(t+1)})^\top$ corresponding to its $j$th largest eigenvalue.
\item
With $\widehat{\bm{A}}^{(t+1)}$, we obtain $\widehat{\bm{c}}_i^{(t+1)} =  \left(\bm{\Phi}^\top \bm{M}_{{\widehat{\bm{A}}}^{(t+1)}}\bm{\Phi} + \alpha \bm{R}\right)^{-1} \bm{\Phi}^\top \bm{M}_{{\widehat{\bm{A}}}^{(t+1)}}\bm{Y}_i$ using~\eqref{eq:1}
\item
We then repeat steps 2 and 3 until $\|\widehat{\bm{c}}_i^{(t+1)} - \widehat{\bm{c}}_i^{(t)}\| < \delta$, where $\delta$ is a small positive constant.
\end{enumerate}
 \caption{Iterations for estimating FASM}\label{ag}
\end{algorithm}

\begin{remark}
In this paper, we use $\widehat{\bm{A}}^{(0)} = \bm{0}$. This means we start by ignoring the factor model component so the initial value for the smoothing coefficient $\widehat{\bm{c}}_i^{(0)} =  \left(\bm{\Phi}^\top \bm{\Phi} + \alpha \bm{R} \right)^{-1} \bm{\Phi}^\top\bm{Y}_i$, which is simply the ridge estimator. The convergence of Newton's numeric iteration requires the convergence of this estimator, which in turn requires the factor model component $\eta_{ij}$ to have an expectation of zero. The stopping criterion only focuses on $\widehat{\bm{c}}_i$ because we are interested in estimating $\eta_{ij}$ as a whole.
\end{remark}

\begin{remark}
Common methods for selecting the shrinkage parameter $\alpha$ include the Akaike’s Information Criterion \citep[AIC][]{A74}, the Bayesian Information Criterion \citep[BIC][]{S78}, and cross-validation. In this paper, we use the mean generalized cross-validation (mGCV) method \citep{GHW79}. We define, at step $t$,
\begin{equation}\label{eq:53}
\text{mGCV}^{(t)} = \frac{1}{n}\sum_{i=1}^n \frac{\text{pSSE}_i^{(t)}}{[p-df^{(t)}(\alpha)]^2},
\end{equation} 
where $\text{SSE}_i^{(t)}$ is the sum of squares residual for the $i$th object at step $t$ and $df^{(t)}(\alpha)$ is the equivalent degrees of freedom measure, which can be calculated as
\begin{equation}\label{eq:8}
\textnormal{df}^{(t)}(\alpha) = \textnormal{trace}\left[\bm{\Phi}\left(\bm{\Phi}^\top \bm{M}_{{\widehat{\bm{A}}}^{(t)}}\bm{\Phi} + \alpha \bm{R}\right)^{-1}\bm{\Phi}^\top \bm{M}_{{\widehat{\bm{A}}}^{(t)}}\right].
\end{equation}
At each step of the iteration, the tuning parameter $\alpha$ is chosen by minimizing the $mGCV^{(t)}$.
\end{remark}

\begin{remark}
Algorithm~\ref{ag} is an iteration procedure in which ridge regression and PCA are iterated. 
The convergence of this iterative algorithm is studied in \cite{JY20}. For instance, Theorem 2 of \cite{JY20} provides some sufficient conditions under which the recursive algorithm converges to the true value or some other values. In particular, when the regressors are independent of the common factors, or the factors involved in regressors are weaker than the common factors, this algorithm will converge to the true parameter.  
\end{remark}

After we obtain the estimates $\widehat{\bm{A}}$ and $\widehat{\bm{c}}_i$,  the estimated coefficient matrix $\widehat{\bm{C}}$ is constructed as $
\widehat{\bm{C}} = \left(\widehat{\bm{c}}_1, \dots, \widehat{\bm{c}}_n\right)$, and
the estimated factor can be obtained by 
\begin{equation*}
\widehat{\bm{F}}^\top = \widehat{\bm{A}}^\top (\bm{Y}-\bm{\Phi}\widehat{\bm{C}}).
\end{equation*}
Finally, the functional component can be estimated by $\widehat{\mathcal{X}}_i(u) = \widehat{\bm{c}}_i^\top\bm{\Phi}(u)$, where $\bm{\Phi}(u) $ is defined in~\eqref{eq:2}.

\begin{remark}
Although we have imposed the constraint in \eqref{eq:50} and the identification condition~\ref{id}, $\bm{A}$ and $\bm{f}_i$ are not uniquely determined, since the model~\eqref{eq:14} is unchanged if we replace $\bm{A}$ and $\bm{f}_i$ with $\bm{AU}$ and $\bm{U}^\top\bm{f}_i$ for any orthogonal $r\times r$ matrix $\bm{U}$. However, the linear space spanned by the columns of $\bm{A}$ is uniquely defined. Although we are not able to estimate $\bm{A}$, we can still estimate a rotation of $\bm{A}$, which spans the same space as $\bm{A}$ does. The matrix $\bm{M}_{\bm{A}}$ defined in~\eqref{eq:51} is a projecting matrix onto the orthogonal complement of the linear space spanned by the columns of $\bm{A}$. It is shown in the next section that the estimator $\bm{M}_{\widehat{\bm{A}}}$ for $\bm{M}_{\bm{A}}$ is consistent.
\end{remark}

\section{Asymptotic theory}\label{se:4}

In this section, we study the asymptotic properties of the coefficient estimator $\widehat{\bm{c}}_i$ with growing sample size and dimension. We state the assumptions in Section~\ref{se:4.1} and provide the asymptotic results of $\widehat{\bm{c}}_i$ in Section~\ref{se:4.2}. 

\subsection{Assumptions}\label{se:4.1}
We use $(\bm{c}_i^0, \bm{A}^0)$ to denote the true parameters. In this paper, the norm of a vector or matrix $\bm{U}$ is defined as the Frobenius norm; that is,  $\|\bm{U}\| = [\text{tr}(\bm{U}^\top \bm{U})]^{1/2}$. We introduce the matrix
\begin{equation}\label{eq:15}
\bm{D}_i(\bm{A}) \equiv \frac{1}{p}\bm{\Phi}^\top \bm{M}_{\bm{A}}\bm{\Phi} - \frac{1}{p}\bm{\Phi}^\top \bm{M}_{\bm{A}}\bm{\Phi}\bm{f}_i^\top \left(\frac{\bm{F}^\top\bm{F}}{n}\right)^\top\bm{f}_i.
\end{equation}
This matrix plays an important role in this article. It is used in the proof of the consistency of $\widehat{\bm{c}}_i$, as can be found in Appendix A. The identifying condition for $\bm{c}^0_i$ is that $\bm{D}_i(\bm{A})$ is positive definite for all $i$, which is stated in Assumption~\ref{as:3}.

First, we state the assumptions.

\begin{assumption}\label{as:1}
\begin{equation*}
\sup_u |\phi_k(u)| = O(1), \quad k = 1,\dots, K.
\end{equation*}
\end{assumption}
The above assumption declares that the basis functions are bounded in the norm. This is quite natural as some of the most commonly used basis functions are bounded; for instance, the Fourier basis, B-spline basis, and wavelet basis functions \citep{RS05}.

\begin{assumption}\label{as:2}
\begin{equation*}
\|\bm{c}_i^0\| = O(1), \ \text{for all }i.
\end{equation*}
\end{assumption}
Above, we assume the smoothing coefficients $\bm{c}^0_i$ are bounded uniformly for all $i$. This assumption is introduced to ensure the uniform consistency of the estimated coefficients of $\widehat{\bm{c}}^0_i$.
\begin{assumption}\label{as:3}
Let $\mathcal{A} = \{\bm{A}: \bm{A}^\top\bm{A}/p = \bm{I}, \text{ and } \bm{A} \text{ independent of }\bm{\Phi}\}$. We assume
\begin{equation*}
\inf_{\bm{A}\in \mathcal{A}} \bm{D}_i(\bm{A}) > 0.
\end{equation*}
\end{assumption}

This assumption is the identification condition for $\bm{c}^0_i$. The usual assumption for the least-squares estimator only contains the first term on the right-hand side of~\eqref{eq:15}. The second term on the right-hand side of~\eqref{eq:15} arises because of the unobservable matrices $\bm{F}$ and $\bm{A}$.

\begin{assumption}\label{as:4}
For some constant $M>0$,
$\mathbb{E}\|\bm{a}^0_j\|^4 \le M$, $j = 1,\dots, p$ and $\mathbb{E}\|\bm{f}_i\|^4 \le M$.
\end{assumption}

\begin{assumption}\label{as:5}
For some constant $M>0$, the error terms $\epsilon_{ji}, j = 1,\dots, p, i = 1,\dots, n$ are i.i.d. in both directions, with 
$\mathbb{E}(\epsilon_{ji}) = 0$, $\text{Var}(\epsilon_{ji}) = \sigma^2$, and $\mathbb{E}|\epsilon_{ji}|^8 \le M$.
\end{assumption}

\begin{assumption}\label{as:6}
$\epsilon_{ji}$ is independent of $\phi_{s}$, $\bm{f}_t$, and $\bm{a}^0_s$ for all $j,i,s,t$.
\end{assumption}
We require that the errors are independent in themselves and also of the functional term $\phi(u)$ and factor model terms $\bm{f}_i$ and $\bm{a}^0_j$. To not mask the main contribution of our method, we use a simplified setting on the error terms to exclude endogeneity. Nevertheless, with simple but tedious modifications, Assumption~\ref{as:5} can be relaxed, and our model can be extended to more complicated settings where correlations between the error term and the factor model term are allowed. 

\begin{assumption}\label{as:7}
The tuning parameter $\alpha$ satisfies $\alpha = o(p)$.
\end{assumption}
This is conventionally assumed in ridge regression \citep[see, e.g.,][]{KF00} and assures that the estimator's asymptotic bias is zero.

Before stating the next assumption, we introduce some notations. Let $\bm{\omega}_j, \ j =  1, \dots, p$ denote the $j$th column of the $K\times p$ matrix $\bm{\Phi}^\top\bm{M}_{{\bm{A}}^0}$, and let $\psi_{ik}$ denote the $(i,k)$th element of the matrix $\bm{M}_{\bm{F}}$, where
\begin{equation}\label{eq:54}
\bm{M}_{\bm{F}} \equiv \bm{I}_n - \bm{F}\left(\bm{F}^\top\bm{F}\right)\bm{F}^\top.
\end{equation}

Then, for any vector  $\bm{b} = (b_1, \dots, b_n)^\top$, we can write
\begin{equation}\label{eq:44}
\frac{1}{\sqrt{np}}\bm{\Phi}^\top \bm{M}_{\bm{A}^0}\bm{E}\bm{M}_{\bm{F}} \bm{b} = \frac{1}{\sqrt{np}}\sum_i^n\sum_j^p\bm{\omega}_j\epsilon_{ji}\sum_k^n\psi_{ik}b_k \equiv \frac{1}{\sqrt{np}}\sum_i^n\sum_j^p\bm{x}_{ij}.
\end{equation}
In~\eqref{eq:44}, for notation simplicity, we define $\bm{x}_{ij}$ as $\bm{\omega}_j\epsilon_{ji}\sum_k^n\psi_{ik}b_k$. The matrix $\bm{\Phi}^\top \bm{M}_{\bm{A}^0}\bm{E}\bm{M}_{\bm{F}}$ is of interest because it is the main component that contributes to the asymptotic distribution of the estimators, as shall be seen in the next section. 

Let
\begin{equation}\label{eq:48}
\bm{L}_{np} \equiv \frac{\sigma^2}{np}\sum_i^n\sum_j^p \bm{\omega}_j^\top \bm{\omega}_j\left(\sum_k^n \psi_{ik}b_k\right)^2.
\end{equation}
We make the following assumption.

\begin{assumption}\label{as:8}
We assume there exists a $K\times K$ matrix $\bm{L}$ such that 
\begin{equation}\label{eq:45}
\bm{L} \equiv \lim_{n,p\rightarrow \infty}\bm{L}_{np},
\end{equation}
where $\bm{L}_{np}$ is defined in~\eqref{eq:48}. Let $\nu^2$ be the smallest eigenvalue of the matrix $\bm{L}$ defined in~\eqref{eq:45}, then assume that $\nu^2 >0$, and that, for all $\varepsilon >0$,
\begin{equation}\label{eq:56}
\lim _{n, p \rightarrow \infty} \frac{1}{np \nu^{2}}\sum_{i=1}^{n}\sum_{j=1}^p \mathbb{E}\left[\left\|\bm{x}_{ij}\right\|^{2} \mathbb{1}\left(\left\|\bm{x}_{ij}\right\|^{2} \geq \varepsilon np \nu^{2}\right)\right]=0.
\end{equation}
\end{assumption}
This assumption is the multivariate Lindeberg condition, which is needed in constructing the central limit theorem in the next section. This is by no means a strong condition; for instance, when the factor model component is ignored, $\bm{\omega}_j$ is simply $\bm{\phi}_j$, and $\bm{x}_{ij} = \bm{\phi}_jb_i\epsilon_{ji}$. Since we assume $\bm{\phi}_j = O(1)$ in Assumption~\ref{as:1}, the Lindeberg condition in \eqref{eq:56} is met.

\subsection{Asymptotic properties}\label{se:4.2}
As we have mentioned previously, the identification problem of the latent factor implies that we actually use the estimator $\widehat{\bm{A}}$ to estimate a rotation of $\bm{A}^0$. Based on the objective function~\eqref{eq:22} in Section~\ref{se:3}, we use a center-adjusted objective function, defined as below.
\begin{equation}\label{eq:23}
S_{np}(\bm{c}_i, \bm{A}) = 
\frac{1}{np}\sum_{i=1}^n \left[(\bm{Y}_i - \bm{\Phi}\bm{c}_i)^\top  \bm{M}_{\bm{A}}(\bm{Y}_i - \bm{\Phi}\bm{c}_i ) + \alpha \bm{c}_i^\top \bm{R}\bm{c}_i \right] - \frac{1}{np}\sum_{i=1}^n\epsilon_i^\top \bm{M}_{\bm{A}^0}\epsilon_i,
\end{equation}
where $\bm{M}_{\bm{A}}$ is defined in \eqref{eq:51}, satisfying $\bm{A}^\top \bm{A}/p = \bm{I}_r.$ The second term on the right-hand side of~\eqref{eq:23} does not contain the unknown $\bm{A}$ and $\bm{c}_i$, so the inclusion of this term does not affect the optimization result. This term is only used for center adjusting, so that the resulting objective function has an expectation zero. We estimate $\bm{c}_i^0$ and $\bm{A}^0$ by
\begin{equation}\label{eq:36}
(\widehat{\bm{c}}_i ,\widehat{\bm{A}}) = \argmin_{\bm{c}_i, \bm{A}} S_{np} (\bm{c}_i, \bm{A}).
\end{equation}

In the following, we establish the asymptotic properties for the estimated coefficient matrix $\widehat{\bm{C}}$. In Theorem~\ref{th:1}, the consistency of the matrix $\widehat{\bm{C}}$ is proved. In Theorem~\ref{th:2}, we show the rate of convergence of $\widehat{\bm{C}}$. Theorem~\ref{th:3} provides the asymptotic distribution of $\widehat{\bm{C}}$.

Let $\bm{P}_{\bm{U}} = \bm{U}(\bm{U}^\top \bm{U})^{-1}\bm{U}^\top$ for a matrix $\bm{U}$. 
\begin{theorem}\label{th:1}
Under Assumptions~\ref{as:1} - \ref{as:6}, as $n, p \rightarrow \infty$, we have the following statements
\begin{enumerate}
\item[(i)]
$\frac{1}{\sqrt{n}}\left\|\bm{C}-\widehat{\bm{C}}\right\|  \overset{p}{\to} 0.$
\item[(ii)]
$\left\|\bm{P}_{\widehat{\bm{A}}}-\bm{P}_{\bm{A}^0}\right\|  \overset{p}{\to} 0$.
\end{enumerate}
\end{theorem}
We start by proving consistency for the vector $\widehat{\bm{c}}_i$. This consistency is uniform for all $i= 1,\dots, n$. Therefore, we could combine $\bm{\bm{c}}_i$ for all $i = 1,\dots, n$, and have the result for the coefficient matrix $\widehat{\bm{C}}$ in $(i)$. The matrix $\widehat{\bm{C}}$ is of dimension $K\times n$, where $K$ is fixed and the sample size $n$ goes to infinity, so there is a $\frac{1}{\sqrt{n}}$ scale in the result of $(i)$. In the second part of the theorem, note that $\bm{P}_{\bm{A}} = \bm{I}_p - \bm{M}_{\bm{A}}$, where $\bm{M}_{\bm{A}}$ is the projection matrix onto the orthogonal complement of the linear space spanned by the columns of $\bm{A}$. Thus, $\bm{P}_{\widehat{\bm{A}}}$ and $\bm{P}_{\bm{A}^0}$ represent the spaces spanned by $\widehat{\bm{A}}$ and $\bm{A}^0$, and we show that they are asymptotically the same in $(ii)$.

Next, we obtain the rate of convergence. 

\begin{theorem}\label{th:2}
Under Assumptions~\ref{as:1} - \ref{as:6}, if $p/n\rightarrow \rho >0$,
\begin{equation*}
\left\|\frac{\left(\bm{C}^0-\widehat{\bm{C}}\right)}{\sqrt{n}}\bm{M}_{\bm{F}}\right\| = O_p\left(\frac{1}{\sqrt{p}}\right),
\end{equation*}
where $\bm{M}_{\bm{F}}$ is defined in~\eqref{eq:54}.
\end{theorem}

We study the case when the dimension $p$ and the sample size $n$ are comparable. We achieve rate $\sqrt{p}$ convergence, considering $\frac{\left\|\bm{C}^0-\widehat{\bm{C}}\right\|}{\sqrt{n}} $ on the whole. It is expected that the rate of convergence for smoothing models depends on the number of discrete points $p$ observed on each curve.

\begin{remark}
The asymptotic result in Theorem~\ref{th:2} contains a projection matrix $\bm{M}_{\bm{F}}$. This matrix projects $\bm{C}^0-\widehat{\bm{C}}$ onto the space orthogonal to the factor matrix $\bm{F}$. This theorem shows the interplay between $\bm{C}^0$ and $\bm{F}$. When $\bm{C}^0$ and $\bm{F}$ are orthogonal, $(\bm{C}^0-\widehat{\bm{C}})\bm{M}_{\bm{F}} = \bm{C}^0-\widehat{\bm{C}}$, and we obtain the rate of convergence of $\bm{C}^0-\widehat{\bm{C}}$. When $\bm{C}^0$ and $\bm{F}$ are not orthogonal, the inference on $\bm{C}^0$ will be affected by the existence of the factor model component.
\end{remark}

We further begin to establish the limiting distribution. It is shown in Appendix A that
\begin{equation*}
\left\|\sqrt{p}\frac{\left(\bm{C}^0-\widehat{\bm{C}}\right)}{\sqrt{n}}\bm{M}_{\bm{F}}\right\| =  \left\|\left(\frac{1}{p}\Phi^\top \bm{M}_{\bm{A}^0}\bm{\Phi}\right)^{-1}\frac{1}{\sqrt{np}}\bm{\Phi}^\top \bm{M}_{\bm{A}^0}\bm{E}\bm{M}_{\bm{F}} \right\| + o_p(1).
\end{equation*}
The limiting distribution is constructed based on the first term on the right-hand side. Let $\bm{\omega}_j$ denote the $j$th column of the $K \times p$ matrix $\bm{\Phi}^\top \bm{M}_{\bm{A}^0}$. We then have the following lemma.
\begin{lemma}\label{le:4}
Under Assumptions~\ref{as:1} -~\ref{as:7}, for any vector $\bm{b} = (b_1, \dots,  b_n)^\top$,
\begin{equation*}
\frac{1}{\sqrt{np}}\bm{\Phi}^\top \bm{M}_{\bm{A}^0}\bm{E}\bm{M}_{\bm{F}} \bm{b} \overset{d}{\to} \mathcal{N}(\bm{0},\bm{L}),
\end{equation*}
where $\bm{L}$ is defined in~\eqref{eq:45}.
\end{lemma}
This lemma paves the way for the next theorem on asymptotic normality.

\begin{theorem}\label{th:3}
Under Assumptions~\ref{as:1} -~\ref{as:7}, if $p/n\rightarrow \rho >0$, we have for any vector $\bm{b} \in \mathbb{R}^n$
\begin{equation*}
\sqrt{p}\left(\frac{\bm{C}^0-\widehat{\bm{C}}}{\sqrt{n}}\right)\bm{M}_{\bm{F}}\bm{b} \overset{d}{\to} \mathcal{N}\left(\bm{0}, \bm{Q}\left(\bm{A}^0\right)^{-1}\bm{L}\bm{Q}\left(\bm{A}^0\right)^{-1}\right),
\end{equation*}
where $\bm{M}_{\bm{F}}$ is defined in Theorem~\ref{th:2}, $\bm{L}$ is defined in~\eqref{eq:45}, and 
\begin{equation*}
\bm{Q}\left(\bm{A}^0\right) \equiv \frac{1}{p}\bm{\Phi}^\top \bm{M}_{\bm{A}^0} \bm{\Phi}.
\end{equation*}
\end{theorem}
The vector $\bm{b}$ comes from the same vector in Lemma~\ref{le:4}. The asymptotic bias is zero since we assume no serial or cross-sectional correlation in the error terms. This is a simplified setting, which can be extended to allow for weak correlations in errors in both directions. In that case, the asymptotic distribution will include a non-zero bias term.

\begin{remark}
Theorem~\ref{th:3} shows that the asymptotic distribution of the coefficient matrix $\widehat{\bm{C}}$ relies on the unobserved factor loading matrix $\bm{A}^0$. Although we are unable to consistently estimate $\bm{A}^0$ with $\widehat{\bm{A}}$, what we need is in fact the projected matrix $\bm{M}_{\bm{A}^0}$, which can be estimated by $\bm{M}_{\widehat{\bm{A}}}$.  We are able to find estimators for $\bm{Q}$ and $\bm{L}$ based on $\bm{M}_{\widehat{\bm{A}}}$,
\begin{align*}
\widehat{\bm{Q}} &= \frac{1}{p}\bm{\Phi}^\top \bm{M}_{\widehat{\bm{A}}} \bm{\Phi}\\
\widehat{\bm{L}} &= \frac{\sigma^2}{np}\sum_i^n\sum_j^p \widehat{\bm{\omega}}_j^\top\widehat{\bm{\omega}}_j \left(\sum_k^n \widehat{\psi}_{ik}b_k\right)^2,
\end{align*}
where $\widehat{\bm{\omega}}_j$ is the $j$th column of the $K \times p$ matrix $\bm{\Phi}^\top \bm{M}_{\widehat{\bm{A}}}$ and $\widehat{\psi}_{ik}$ is the $(i,k)$th element in the matrix $\bm{M}_{\widehat{\bm{F}}}$.
\end{remark}

\section{Extending the model to nonparametric smoothing}\label{se:nonparam}

In the FASM, we use basis smoothing where we assume the basis functions $\phi_k(u)$ are known, and we show the asymptotic properties in the previous section. However, in practice, the basis functions are usually unknown, and nonparametric smoothing techniques are frequently used. In this section, we extend the proposed model to a smoothing spline.

In spline smoothing, a spline basis is used to model functions. We consider smoothing splines, where regularized regression is performed, and the knots are placed on all the observed discrete points. The most commonly considered basis is the cubic smoothing splines, where the order is 4. The number of basis functions equals the number of interior knots plus the order. Thus, we use $p+2$ spline basis functions. Denote the $k$th basis function as $\psi_k(u)$, $k = 1,\dots, p+2$. Let the $p\times (p+2)$ matrix $\bm{\Psi}$ denote the basis matrix, where the $(j,k)$th element is $\psi_k(u_j)$. The objective function can be written as
\begin{equation*}
 \text{SSR}(\bm{w}_i, \bm{A}, \bm{f})  = \sum_{i=1}^n \left[(\bm{Y}_i - \bm{\Psi}\bm{w}_i - \bm{A}\bm{f}_i)^\top  (\bm{Y}_i - \bm{\Psi}\bm{w}_i - \bm{A}\bm{f}_i ) + \alpha \bm{w}_i^\top \bm{U}\bm{w}_i \right],
\end{equation*}
where $\bm{w}_i$ is the vector of smoothing coefficient and the matrix $\bm{U} = \int_{\mathcal{I}}D^2\bm{\Psi}(s)D^2\bm{\Psi}^\top(s)ds$ is similarly defined as the matrix $\bm{R}$ in~\eqref{eq:55}.
The estimators are the solution of the equation system
\begin{align*}
\begin{cases}
\widehat{\bm{w}}_i = \left(\bm{\Psi}^\top \bm{M}_{\widehat{\bm{A}}}\bm{\Psi} + \alpha \bm{U}^\top \right)^{-1} \bm{\Psi}^\top \bm{M}_{\widehat{\bm{A}}}\bm{Y}_i, \quad i = 1,\dots, n\\
\left[\frac{1}{np}\sum_{i=1}^n(\bm{Y}_i-\bm{\Psi} \widehat{\bm{c}}_i)(\bm{Y}_i-\bm{\Psi} \widehat{\bm{w}}_i)^\top\right]\widehat{\bm{A}} = \widehat{\bm{A}}\bm{V}_{np},\end{cases}
\end{align*}
which can be solved by iteration approach in Section~\ref{se:3.2}. The function estimator is $\widehat{\mathcal{X}}_i(u) = \widehat{\bm{w}}^\top_i\bm{\Psi}(u)$, where $\bm{\Psi}(u)$ is a vector containing all the $p+2$ basis functions $\psi_k(u)$.

It could be seen that the model is almost the same as the model proposed in Section~\ref{se:3}. The difference lies in the dimensions of the matrices $\bm{\Psi}$ and $\bm{\Phi}$. In the parametric model, we assume the basis functions are known, and the number of basis functions is fixed, so $\bm{\Phi}$ is a $p\times K$ matrix. While in smoothing spline modeling, the matrix $\bm{\Psi}$ is of dimension $p\times (p+2)$. The number of basis functions $p+2$ goes to infinity. This does not render our model estimation infeasible, however, because we have included a  penalty term.

\section{Statistical inference on covariance matrix estimation}\label{se:5}

Having presented the model estimation approach and the estimators' asymptotic properties, we now consider statistical inference with FASM. Our model serves as a dimension reduction technique and avoids the curse of dimensionality problem, rendering making inferences from the model convenient.

Covariance estimation is fundamental in both FDA and high-dimensional data analysis. In these areas, data are of high dimensions, which brings many challenges. In the FDA, the number of discrete points on each curve is often larger than the number of curves. Similarly, the dimension $p$ of high-dimensional data is typical of the same order or larger than the sample size $n$. In this case, the traditional sample covariance estimator no longer works. Dimension reduction by imposing some structure on the data is one of the main ways to solve this problem \citep[see, e.g.,][]{ W03, BL08, FFL08}. By reducing the data dimension with a smoothing model and a factor model in FASM, we propose an alternative covariance matrix estimator.

We consider the covariance matrix of the observed high-dimensional data $\bm{Y}_i$. Let
\begin{equation*}
\bm{\Sigma}_Y \equiv \text{cov}(\bm{Y}).
\end{equation*}
Based on the FASM where
\begin{equation*}
\bm{Y}_i = \bm{\Phi}\bm{c}_i + \bm{A}\bm{f}_i + \bm{\epsilon}_i,
\end{equation*}
we obtain
\begin{equation}\label{eq:12}
\bm{\Sigma}_{\bm{Y}} = \bm{\Phi}\bm{\Sigma}_{\bm{c}}\bm{\Phi}^\top + \bm{A}\bm{\Sigma}_{\bm{f}}\bm{A}^\top + \bm{\Sigma}_{\bm{\epsilon}},
\end{equation}
where $\bm{\Sigma}_{\bm{c}}$ and $\bm{\Sigma}_{\bm{f}}$ are covariance matrices of the vectors $\bm{c}$ and $\bm{F}$ respectively and $\bm{\Sigma}_\epsilon$ denotes the error variance structure and is a diagonal matrix under Assumption~\ref{as:4}. Based on the above equation, we have an estimator
\begin{equation}\label{eq:11}
\widehat{\bm{\Sigma}}_{\bm{Y}} = \bm{\Phi}\widehat{\bm{\Sigma}}_{\bm{c}}\bm{\Phi}^\top + \widehat{\bm{A}} \widehat{\bm{\Sigma}}_{\bm{f}} \widehat{\bm{A}}^\top + \widehat{\bm{\Sigma}}_{\bm{\epsilon}},
\end{equation}
where $\widehat{\bm{\Sigma}}_{\bm{c}}$ and $\widehat{\bm{\Sigma}}_{\bm{f}}$ can be calculated by
\begin{align*}
\widehat{\bm{\Sigma}}_{\bm{c}} &= \frac{1}{n-1}\bm{C}\bm{C}^\top - \frac{1}{n(n-1)}\bm{C}\bm{1}\bm{1}^\top\bm{C}^\top\\
\widehat{\bm{\Sigma}}_{\bm{f}} &= \frac{1}{n-1}\bm{F}\bm{F}^\top - \frac{1}{n(n-1)}\bm{F}\bm{1}\bm{1}^\top\bm{F}^\top,
\end{align*}
where $\bm{1}$s are vectors containing ones, the dimensions of which depend on the matrices multiplied before and after the vectors. The diagonal error covariance matrix $\bm{\Sigma}_{\bm{\epsilon}}$ is estimated by 
\begin{equation*}
\widehat{\bm{\Sigma}}_{\bm{\epsilon}} = \text{diag}\left(n^{-1}\widehat{\bm{E}}\widehat{\bm{E}}^\top\right),
\end{equation*}
where $\widehat{\bm{E}}$ is the residual matrix calculated as $\widehat{\bm{E}} = \bm{Y} - \bm{\Phi}\widehat{\bm{C}} - \widehat{\bm{A}}\widehat{\bm{F}}^\top$. 

\begin{remark}
In functional data analysis where the functional signal is the focus, the estimation of the covariance function $\bm{\Phi}\bm{\Sigma}_{\bm{c}}\bm{\Phi}^\top$ is of main interest. In this paper, we study the covariance structure of the mixture of functional data and high-dimensional data. This type of covariance estimator based on factor models has also been used in previous literature. For example, \cite{FFL08} employed a multi-factor model where the factors are assumed observable, while \cite{FLM11} considered an extension to approximate factor models where cross-sectional correlation is allowed in the error terms.
\end{remark}

We compare the finite sample performance using mean squared error (MSE) of the proposed covariance estimator with the ordinary sample covariance estimator. When the factor structure is ignored, the sample covariance estimator is expected to have a larger variance than our estimator. The advantage of the proposed estimator is shown in Section~\ref{se:6.4}.

\section{Simulation studies}\label{se:6}

In this section, we use simulated data to illustrate the superiority of the proposed model. The FASM is compared with the smoothing model in Section~\ref{se:6.1} to~\ref{se:6.3}. In Section~\ref{se:6.4}, we compare the finite sample performance of the covariance matrix estimator introduced in Section~\ref{se:5} with the ordinary sample covariance estimator. In Section~\ref{se:6.5}, we show how the FASM performs when applied to functional data with step jumps.

\subsection{Data generation}\label{se:6.1}

We generate simulated data $Y_{ij}$, where $i = 1,\dots, n$ and $j = 1,\dots, p$ from the following model:
\begin{align*}
Y_{ij} &= \mathcal{X}_i(u_j) + \eta_{ij} + \epsilon_{ji}\\
&= \sum_{k=1}^{13} c_{ik} \phi_k(u_j) + \sum_{k=1}^4\lambda_{ik} F_{kj} + \epsilon_{ji},
\end{align*}
where $\phi_k(u)$ are chosen as B-spline basis functions of order 4 and the smoothing coefficients $c_{ik}$ are generated from $\mathcal{N}(0, 1.5^2)$. The factors $F_{kj}$ follow $\mathcal{N}(0, 0.5^2)$ and the factor loadings $(\lambda_{i1}, \lambda_{i2}, \lambda_{i3}, \lambda_{i4})^\top \sim \mathcal{N}(\bm{\mu}, \bm{\Sigma})$, where $\bm{\Sigma}$ is a 4 by 4 covariance matrix. We set the multivariate mean term $\bm{\mu} = \bm{0}$ and variance $\bm{\Sigma} = \sigma^2 \bm{I}_4$. We adjust the value of $\sigma^2$ to control the signal-to-noise ratio. When $\sigma^2$ is large, the signal-to-noise level is low, and when $\sigma^2$ is small, the signal-to-noise level is high. The random error terms $\epsilon_{ji}$ follow $\mathcal{N}(0, 0.5^2)$.

\subsection{Estimation}\label{se:6.2}

The numeric iteration procedure for finding $(\widehat{\bm{c}}_i, \widehat{\bm{A}}, \widehat{\bm{f}})$ is introduced in Section~\ref{se:3}. We compare the FASM with the smoothing model, where the factor model component is ignored. The smoothing model can be expressed as:
\begin{equation*}
\bm{Y}_{i} = \bm{\Phi}\bm{c}_i + \bm{\epsilon}_{i},
\end{equation*}
where the coefficient estimator is calculated as:
\begin{equation*}
\widehat{\bm{c}_i} = \left(\bm{\Phi}^\top \bm{\Phi} + \alpha \bm{R} \right)^{-1} \bm{\Phi}^\top \bm{Y}_i, \qquad i = 1,\dots, n.
\end{equation*}
The tuning parameter $\alpha$ is also chosen using mGCV defined in~\eqref{eq:53}.

\subsection{Results}\label{se:6.3}

We repeat the simulation setup 100 times and obtain the estimated smooth function $\widehat{\mathcal{X}}_i(u) = \widehat{\bm{c}}_i^\top \bm{\Phi}(u)$.  The averaged mean squared error (aMSE) for function estimation is calculated as
\begin{equation*}
\text{aMSE} = \frac{1}{np}\sum_{i=1}^n\sum_{j=1}^p \left[\mathcal{X}_i(u_j)-\widehat{\mathcal{X}}_i(u_j)\right]^2.
\end{equation*}

The results are reported in Table~\ref{tab:1}. With the same sample size $n$, increasing the number of points $p$ on the curve decreases the estimation error. However, with the same value for $p$, increasing the sample size does not decrease the estimation error. This is consistent with the convergence rate stated in Section~\ref{se:4}, where the estimator converges with the rate related to $p$. When $\sigma$ is large, such that the signal-to-noise ratio is high, the FASM performs better than the smoothing model.

\begin{center}
\tabcolsep 0.43in
\begin{longtable}{@{}llcc@{}}
\caption{The aMSE of the function estimates with different samples sizes and dimensions. The adjustment of $\sigma^2$ value is used to control the signal-to-noise ratio.}\label{tab:1} \\
\toprule
&& \multicolumn{2}{c}{aMSE} \\
Dimension & $\sigma$ & FASM & Smoothing model\\
\midrule 
\endfirsthead
\toprule
& \multicolumn{3}{c}{aMSE} \\
Dimension & Size of $\eta$ & FASM & Smoothing model\\
\hline 
\endhead
\midrule \multicolumn{4}{r}{{Continued on next page}} \\
\endfoot
\endlastfoot
  $n = 20, p = 51$& $\sigma = 0.5$ & $0.1669$&$ \bm{0.1177}$ \\ 
&$\sigma = 0.75$ & $0.1774$&$ \bm{0.1729}$ \\ 
& $\sigma = 1$&$\bm{0.2151}$&$0.2424$\\\midrule
$n = 20, p = 101$ & $\sigma = 0.5$ & $0.0719 $&$\bm{0.0661}$ \\ 
&$\sigma = 0.75$ & $\bm{0.0893}$&$ 0.0979$ \\ 
& $\sigma = 1$&$\bm{0.1163} $&$0.1384$\\\midrule
$n = 50, p = 51$ & $\sigma = 0.5$ & $0.1311$&$ \bm{0.1207}$ \\ 
&$\sigma = 0.75$ & $\bm{0.1522}$&$ 0.1787$ \\ 
& $\sigma = 1$&$\bm{0.1943}$&$ 0.2518$\\\midrule
 $n = 100, p = 101$ & $\sigma = 0.5$ & $\bm{0.0593}$&$ 0.0674$ \\ 
&$\sigma = 0.75$ & $\bm{0.0794}$&$ 0.0989$ \\ 
& $\sigma = 1$&$\bm{0.1051}$&$ 0.1385$\\
\bottomrule
\end{longtable}
\end{center}

\subsection{Covariance matrix estimation}\label{se:6.4}

This section shows the finite sample performance of the covariance estimator defined in~\eqref{eq:11}. We also calculate the regular sample covariance estimator $\widehat{\bm{\Sigma}}^*_{Y}$ using
\begin{equation*}
\widehat{\bm{\Sigma}}_{\bm{Y}}^* = \frac{1}{n-1}\left(\bm{Y}-\overbar{\bm{Y}}\right)\left(\bm{Y}-\overbar{\bm{Y}}\right)^\top,
\end{equation*}
where the $p \times n$ matrix $\overbar{\bm{Y}}$ is the sample mean matrix whose $j$th row elements are $\frac{1}{n}\sum_{i=1}^n Y_{ij}$.

Both estimators are compared with the population covariance matrix, which is calculated using~\eqref{eq:12}. We calculate the estimation errors under the Frobenius norm as
\begin{equation*}
\text{MSE} = \frac{1}{p}\left\|\widehat{\bm{\Sigma}}_{\bm{Y}} - \bm{\Sigma}_{\bm{Y}}\right\|^2.
\end{equation*}
We show the MSE results in Table~\ref{tab:4}. It can be seen that the FASM produces smaller MSE values in almost all cases. 

\begin{center}
\tabcolsep 0.43in
\begin{longtable}{@{}llcc@{}}
\caption{The MSE of the two covariance estimators with different sample sizes and dimensions. The adjustment of $\sigma^2$ value is used to control the signal-to-noise ratio.}\label{tab:5} \\
\toprule
& &\multicolumn{2}{c}{MSE} \\
Dimension &$\sigma$ & FASM & Sample covariance\\
\midrule 
\endfirsthead
\toprule
& \multicolumn{3}{c}{aMSE} \\
Dimension & Size of $\eta$ & FASM & Sample covariance\\
\hline 
\endhead
\midrule \multicolumn{4}{r}{{Continued on next page}} \\
\endfoot
\endlastfoot
$n = 20, p = 51$&$\sigma = 0.5$ &$\bm{0.090}$ & $0.143$ \\ 
&$\sigma = 0.75$ &  $\bm{0.128}$ & $0.198$ \\ 
&$\sigma = 1$ & $\bm{0.218}$ & $0.326$ \\ \hline
$n = 20, p = 101$&$\sigma = 0.5$&  $\bm{0.089}$ & $0.145$ \\ 
&$\sigma = 0.75$ &  $\bm{0.118}$ & $0.197$ \\ 
&$\sigma = 1$ &$\bm{0.211}$ & $0.333$ \\ \hline
$n = 50, p = 51$&$\sigma = 0.5$ &$\bm{0.041}$ & $0.058$ \\ 
&$\sigma = 0.75$ & $\bm{0.059}$ & $0.078$ \\ 
&$\sigma = 1$ & $\bm{0.117}$ & $0.122$ \\  \hline
$n = 100, p = 101$&$\sigma = 0.5$&  $\bm{0.019}$ & $0.027$ \\ 
&$\sigma = 0.75$ &  $\bm{0.031}$ & $0.038$ \\ 
&$\sigma = 1$ &  $0.066$ & $\bm{0.062}$ \\
\bottomrule
\end{longtable}
\end{center}

\subsection{Nonparametric smoothing model}\label{se:6.7}

In this section, we apply the factor-augmented nonparametric smoothing model introduced in Section~\ref{se:nonparam} to simulated data and compare the results with using nonparametric smoothing models without the factor component. 

We generate simulated data $Y_{ij}$, where $i = 1,\dots, n$ and $j = 1,\dots, p$ from the following model:
\begin{align*}
Y_{ij} &= \mathcal{X}_i(u_j) + \eta_{ij} + \epsilon_{ji}\\
&= \sum_{k=1}^{9} c_{ik} \phi_k(u_j) + \sum_{k=1}^4\lambda_{ik} F_{kj} + \epsilon_{ji},
\end{align*}
where $\phi_k(u)$ are Fourier basis functions and the smoothing coefficients $c_{ik}$ are generated from $\mathcal{N}(0, 1.5^2)$. The factors $F_{kj}$ follow $\mathcal{N}(0, 0.5^2)$ and the factor loadings $(\lambda_{i1}, \lambda_{i2}, \lambda_{i3}, \lambda_{i4})^\top \sim \mathcal{N}(\bm{\mu}, \bm{\Sigma})$, where $\bm{\Sigma}$ is a 4 by 4 covariance matrix. The random error terms $\epsilon_{ji}$ follow $\mathcal{N}(0, 0.5^2)$. We set the multivariate mean term $\bm{\mu} = \bm{0}$ and variance $\bm{\Sigma} = \sigma^2 \bm{I}_4$. We adjust the value of $\sigma^2$ to control the signal-to-noise ratio. When $\sigma^2$ is large, the signal-to-noise level is low, and when $\sigma^2$ is small, the signal-to-noise level is high.

\subsubsection*{Smoothing spline}

We use order 4 B-spline basis with knots at every data point. With data of dimension $p$, we use $p+2$ basis functions. The tuning parameter $\lambda$ is selected by the mean generalized cross-validation~\eqref{eq:53} at each iteration step. The covariance estimate is calculated using~\eqref{eq:11}. Table~\ref{tab:4} presents the results.

\begin{center}
\tabcolsep 0.2in
\begin{longtable}{@{}llcccc@{}}
\caption{Using smoothing splines: The aMSE of the estimated functions and the MSE of the two covariance estimators with different sample sizes and dimensions. The adjustment of $\sigma^2$ value is used to control the signal-to-noise ratio.}\label{tab:4} \\
\toprule
& & \multicolumn{2}{c}{aMSE} &  \multicolumn{2}{c}{ MSE} \\
Dimension & $\sigma$  & FASM & Smoothing & FASM & Sample covariance\\
\midrule 
\endfirsthead
\toprule
& & \multicolumn{2}{c}{aMSE} &  \multicolumn{2}{c}{MSE} \\
Dimension & Size of $\eta$ & FASM & Smoothing & FASM & Sample covariance\\
\hline 
\endhead
\midrule \multicolumn{6}{r}{{Continued on next page}} \\
\endfoot
\endlastfoot
$n = 20, p = 51$ &$\sigma=0.5$ & $0.154$ & $\bm{0.142}$ & $\bm{1.535}$ & $1.882$ \\ 
&$\sigma=0.75$ & $\bm{0.212}$ & $0.219$ & $\bm{1.485}$ & $1.931$ \\ 
&$\sigma=1$& $\bm{0.293}$ & $0.317$ & $\bm{1.649}$ & $2.336$ \\ \midrule
$n = 20, p = 101$ &$\sigma=0.5$ & $\bm{0.076}$ & $\bm{0.076}$ & $\bm{1.479}$ & $1.814$ \\ 
&$\sigma=0.75$ & $\bm{0.112}$ & $0.121$ & $\bm{1.593}$ & $2.086$ \\ 
&$\sigma=1$& $\bm{0.159}$ & $0.175$ & $\bm{1.587}$ & $2.268$ \\ \midrule
$n = 50, p = 51$&$\sigma=0.5$ & $\bm{0.138}$ & $0.142$ & $\bm{0.612}$ & $0.713$ \\ 
&$\sigma=0.75$ & $\bm{0.192}$ & $0.213$ & $\bm{0.630}$ & $0.765$ \\ 
&$\sigma=1$& $\bm{0.270}$ & $0.313$ & $\bm{0.690}$ & $0.869$ \\ \midrule
$n = 100, p = 101$&$\sigma=0.5$ & $\bm{0.070}$ & $0.075$ & $\bm{0.298}$ & $0.349$ \\ 
&$\sigma=0.75$ & $\bm{0.104}$ & $0.119$ & $\bm{0.309}$ & $0.377$ \\ 
&$\sigma=1 $& $\bm{0.153}$ & $0.177$ & $\bm{0.352}$ & $0.429$ \\ 
\bottomrule
\end{longtable}
\end{center}

\subsection{Misidentification of the basis function}\label{se:6.5}

We elaborate on the example presented in Section~\ref{se:2.2}. We generate data from 
\begin{equation}\label{eq:57}
Y_{ij} = \sum_{k=1}^7 c_{ik}\phi_k(u_j) + \epsilon_{ji}, \quad i = 1,\dots, n,\  j = 1,\dots, p,
\end{equation}
where $\{\phi_k(u), k=1,\dots,7\}$ are a set of Fourier basis functions. The first Fourier basis function $\phi_1(u)$ is the constant function; the remainder are sine and cosine pairs with integer multiples of the base period. We generate the Fourier functions with doubled frequencies in the second half to simulate the change in the basis functions.  In particular, when $u\in [0, 0.5], \ \phi_k(u) = 2\sin (k\pi u),$ for $k = 2, 4, 6,$ and $\phi_k(u) = 2\cos [(k-1)\pi u],$ for $k = 3, 5, 7$, and when $u\in (0.5, 1], \ \phi_k(u) = 2\sin (2k\pi u),$ for $k = 2, 4, 6,$ and $\phi_k(u) = 2\cos [2(k-1)\pi u],$ for $k = 3, 5, 7$. The coefficients $c_{ik}$ are generated from the normal distribution with mean 0 and variance $0.5^2$. The error terms are also drawn from the normal distribution with mean 0 and variance $0.5^2$. The generated $Y_{ij}$ are shown in Figure~\ref{fig:9a}. It can be seen that the data exhibit more variation in the second half of the interval.

Suppose we were unaware of the change in the frequencies of the basis functions and used the bases in the first half to fit the data on the whole interval. The smoothing model residuals, shown in Figure~\ref{fig:9b}, are large in the second half. When the frequency of the basis functions is misidentified, a smoothing model with the wrong set of bases is inadequate. We conduct principal component analysis on the smoothing residuals; the eigenvalues in descending order are shown in Figure~\ref{fig:10a}. The residuals preserve a spiked structure, where six common factors can explain most of the variation.

We also apply FASM to the same data with the wrong set of basis functions. According to the eigenvalue scree plot, we retain six factors in the model ($r = 6$). The resulting residuals are shown in Figure~\ref{fig:10b}. The large residuals in the second part of Figure~\ref{fig:9} (b) are removed. When the basis functions are misidentified, the FASM serves as a remedy.
 
\begin{figure}[!htbp]
\centering
\subcaptionbox{Spikes of the smoothing residuals\label{fig:10a}}
{\includegraphics[width = 8.65cm]{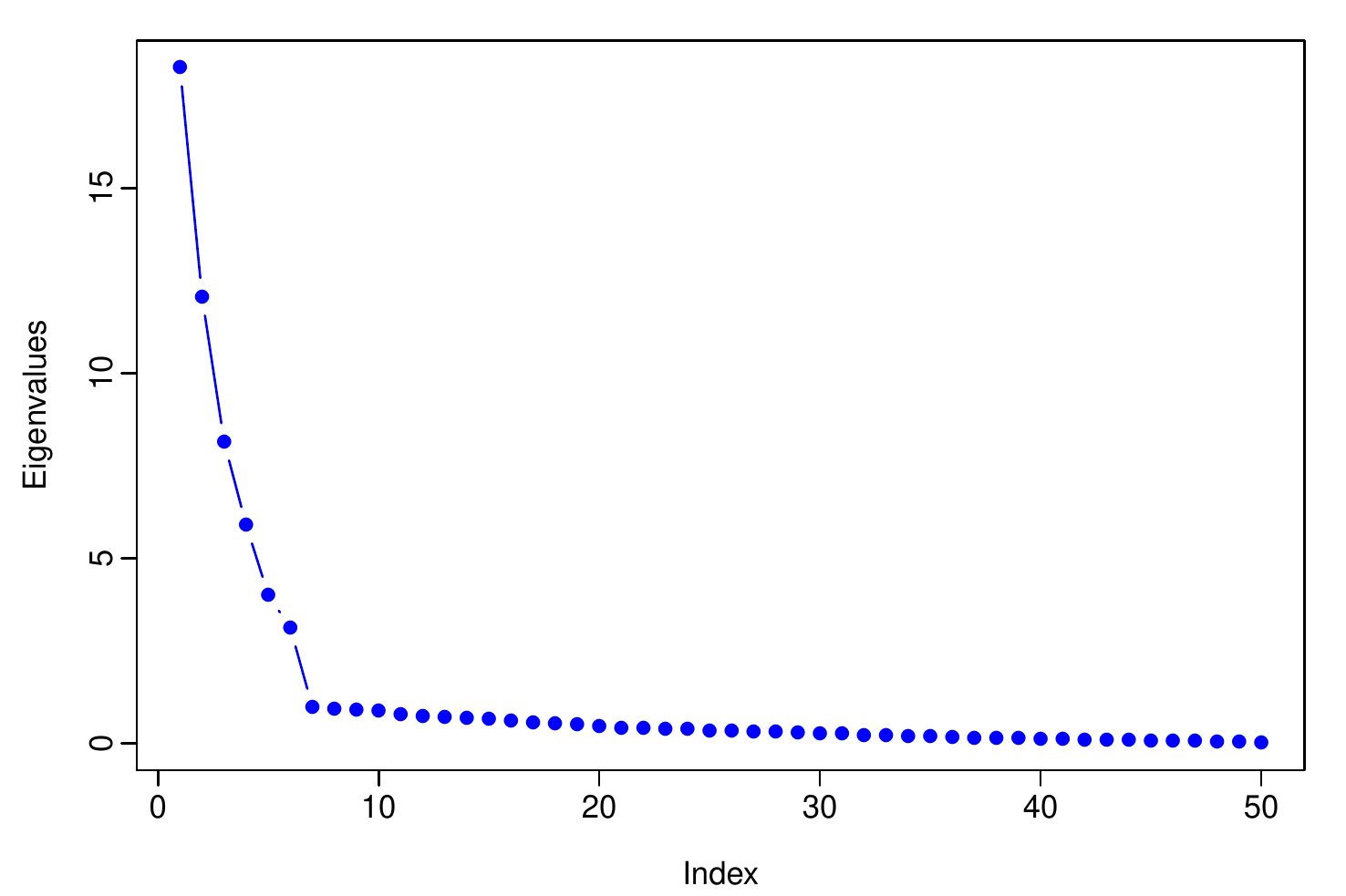}}
\subcaptionbox{Residuals of FASM\label{fig:10b}}
{\includegraphics[width = 8.65cm]{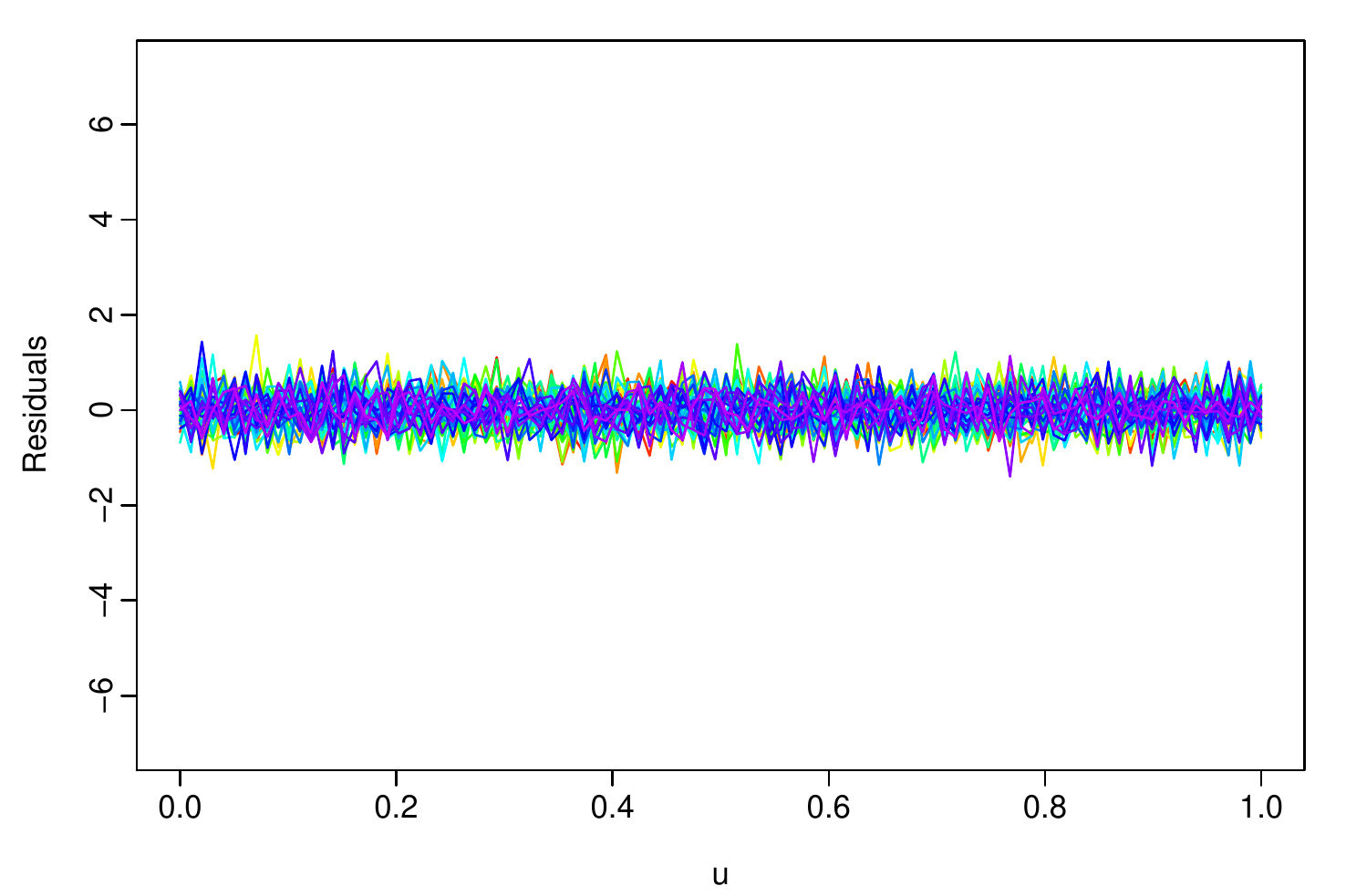}}
\caption{Applying FASM to the data generated by~\eqref{eq:57}.}\label{fig:10}
\end{figure}

\subsection{Functional data with step jumps}\label{se:6.6}

We study the case where the functional data exhibit a dramatic change in the mean level within a small window. 
We generate data from the following model
\begin{equation*}
Y_{ij} = \mu(u_j) + \sum_{k=1}^7 c_{ik}\phi_k(u_j) + \epsilon_{ji}, \quad i = 1,\dots, n,\quad j = 1,\dots, p,
\end{equation*}
where the basis functions $\phi_k(u)$ are order 4 B-spline bases. The coefficients $c_{ik}$ come from $\mathcal{N}(0, 1.5^2)$ and the error terms from $\mathcal{N}(0, 0.5^2)$. The mean function $\mu(u)$ is generated by a linear combination of 25 B-spline basis functions. Figure~\ref{fig:8} shows an example of the mean function- there is a sharp increase in the mean function at around $u = 0.5$.

The change in the mean level happens at $u = 0.5$, and $\delta$ denotes the amount of change. Figure~\ref{fig:1} is generated using $\delta = 2$. Figure~\ref{fig:3} compares the residuals from the smoothing model and the FASM. With the smoothing model, the residuals around the jump are large. In contrast, our model explains the large residuals around the structural break very well. In the aspect of model selection, we consider the trade-off between model fit and model flexibility. We first define a notion of flexibility for a fitted model with the degrees of freedom. We use the same concept as in most textbooks that the degrees of freedom measure the number of parameters estimated from the data required to define the model. The degrees of freedom for the smoothing model are calculated by~\eqref{eq:8} of the last step of convergence. The degrees of freedom for the FASM is 
\begin{equation*}
\text{df} = \text{trace}\left[\bm{\Phi}(\bm{\Phi}^\top \bm{M}_{{\widehat{\bm{A}}}^{(t)}}\bm{\Phi} + \alpha \bm{R})^{-1}\bm{\Phi}^\top \bm{M}_{{\widehat{\bm{A}}}^{(t)}}\right] + r,
\end{equation*}
where $r$ is the number of factors retained in the fitted model, the larger the degrees of freedom, the more flexible the fitted models are. To quantify the model fitting, we use 
\begin{equation*}
\text{RMSE} = \sqrt{\frac{1}{np}\sum_{i=1}^n\sum_{j=1}^p\left(Y_{ij}-\widehat{Y}_{ij}\right)^2},
\end{equation*}
where $\widehat{Y}_{ij} = \sum_{k=1}^K \widehat{c}_{ik}\phi_k(u_j) + \widehat{\eta}_{ij}$. In Table~\ref{tab:2}, we show the simulation results by changing the value of the mean shift $\delta$. The RMSE of the FASM is always smaller than the compared model. The degrees of freedom when $\delta =1$ are similar. When $\delta$ increases, the degrees of freedom is smaller for the proposed model. Therefore, we achieve better fit but less flexibility with the FASM.

\begin{figure}[!htbp]
\centering
{\includegraphics[width = 8.65cm]{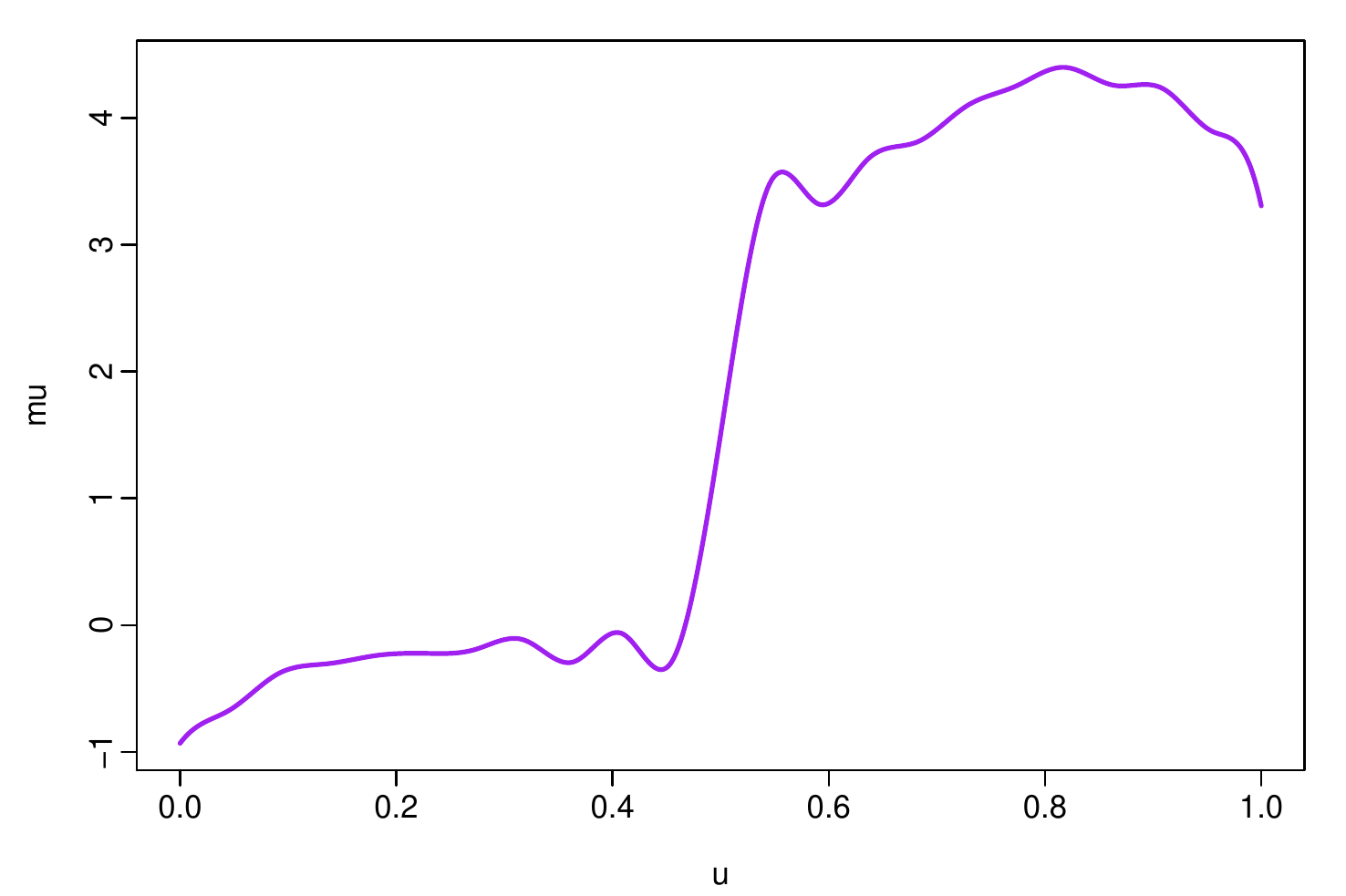}}
\caption{The mean function $\mu(u)$.}\label{fig:8}
\end{figure}

\begin{figure}[!htbp]
\centering
\subcaptionbox{Residuals from applying smoothing model}
{\includegraphics[width = 8.65cm]{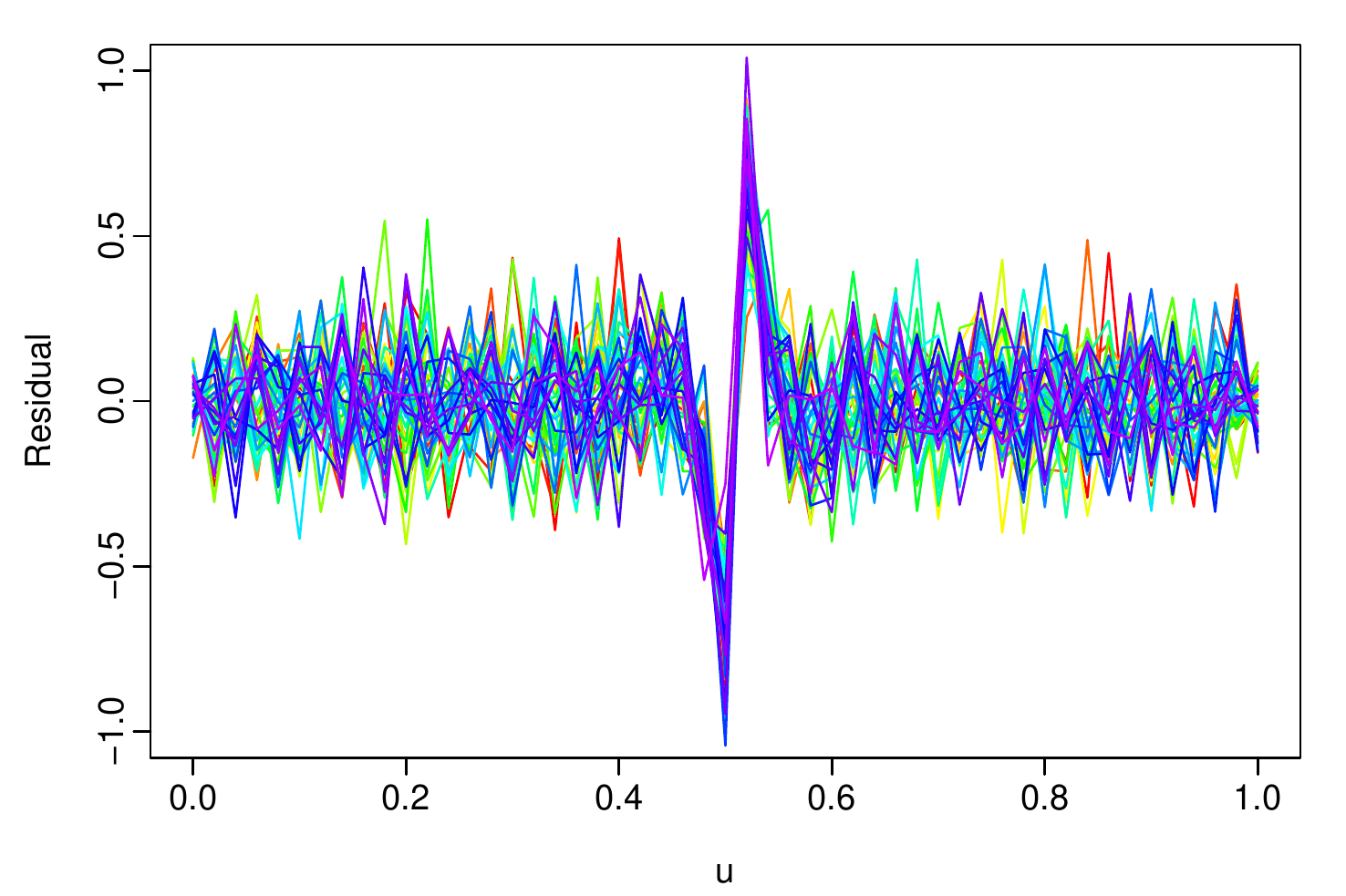}}
\subcaptionbox{Residuals from FASM}
{\includegraphics[width = 8.65cm]{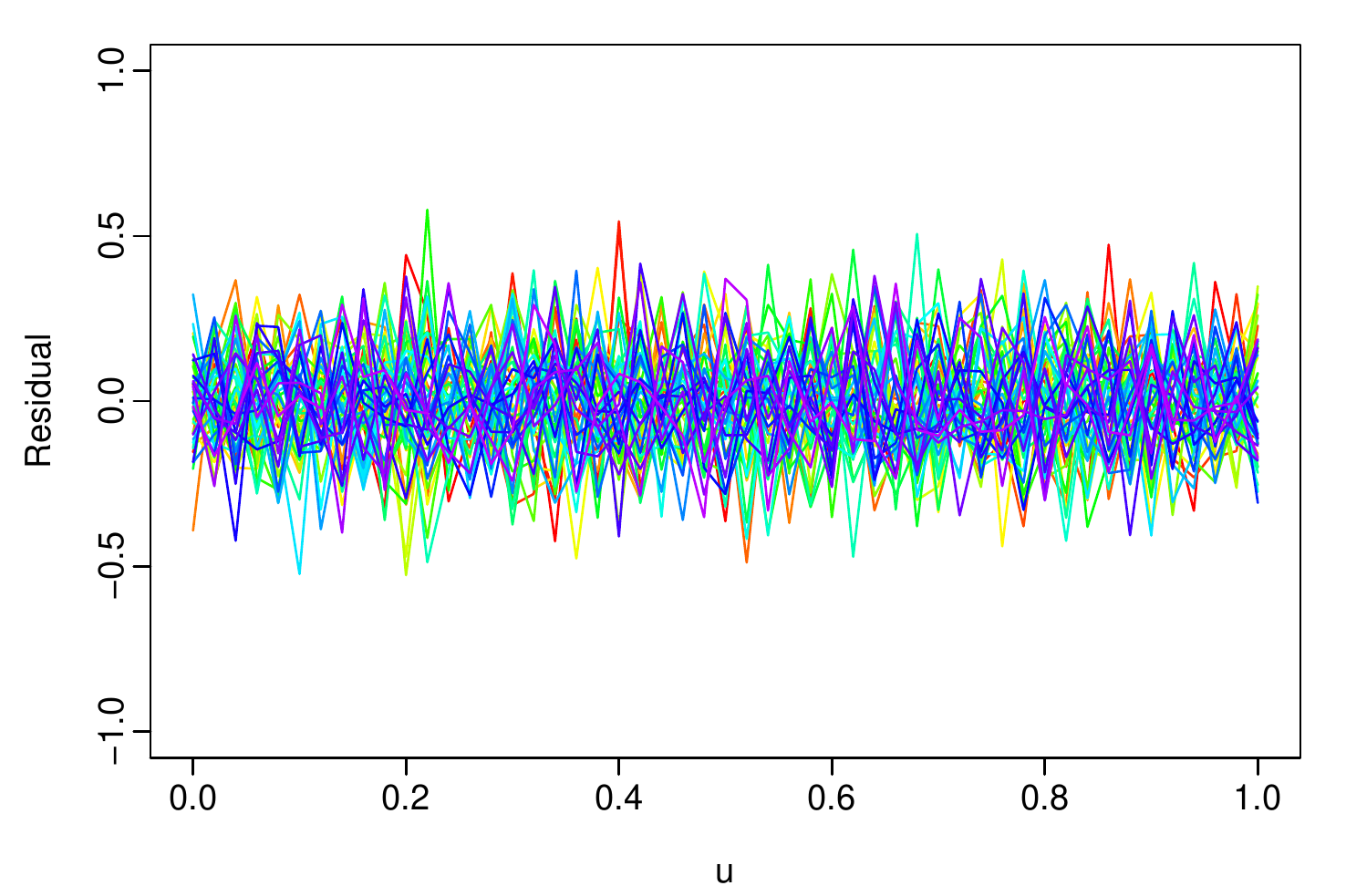}}
\caption{Residual plots of the two models.}\label{fig:3}
\end{figure}

\begin{table} 
\centering 
  \caption{The trade-off between model fitting and flexibility.}\label{tab:2} 
\begin{tabular}{@{\extracolsep{5pt}} ccccc} 
\toprule
& \multicolumn{2}{c}{RMSE}&\multicolumn{2}{c}{DF} \\
\hline
& Smoothing & Proposed & Smoothing & Proposed\\
$\delta = 1$&$0.2045$&$\bm{0.1631}$&$\bm{10.68}$&$11.15$\\
$\delta = 2$&$0.2063$&$\bm{0.1640}$&$17.59$&$\bm{11.03}$\\
$\delta = 3$&$0.3308$&$\bm{0.1647}$&$14.23$&$\bm{10.94}$\\
\hline
\end{tabular}
\end{table}

\section{Application to climatology}\label{se:7}

In this section, we apply the FASM to two real data sets. In Section~\ref{se:7.1}, we compare Canadian yearly temperature and precipitation data and demonstrate the advantages of the FASM when the measurement error is large. In Section~\ref{se:7.2}, we analyze Australian daily temperature data and display the necessity of including the factor model because of the spike structure of the data.

\subsection{Canadian weather data}\label{se:7.1}

In Section~\ref{se:2.1}, we introduced Canadian weather data. Raw observations of daily temperature and precipitation data are presented in Figure~\ref{fig:2}. Since the true basis functions are unknown, we apply the FASM with nonparametric smoothing splines introduced in Section~\ref{se:nonparam} to these two datasets. 

We use order 4 B-spline basis functions with knots at every data point. Thus, when the number of data points is 365, we use 367 basis functions. The number of factors $r$ is chosen with the scree plot showing the fraction of variation explained. For temperature data, we presumed the measurement error is small. The resulting smoothed curves are shown in Figure~\ref{fig:4}. Compared with using the smoothing model introduced in Section~\ref{se:6.2}, the FASM generates similar results. This meets our expectation that our model should work the same as a simple smoothing model when measurement error does not exist. 

In Section~\ref{se:2.1}, we suspect large measurement errors are contained in the raw log precipitation data. We apply the two models to the log precipitation data; the resulting smoothed curves are presented in Figure~\ref{fig:5}. The plot on the right shows smoother curves, especially at the drop in the blue curve (the 'Victoria' Station) at around day 200. Looking at the residual plots in Figure~\ref{fig:6}, our model mainly explains some extreme residuals left out from solely applying the smoothing model. As in Section~\ref{se:5}, we also compare the RMSE and degrees of freedom of the two fitted models; they are 0.1933 and 14.41 for the smoothing model and 0.1659 and 12.71 respectively for the proposed model. Thus, in terms of model selection, our model performs better across both model fit and model simplicity.

\begin{figure}[!htbp]
\centering
\subcaptionbox{Smoothed temperature curves from basis smoothing with the penalty}
{\includegraphics[width = 8.65cm]{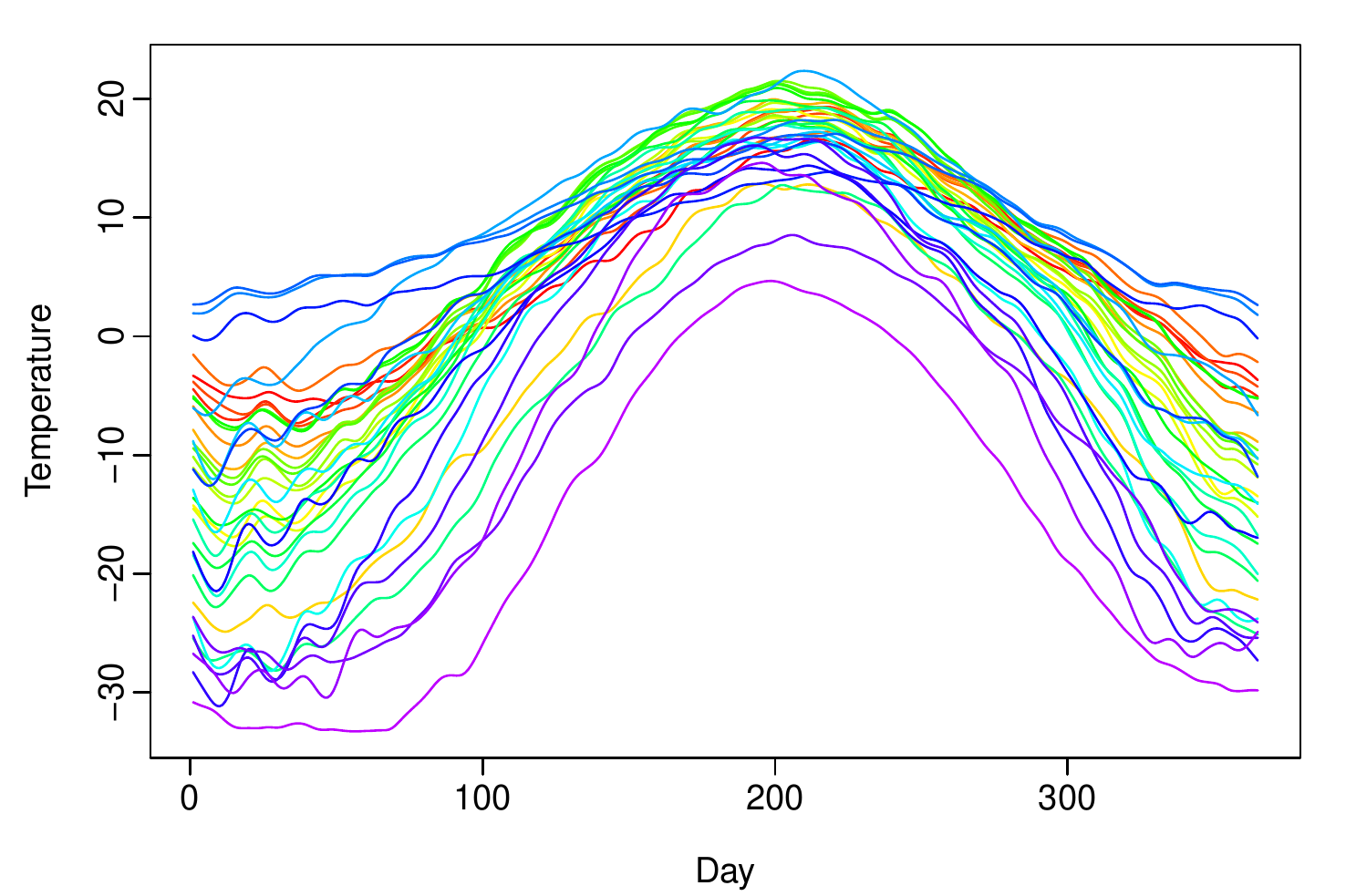}}
\subcaptionbox{Smoothed temperature curves from the FASM}
{\includegraphics[width = 8.65cm]{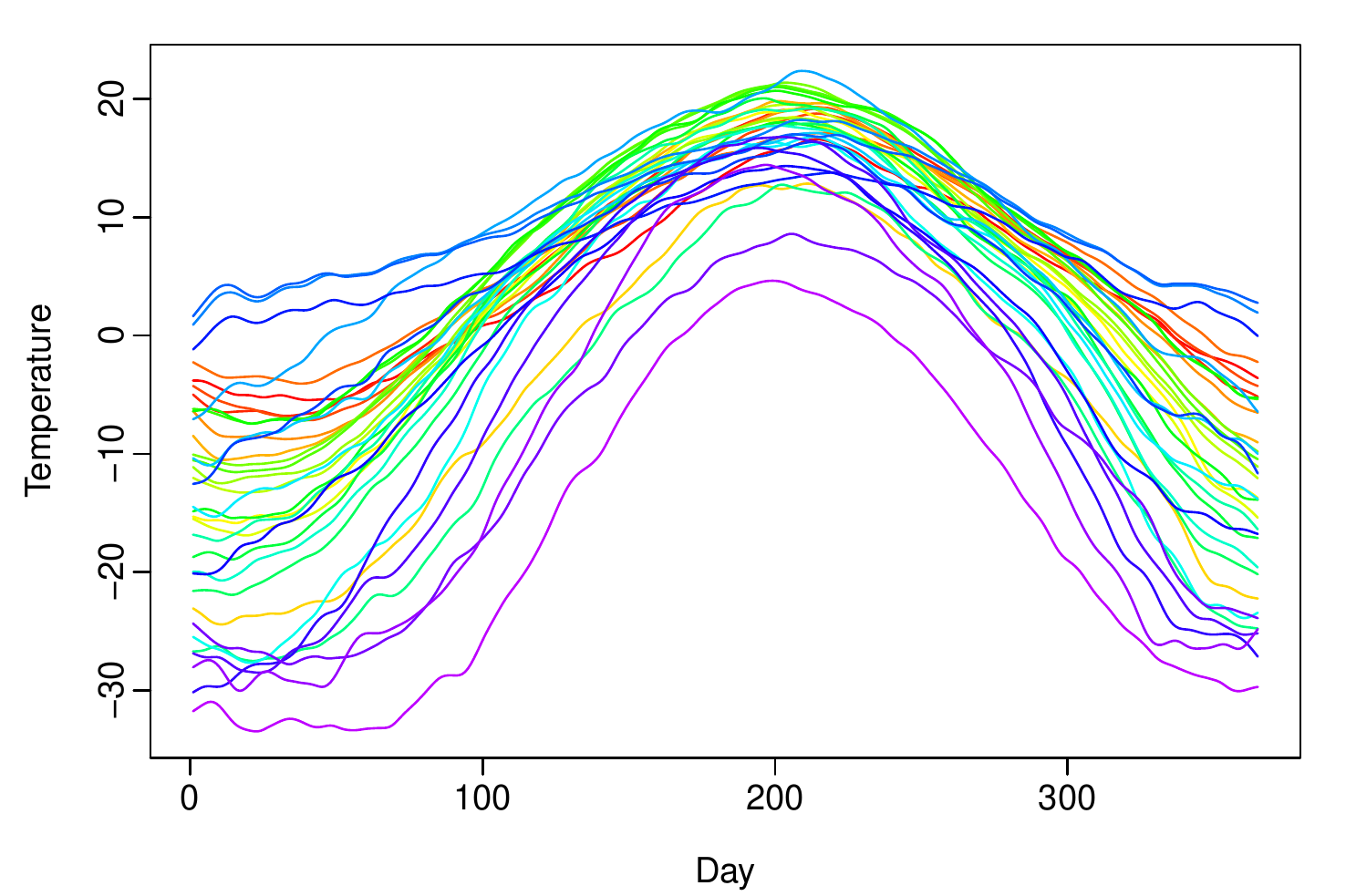}}
\caption{Comparison between the smoothed curves.}\label{fig:4}
\end{figure}

\begin{figure}[!htbp]
\centering
\subcaptionbox{Smoothed log precipitation curves using the smoothing model}
{\includegraphics[width = 8.65cm]{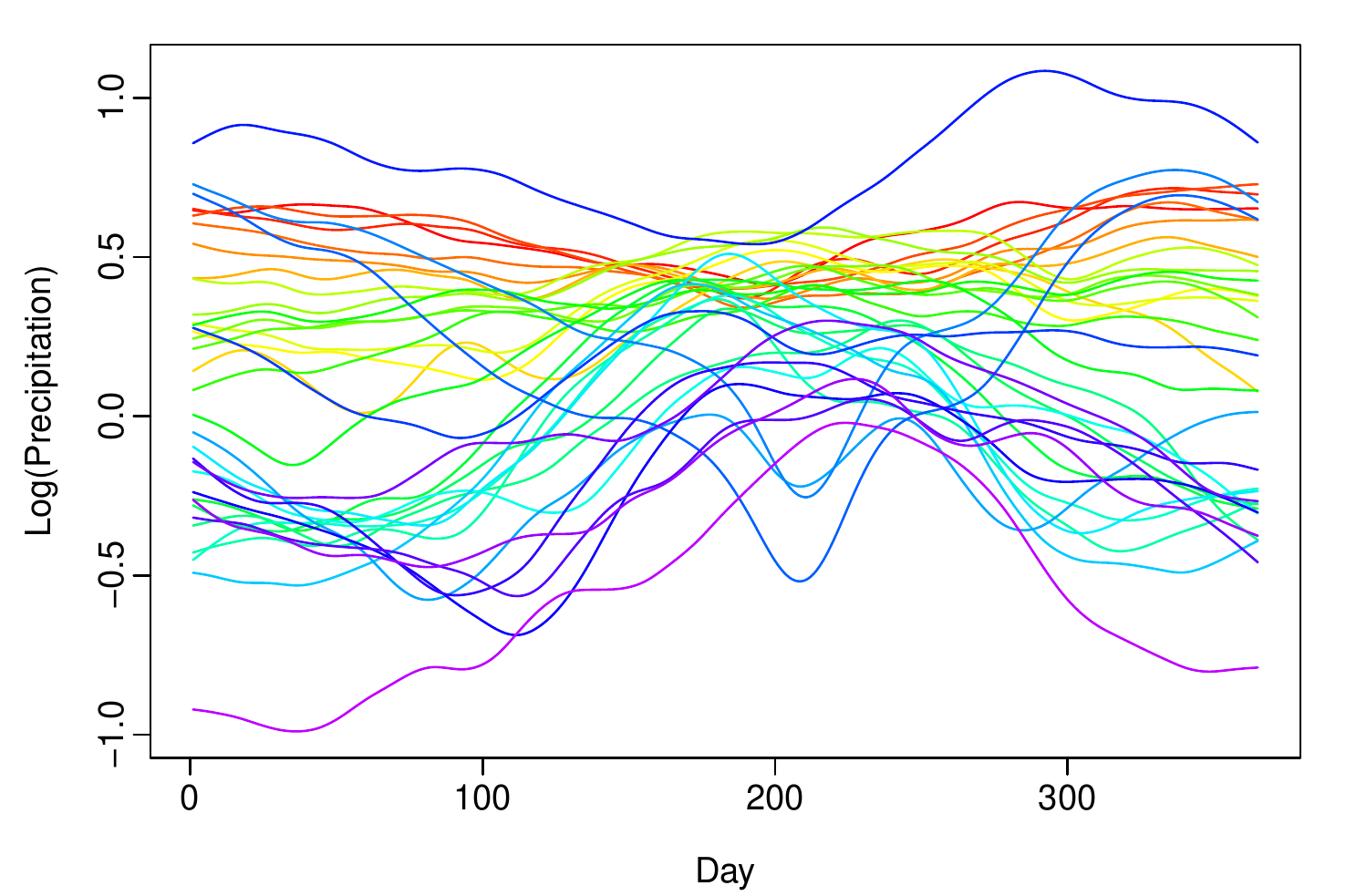}}
\subcaptionbox{Smoothed log precipitation curves using the FASM}
{\includegraphics[width = 8.65cm]{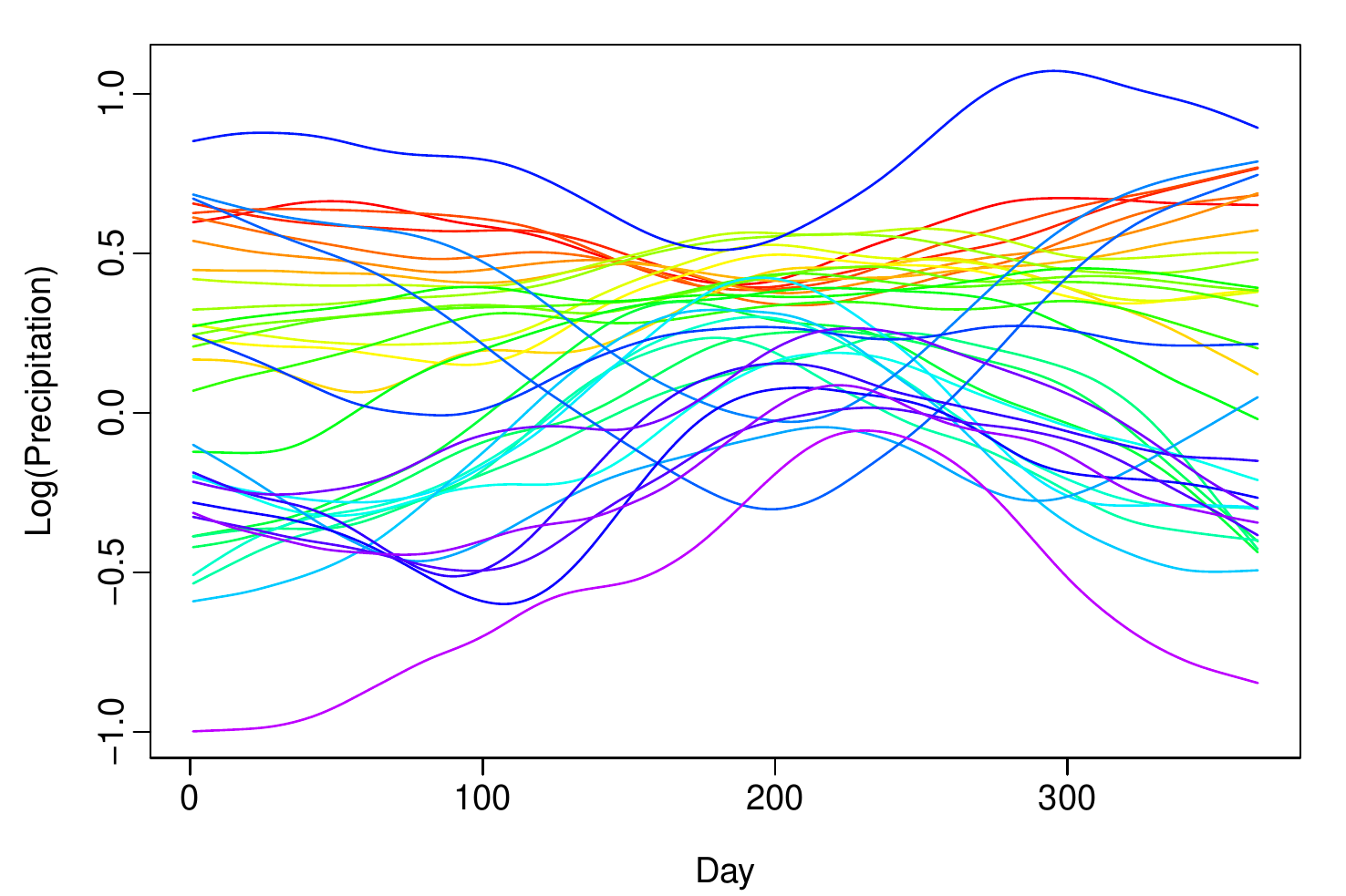}}
\caption{Comparison between the smoothed curves.}\label{fig:5}
\end{figure}

\begin{figure}[!htbp]
\centering
\subcaptionbox{Residuals from basis smoothing with the penalty}
{\includegraphics[width = 8.65cm]{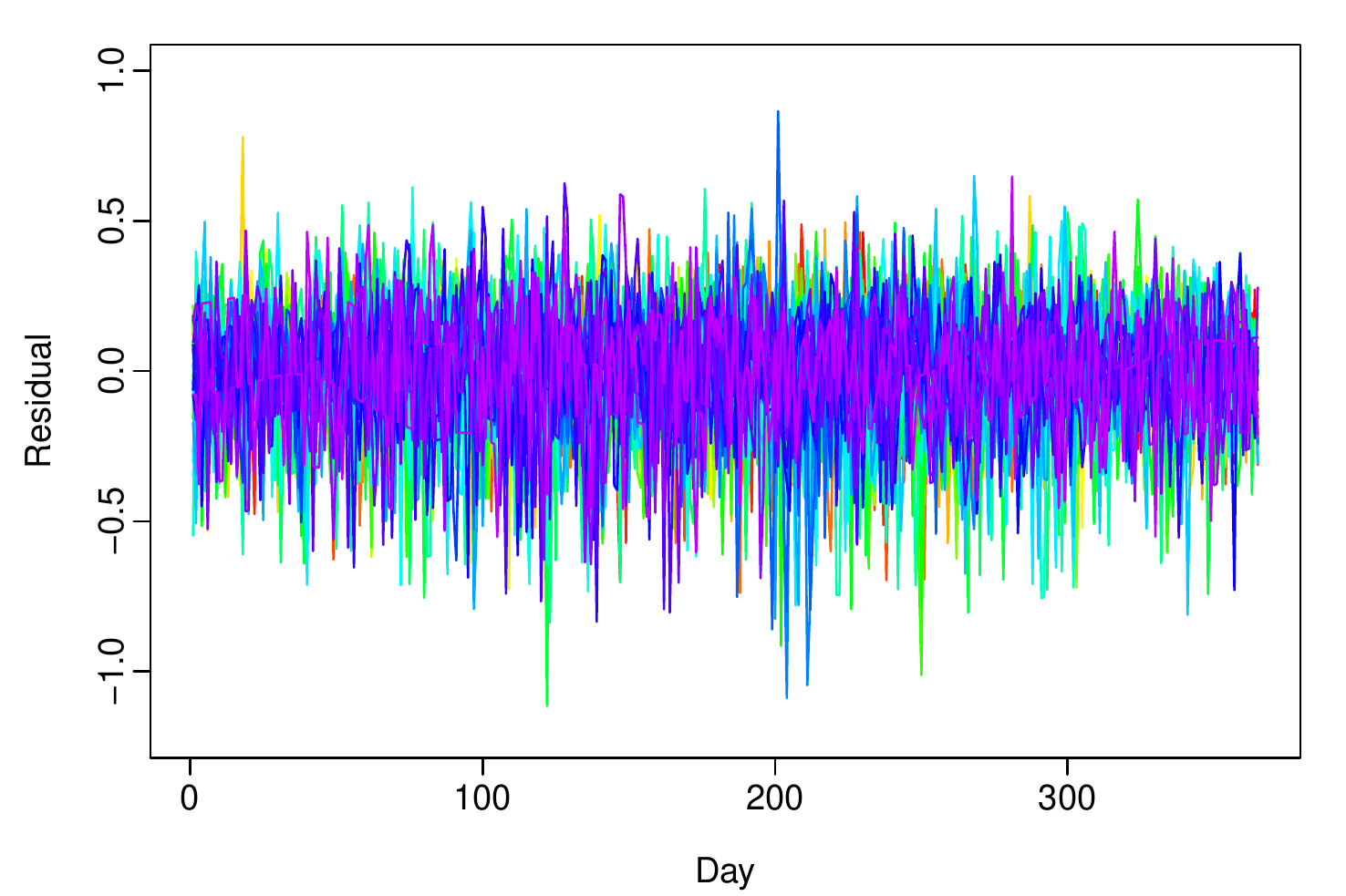}}
\subcaptionbox{Residuals from the FASM}
{\includegraphics[width = 8.65cm]{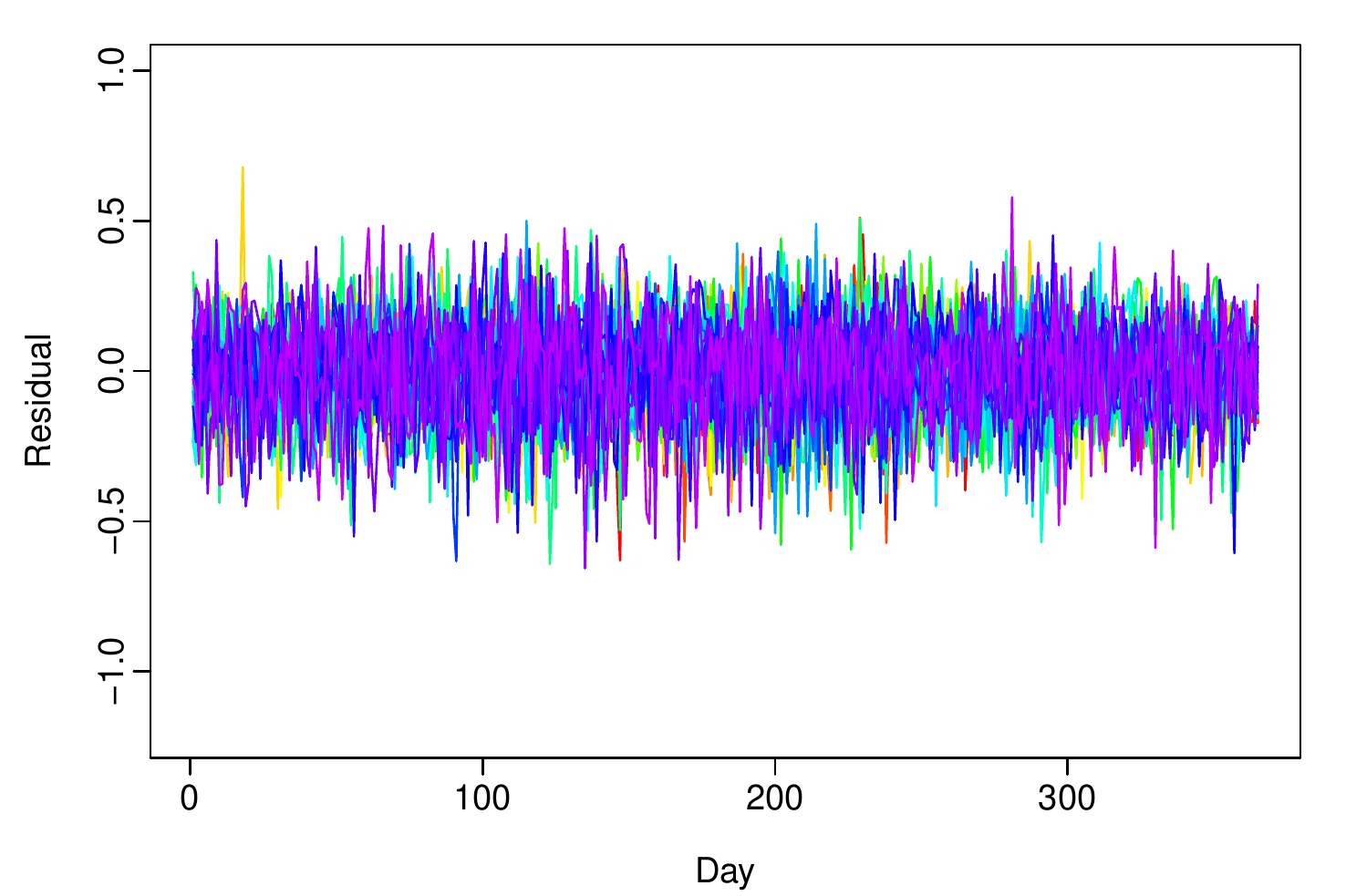}}
\caption{Comparison between the residuals.}\label{fig:6}
\end{figure}

\subsection{Australian temperature data}\label{se:7.2}

In this section, we consider Friday temperature data at Adelaide airport. We choose Adelaide because it tends to have the hottest temperature among Australia's big cities. Data from other weekdays exhibit similar features and are not shown here. The data are measured every half an hour from the year 1997 to 2007. The sample size $n$ is 508, and the number of discrete data points $p$ from each curve is 48. The plot of the raw data can be found in Figure~\ref{fig:7a}. It can be seen that the data are quite noisy, with extreme values in some of the curves due to large measurement errors.

We use the B-spline basis functions of order 4 with knots at every data point. A penalized smoothing model is fitted to the data, with the tuning parameter selected to minimize the mGCV value. The residuals are shown in Figure~\ref{fig:7b}. As can be seen, the smoothing model fails to capture the extreme values in the data. 

We check the residuals' spikiness in Figure~\ref{fig:7b} by conducting a principal component analysis. The eigenvalues in descending order are shown in Figure~\ref{fig:7c}. The first few eigenvalues are significantly larger than the rest. This means the residuals contain information captured by just a few factors, which calls for a further dimension reduction model on the residuals. 

As a comparison, the FASM is also applied to the data. The tuning parameter for the smoothing part is selected based on mGCV at each step of the iteration. The number of factors retained in the factor model component is five. The residuals are shown in Figure~\ref{fig:7d}. The extreme values are almost all removed from the remaining residuals. 

\begin{figure}[!htbp]
\centering
\subcaptionbox{Raw temperature data\label{fig:7a}}
{\includegraphics[width = 8.65cm]{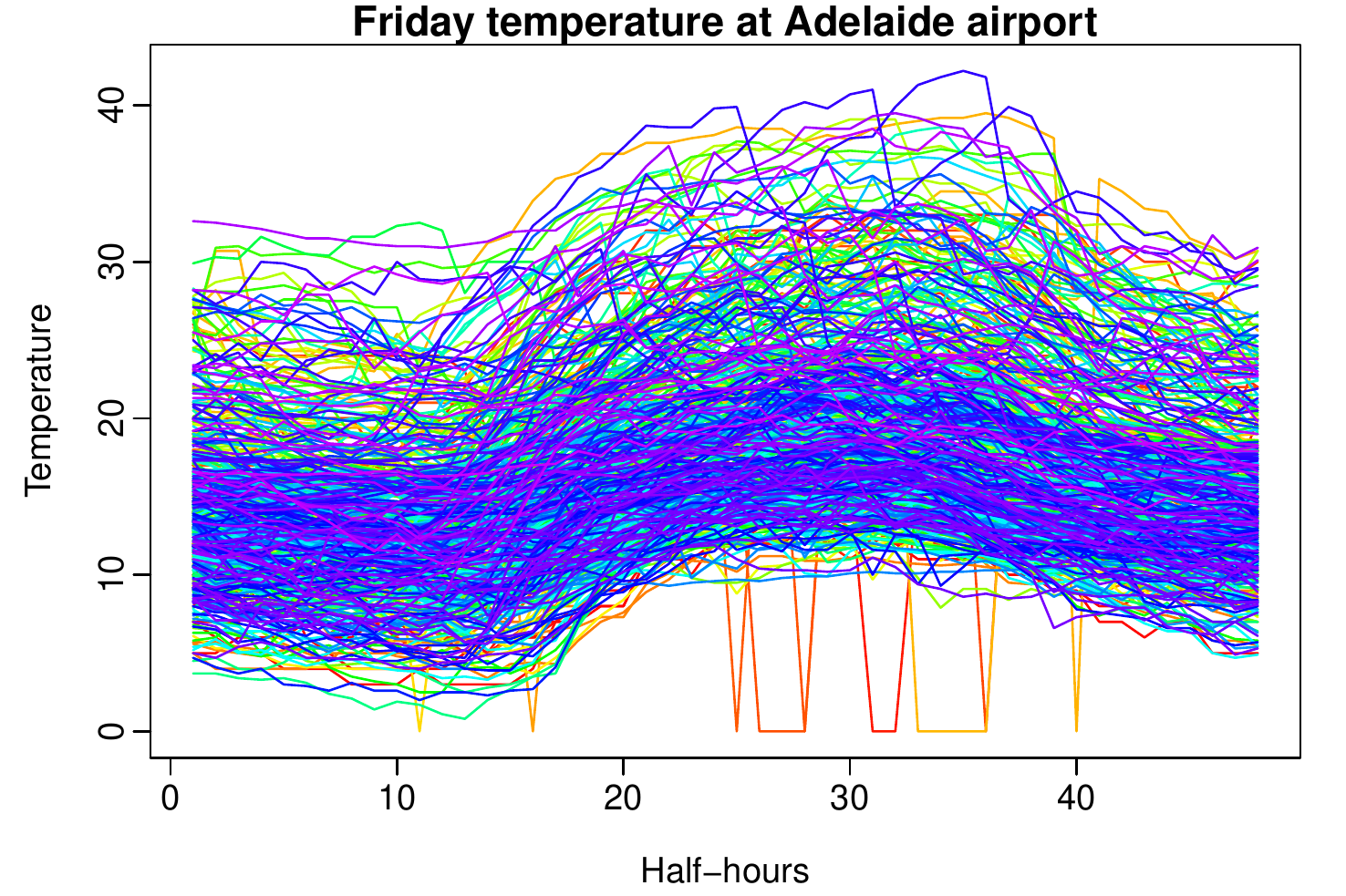}}
\subcaptionbox{Residuals of the smoothing model\label{fig:7b}}
{\includegraphics[width = 8.65cm]{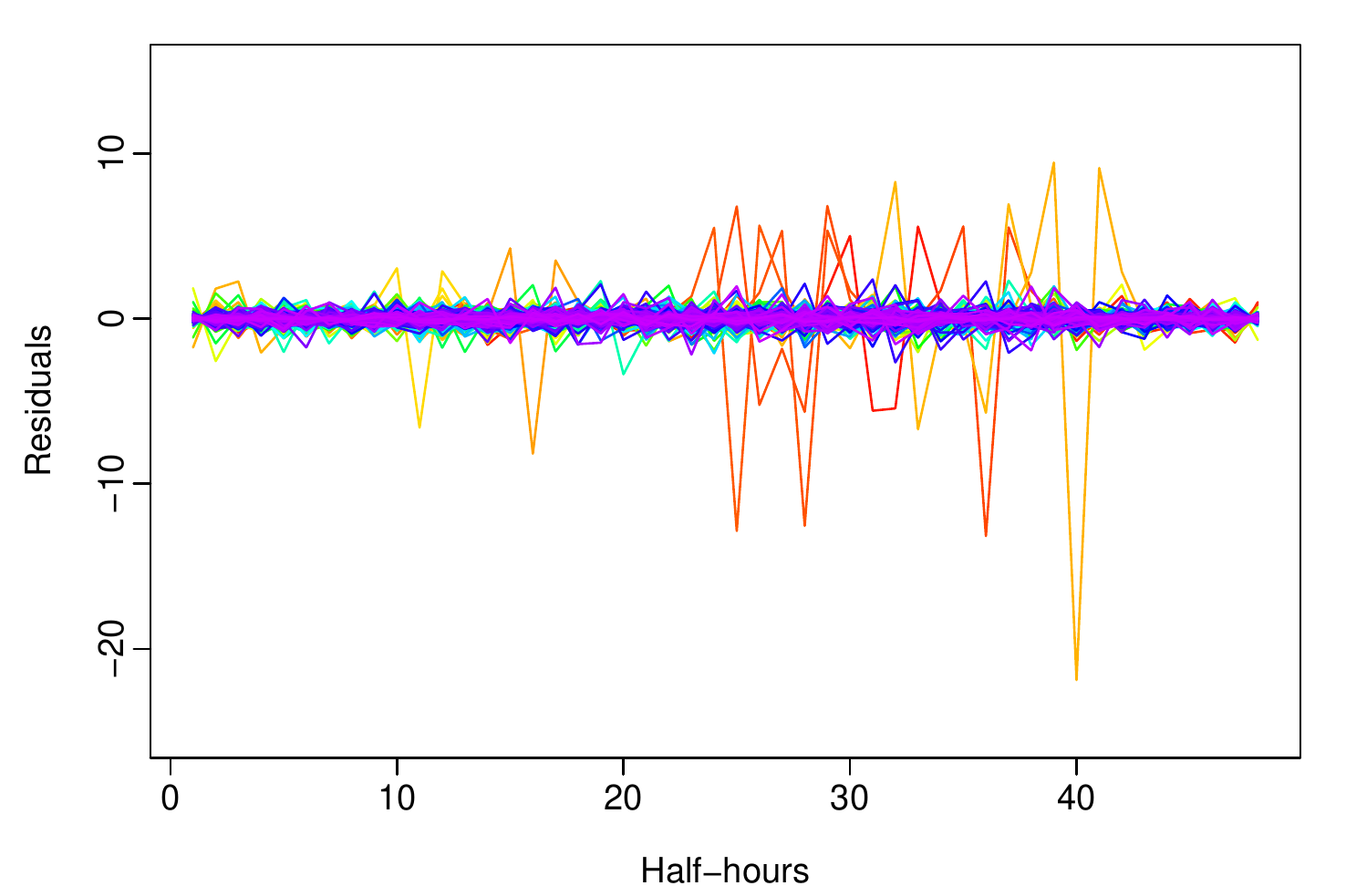}}
\subcaptionbox{Eigenvalues of the residuals in the plot (b)\label{fig:7c}}
{\includegraphics[width = 8.65cm]{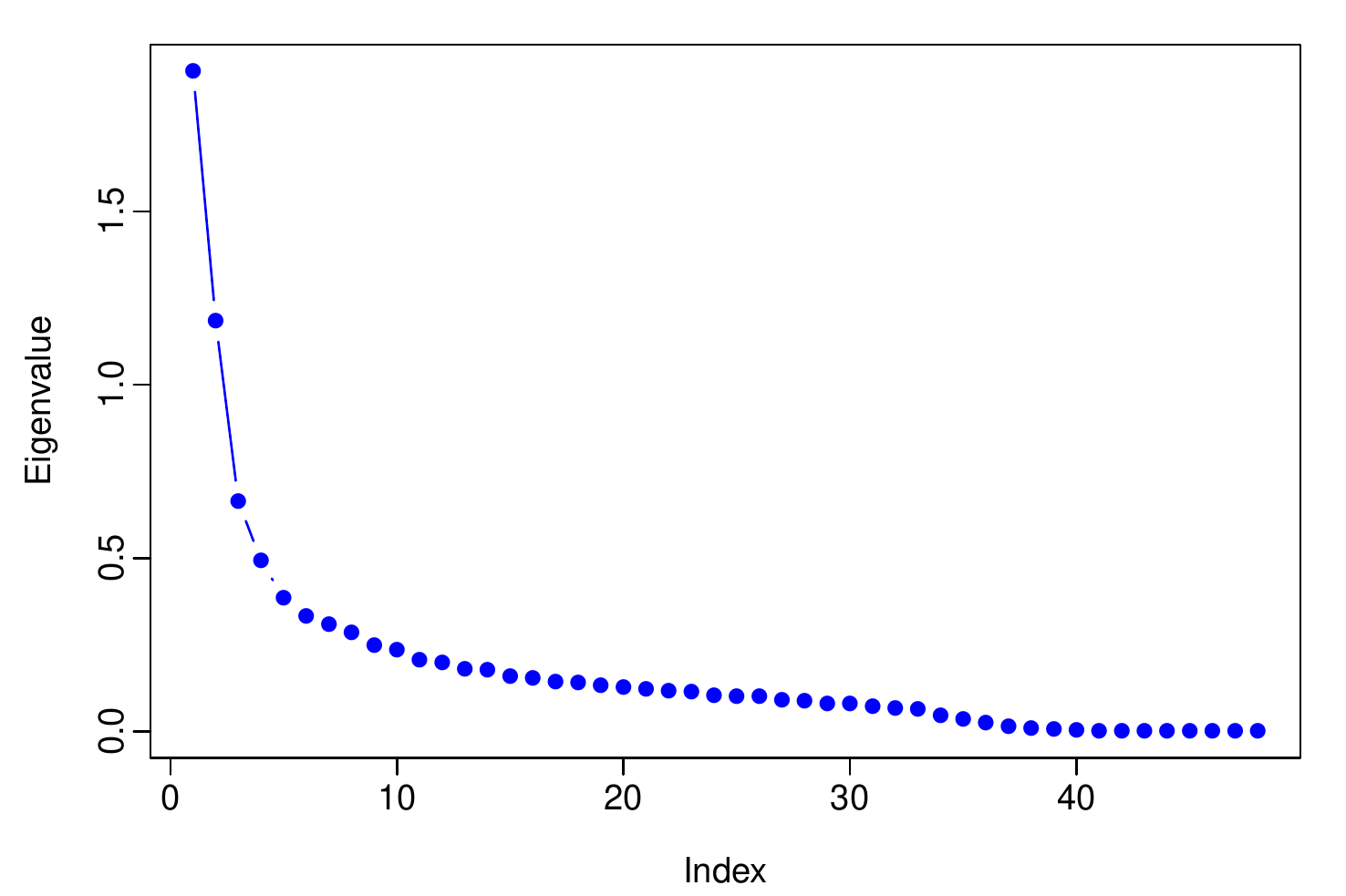}}
\subcaptionbox{Residuals of the FASM\label{fig:7d}}
{\includegraphics[width = 8.65cm]{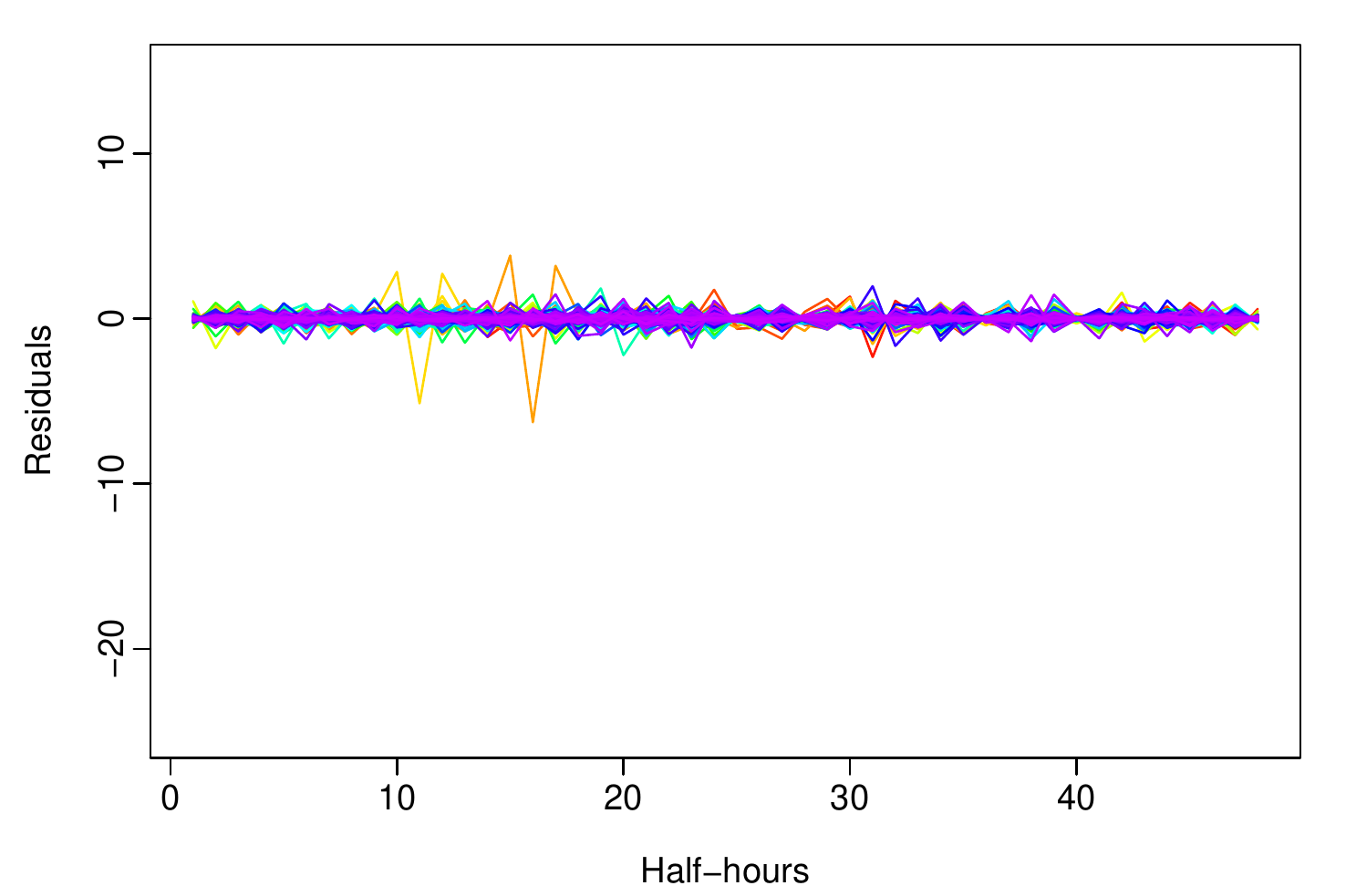}}
\caption{Half-hourly temperature data on Friday at Adelaide airport.}\label{fig:7}
\end{figure}

\section{Conclusion and future work}\label{se:8}

In this paper, we propose a factor-augmented smoothing model for functional data. We study raw functional data, which is a mixture of functional curves and high-dimensional errors. When measurement error is informative, a smoothing model alone is inadequate to capture data variation and recover the signal functional component. The proposed model incorporates a factor structure into the smoothing model to further explain the large residuals. We propose a numerical iteration approach to simultaneously obtain estimates in the smoothing model and the factor model. The asymptotic distribution of the estimators is given with proof. Our model also serves as a dimension reduction method on functional and high-dimensional mixture data, easing the path to making inferences. We provide an example of the construction of a covariance estimator for the raw data. Further, we show that the model can be applied in situations where there is misidentification in the data structure, two examples of which are the misspeficications of smoothing basis functions and the neglect of the step jumps in the mean level of the functions. The advantages of the proposed model are demonstrated in extensive simulation studies. We also show how our model performs via applications to Canadian weather data and Australian temperature data.

The proposed model is a good start point for modeling complex data structures. The data we deal with are a mixture of smooth functional curves and high-dimensional measurement error. The factor model component can be regarded as a ”boosting” component that improves model accuracy. Extending from this idea, the model can be applied to other data structures. One example is the data that contain change points. Change point is a popular problem in many statistics and econometric topics and has been extensively studies in the multivariate setting. Previous literature on change point in functional data include \cite{BG09, HK10} and \cite{HK15}. It is shown in simulation examples that our model can be used for modeling functional data with change point in the cross-sectional direction. The model can be modified to account for change point also in the sample direction. Further research can be conducted along this line.

\newpage
\bibliographystyle{agsm}
\bibliography{FAST}

\newpage
\clearpage


%
\def\spacingset#1{\renewcommand{\baselinestretch}
{#1}\small\normalsize} \spacingset{1}

\if0\blind
{
  \title{\bf Supplement to ``Factor-augmented Smoothing Model for Functional Data"}
\author{Yuan Gao\thanks{Postal address: Research School of Finance, Actuarial Studies and Statistics, Level 4, Building 26C, Kingsley St, Australian National University, Canberra, ACT 2601, Australia; Email: yuan.gao@anu.edu.au}
 \hspace{.2cm}\\
Research School of Finance, Actuarial Studies and Statistics\\ 
Australian National University \\ 
\\
  Han Lin Shang
  \hspace{.2cm}\\
    Department of Actuarial Studies and Business Analytics \\
    Macquarie University \\
 \\ 
 Yanrong Yang \\  
Research School of Finance, Actuarial Studies and Statistics\\ 
Australian National University \\ 
 }
 \emptythanks
  \maketitle
} \fi

\pagenumbering{arabic}
\spacingset{1.48}

This section contains the proofs for the theorems in the main article. In Appendix A, we provide the proofs for the theorems in Section~\ref{se:4}. In Appendix B,  we include the results of a proposition and its proof. In Appendix C, the lemmas used for the proofs in Appendix A and B are stated as well as their proofs.

\subsection*{Appendix A}\label{app:a}

Theorem~\ref{th:2} is the main result of the asymptotic theories, and the proof of it is lengthy. Thus, we include in the following the outlines for the proof before we show the details.

\subsubsection*{Outlines for proof of Theorem 2}

In Theorem~\ref{th:2}, we find the order of convergence of the estimated coefficient matrix $\widehat{\bm{C}}$. The difference between $\widehat{\bm{C}}$ and $\bm{C}^0$ could be written into three terms: 
\begin{equation}
\frac{1}{p}\left(\bm{\Phi}^\top \bm{M}_{\widehat{\bm{A}}}\bm{\Phi} + \alpha \bm{R} \right)(\widehat{\bm{C}}-\bm{C}^0) = \frac{1}{p}\alpha \bm{R}\bm{C}^0 + \frac{1}{p}\bm{\Phi}^\top \bm{M}_{\widehat{\bm{A}}}\bm{A}^0\bm{F}^\top + \frac{1}{p}\bm{\Phi}^\top \bm{M}_{\widehat{\bm{A}}}\bm{E}.
\end{equation}
The term $\frac{1}{p}\left(\bm{\Phi}^\top \bm{M}_{\widehat{\bm{A}}}\bm{\Phi} + \alpha \bm{R} \right)$ is $O_p(1)$. The first term on the right-hand side of~\eqref{eq:23} comes from the penalty, and the order can be found easily from Assumption~\ref{as:6}. The third term contains the random error matrix $\bm{E}$, and the order can be found using the result in Lemma~\ref{le:5}. The second term is the most complicated one, and we show in the following proof that it could be further broken down into eight terms. We find the order of each of the eight terms using the lemmas in Appendix C. Most of the terms can be shown to be $o_p\left(\left\|\bm{C}^0-\widehat{\bm{C}}\right\|\right)$ and thus can be omitted. Combining the remaining terms, we arrive at the result
\begin{align}\label{eq:13}
 \left(\widehat{\bm{C}}-\bm{C}^0\right)\bm{M}_{\bm{F}} =& \bm{Q}^{-1}\left(\widehat{\bm{A}}\right)\frac{1}{p}\alpha \bm{R} \bm{C}^0 + \bm{Q}^{-1}\left(\widehat{\bm{A}}\right)\frac{1}{p}\bm{\Phi}^\top \bm{M}_{\widehat{\bm{A}}}\bm{E}\bm{M}_{\bm{F}} \nonumber\\
&+ O_p\left(\frac{1}{\min (n,p)}\right) + O_p\left(\frac{\sqrt{n}}{p\sqrt{p}} \right) + O_p\left(\frac{1}{\sqrt{np}}\right),
\end{align}
where matrix $\bm{Q}$ and $\bm{M}_{\bm{F}}$ are
\begin{equation*}
\bm{Q}\left(\bm{\widehat{\bm{A}}}\right) = \frac{1}{p}\bm{\Phi}^\top\bm{M}_{\widehat{\bm{A}}}\bm{\Phi} \qquad \bm{M}_{\bm{F}} = \bm{I}_n - \bm{F}\left(\bm{F}^\top\bm{F}\right)^{-1}\bm{F}^\top. 
\end{equation*}
The first term on the right-hand side of~\eqref{eq:13} is $O_p(1)$ using the assumption on the tuning parameter $\alpha$. We also show the second term is $O_p(1)$ using results from the lemmas. When $n$ and $p$ are of the same order, we are able to show the norm of $ \left(\widehat{\bm{C}}-\bm{C}^0\right)$ projected on the matrix $\bm{M}_{\bm{F}}$ is $O_p(1)$.

Next begins the formal proofs.

\subsubsection*{Proof of Theorem~\ref{th:1}}

\begin{proof}
The concentrated objective function defined in Section~\ref{se:4.2} is 
\begin{equation*}
S_{np} (\bm{c}_i, \bm{A}) = \frac{1}{np}\sum_{i=1}^n \left[ (\bm{Y}_i-\bm{\Phi}\bm{c}_i)^\top \bm{M}_{\bm{A}} (\bm{Y}_i-\bm{\Phi}\bm{c}_i) + \alpha \bm{c}_i^\top \bm{R}\bm{c}_i \right] - \frac{1}{np}\sum_{i=1}^n \bm{\epsilon}_i^\top \bm{M}_{\bm{A}^0}\bm{\epsilon}_i.
\end{equation*}
Assume $\bm{c}_i^0 = \bm{0}$ for simplicity without loss of generality. From $\bm{Y_i} = \bm{\Phi} \bm{c}_i^0 + \bm{A}^0 \bm{f}_i + \bm{\epsilon}_i = \bm{A}^0 \bm{f}_i + \bm{\epsilon}_i$, we have
\begin{align*}
S_{np} (\bm{c}_i, \bm{A}) =& \frac{1}{np}\sum_{i=1}^n \left[ (\bm{A}^0\bm{f}_i + \bm{\epsilon}_i- \bm{\Phi}\bm{c}_i)^\top \bm{M}_{\bm{A}} (\bm{A}^0\bm{f}_i + \bm{\epsilon}_i-\bm{\Phi}\bm{c}_i) + \alpha \bm{c}_i^\top \bm{R}\bm{c}_i \right] - \frac{1}{np}\sum_{i=1}^n \bm{\epsilon}_i^\top \bm{M}_{\bm{A}^0}\bm{\epsilon}_i\\
=& \frac{1}{np}\sum_{i=1}^n  \bm{f}_i^\top {\bm{A}^0}^\top \bm{M}_{\bm{A}} \bm{A}^0\bm{f}_i +  \frac{1}{np}\sum_{i=1}^n  \bm{c}_i^\top {\bm{\Phi}}^\top \bm{M}_{\bm{A}} \bm{\Phi}\bm{c}_i - \frac{2}{np}\sum_{i=1}^n  \bm{f}_i^\top {\bm{A}^0}^\top \bm{M}_{\bm{A}} \bm{\Phi}\bm{c}_i  \\
&+  \frac{2}{np}\sum_{i=1}^n  \bm{\epsilon}_i^\top \bm{M}_{\bm{A}} \bm{A}^0\bm{f}_i - 
 \frac{2}{np}\sum_{i=1}^n  \bm{\epsilon}_i^\top \bm{M}_{\bm{A}} \bm{\Phi}\bm{c}_i + 
  \frac{1}{np}\sum_{i=1}^n  \bm{\epsilon}_i^\top (\bm{M}_{\bm{A}}-\bm{M}_{\bm{A}^0}) \bm{\epsilon}_i + \frac{\alpha}{np}\sum_{i=1}^n  \bm{c}_i^\top \bm{R}\bm{c}_i. 
\end{align*}

Denote the first three terms in the above equation as
\begin{equation*}
\widetilde{S}_{np}(\bm{c}_i, \bm{A}) = \frac{1}{np}\sum_{i=1}^n  \bm{f}_i^\top {\bm{A}^0}^\top \bm{M}_{\bm{A}} \bm{A}^0\bm{f}_i +  \frac{1}{np}\sum_{i=1}^n  \bm{c}_i^\top {\bm{\Phi}}^\top \bm{M}_{\bm{A}} \bm{\Phi}\bm{c}_i - \frac{2}{np}\sum_{i=1}^n  \bm{f}_i^\top {\bm{A}^0}^\top \bm{M}_{\bm{A}} \bm{\Phi}\bm{c}_i.
\end{equation*}
Then by Lemma~\ref{le:9},
\begin{equation*}
S_{np}(\bm{c}_i, \bm{A})  = \widetilde{S}_{np}(\bm{c}_i, \bm{A}) + o_p(1).
\end{equation*}
It is easy to see that $\widetilde{S}_{np}(\bm{c}_i^0 = \bm{0}, \bm{A}^0\bm{H} ) = 0 $ for any $r\times r$ invertible $\bm{H}$, because $\bm{M}_{\bm{A}^0\bm{H}} = \bm{M}_{\bm{A}^0} $ and $\bm{M}_{\bm{A}^0}\bm{A}^0=\bm{0} $.

Here we define two matrix operations before further transformations on $\tilde{S}_{np}(\bm{c}_i,\bm{A})$. For an $m\times n $ matrix $\bm{U}$ and a $p \times q$ matrix $\bm{V}$, the vectorization of $\bm{U}$ is defined as
\begin{equation*}
\text{vec}(\bm{U}) \equiv (u_{1,1},\dots, u_{m,1}, u_{1,2}, \dots, u_{m,2}, u_{1,n}, \dots, u_{m,n})^\top,
\end{equation*}
and the Kronecker product $\bm{U}\otimes\bm{V}$ is the $pm\times qn$ block matrix defined as
\begin{align*}
\bm{U}\otimes \bm{V} \equiv \begin{bmatrix}
u_{1,1}\bm{V} & \dots & u_{1,n} \bm{V}\\
\vdots & &\vdots\\
u_{m,1}\bm{V} &\dots & u_{m,n}\bm{V},
\end{bmatrix}
\end{align*}
where $u_{ij}$ represents the element on the $i$th row and $j$th column of matrix $\bm{U}$.

Next we can further write $\tilde{S}_{n,p}(\bm{c}_i, \bm{A})$ as
\begin{align*}
\widetilde{S}_{np}(\bm{c}_i, \bm{A}) &=  \text{vec}(\bm{M}_{\bm{A}}\bm{A}^0)^\top \left(\frac{\bm{F}^\top\bm{
F}}{np}\otimes \bm{I}_p\right) \text{vec}(\bm{M}_{\bm{A}}\bm{A}^0) +\frac{1}{n}\sum_{i=1}^n \bm{c}_i^\top \left(\frac{1}{p}\bm{\Phi}^\top \bm{M}_{\bm{A}}\bm{\Phi}\right)\bm{c}_i\\
&- \frac{1}{n}\sum_{i=1}^n 2 \bm{c}_i^\top \left(\frac{1}{p}\bm{f}_i\otimes \bm{M}_{\bm{A}}\bm{\Phi}\right)\text{vec}(\bm{M}_{\bm{A}}\bm{A}^0).
\end{align*}
If we denote
\begin{align*}
\bm{P} &= \frac{1}{p}\bm{\Phi}^\top\bm{M}_{\bm{A}}\bm{\Phi},\\
\bm{W} &= \frac{\bm{F}^\top\bm{
F}}{np}\otimes \bm{I}_p,\\
\bm{V}_i &= \frac{1}{p}\bm{f}_i\otimes \bm{M}_{\bm{A}}\bm{\Phi},
\end{align*}
and $\bm{\gamma} = \text{vec}(\bm{M}_{\bm{A}}\bm{A}^0)$, then we can write
\begin{align*}
\widetilde{S}_{np}(\bm{c}_i, \bm{A}) &= \frac{1}{n}\sum_{i=1}^n \left[\bm{c}_i^\top \bm{P} \bm{c}_i + \bm{\gamma}^\top \bm{W} \bm{\gamma} -2\bm{c}_i^\top \bm{V}_i^\top \bm{\gamma}\right]\\
&= \frac{1}{n}\sum_{i=1}^n \left[\bm{c}_i^\top\left( \bm{P}-\bm{V}_i^\top \bm{W}^{-1}\bm{V}_i\right)\bm{c}_i + (\bm{\gamma}^\top - \bm{c}_i^\top \bm{V}_i^\top \bm{W}^{-1})\bm{W}(\bm{\gamma}^\top - \bm{W}^{-1}\bm{V}_i\bm{c}_i) \right]\\
&\equiv \frac{1}{n}\sum_{i=1}^n \left[ \bm{c}_i^\top \bm{D}_i(\bm{A}) \bm{c}_i + \bm{\theta}_i^\top \bm{W}\bm{\theta}_i \right].
\end{align*}

In the last equation, 
\begin{align*}
\bm{D}_i(\bm{A}) &\equiv \bm{P}-\bm{V}_i^\top \bm{W}^{-1}\bm{V}_i\\
&= \frac{1}{p}\bm{\Phi}^\top \bm{M}_{\bm{A}}\bm{\Phi} - \frac{1}{p}\bm{\Phi}^\top \bm{M}_{\bm{A}}\bm{\Phi}\bm{f}_i^\top\left(\frac{\bm{F}^\top\bm{F}}{n}\right)^{-1}\bm{f}_i,\\
\bm{\theta}_i &\equiv \bm{\gamma}^\top - \bm{W}^{-1}\bm{V}_i\bm{c}_i.
\end{align*}

By Assumption~\ref{as:2} and~\ref{as:3}, the matrices $\bm{D}_i$ and $\bm{W}$ are positive definite for each $i$. Thus we have $\widetilde{S}_{np}(\bm{c}_i, \bm{A})\ge 0$. In addition, if either $\bm{c}_i \neq \bm{c}_i^0$ or $\bm{A}\neq \bm{A}^0\bm{H}$, then $\widetilde{S}_{np}(\bm{c}_i, \bm{A})> 0$. Thus $\widetilde{S}_{np}(\bm{c}_i, \bm{A})$ achieves its unique minimum at $(\bm{c}_i^0, \bm{A}^0) $. Thus we have
\begin{equation*}
\widehat{\bm{c}}_i - \bm{c}_i^0 = o_p(1), \qquad i = 1,\dots, n.
\end{equation*}

Next, we show $\widehat{\bm{c}}_i$ is consistent uniformly in $i$.

We can write
\begin{align*}
S_{np} (\bm{c}_i, \bm{A}) - \widetilde{S}_{np}(\bm{c}_i, \bm{A}) =&  \frac{2}{np}\sum_{i=1}^n  \bm{\epsilon}_i^\top \bm{M}_{\bm{A}} \bm{A}^0\bm{f}_i - 
 \frac{2}{np}\sum_{i=1}^n  \bm{\epsilon}_i^\top \bm{M}_{\bm{A}} \bm{\Phi}\bm{c}_i \\
 &+ 
  \frac{1}{np}\sum_{i=1}^n  \bm{\epsilon}_i^\top (\bm{M}_{\bm{A}}-\bm{M}_{\bm{A}^0}) \bm{\epsilon}_i 
  + \frac{\alpha}{np}\sum_{i=1}^n  \bm{c}_i^\top \bm{R}\bm{c}_i. 
\end{align*}
Using Taylor's expansion at $\bm{c}_i$,
\begin{align*}
S_{np} (\bm{c}_i, \bm{A}) - \widetilde{S}_{np}(\bm{c}_i, \bm{A}) =&  \frac{2}{np}\sum_{i=1}^n  \bm{\epsilon}_i^\top \bm{M}_{\bm{A}} \bm{A}^0\bm{f}_i - 
 \frac{2}{np}\sum_{i=1}^n  \bm{\epsilon}_i^\top \bm{M}_{\bm{A}} \bm{\Phi}\bm{c}_i^0 \\
 &+ 
  \frac{1}{np}\sum_{i=1}^n  \bm{\epsilon}_i^\top (\bm{M}_{\bm{A}}-\bm{M}_{\bm{A}^0}) \bm{\epsilon}_i 
  + \frac{\alpha}{np}\sum_{i=1}^n  {\bm{c}_i^0}^\top \bm{R}\bm{c}_i^0\\
  &+ \left( -\frac{2}{np}\bm{\epsilon}_i^\top\bm{M}_{\bm{A}} \bm{\Phi} +  \frac{2\alpha}{np} {\bm{c}_i^0}^\top \bm{R}\right)(\bm{c}_i-\bm{c}_i^0) + \Delta,
\end{align*}
where $\Delta$ denotes the small order terms. Then we have
\begin{align}\label{eq:52}
 \left( -\frac{2}{np}\bm{\epsilon}_i^\top\bm{M}_{\bm{A}} \bm{\Phi} +  \frac{2\alpha}{np} {\bm{c}_i^0}^\top \bm{R}\right)\left(\bm{c}_i-\bm{c}_i^0\right) =&
S_{np} (\bm{c}_i, \bm{A}) - \widetilde{S}_{np}(\bm{c}_i, \bm{A}) -   \frac{2}{np}\sum_{i=1}^n  \bm{\epsilon}_i^\top \bm{M}_{\bm{A}} \bm{A}^0\bm{f}_i \nonumber\\
&+ \frac{2}{np}\sum_{i=1}^n  \bm{\epsilon}_i^\top \bm{M}_{\bm{A}} \bm{\Phi}\bm{c}_i^0 +
  \frac{1}{np}\sum_{i=1}^n  \bm{\epsilon}_i^\top (\bm{M}_{\bm{A}}-\bm{M}_{\bm{A}^0}) \bm{\epsilon}_i \nonumber\\
  &- \frac{\alpha}{np}\sum_{i=1}^n  {\bm{c}_i^0}^\top \bm{R}\bm{c}_i^0+ \Delta
\end{align}
In the above equation, the right-hand side is $o_p(1)$ uniformly in $i$. This is because$ S_{np}(\bm{c}_i, \bm{A})  -\widetilde{S}_{np}(\bm{c}_i, \bm{A}) = o_p(1)$ and $S_{np}(\bm{c}_i, \bm{A})$ and $\widetilde{S}_{np}(\bm{c}_i, \bm{A})$ both consist of summations over $i$. Furthermore, all other terms on the right-hand side are $o_p(1)$ as proved in Lemma~\ref{le:9} and all contain summations over $i$. On the left-hand side of~\eqref{eq:52}, $-\frac{2}{np}\bm{\epsilon}_i^\top\bm{M}_{\bm{A}} \bm{\Phi}$ is $o_p(1)$ uniformly because $\mathbb{E}\left(\frac{1}{np}\left\|\bm{\epsilon}^\top_i\bm{M}_{\bm{A}} \bm{c}_i\right\|\right) = o(1)$ as shown in Lemma~\ref{le:9} $(ii)$. Moreover, the term $ \frac{2\alpha}{np} {\bm{c}_i^0}^\top \bm{R}$ is also $o_p(1)$ uniformly using Lemma~\ref{le:9} $(iv)$ and that we assume $\bm{c}_i$ are bounded uniformly in Assumption~\ref{as:2}. This leads us to the result that
\begin{equation*}
\widehat{\bm{c}}_i - \bm{c}_i^0 = o_p(1), \quad \text{uniformly for all } i = 1,\dots, n
\end{equation*}
Combining the $i$, we have
\begin{equation*}
\frac{\left\|\widehat{\bm{C}} - \bm{C}^0\right\|}{\sqrt{n}} = o_p(1).
\end{equation*}

To prove part $(ii)$, note that the centred objective function satisfies $S_{np}(\bm{c}_i^0 = \bm{0}, \bm{A}^0) = 0$ and, by definition in~\eqref{eq:36}, we have  $S_{np}(\widehat{\bm{c}}_i, \widehat{\bm{A}})\le 0$. Therefore,
\begin{equation*}
0 \ge S_{np}(\widehat{\bm{c}}_i, \widehat{\bm{A}}) = \widetilde{S}_{np}(\widehat{\bm{c}}_i, \widehat{\bm{A}}) + o_p(1).
\end{equation*}
Combined with $\widetilde{S}_{np}(\widehat{\bm{c}}_i, \widehat{\bm{A}})\ge 0$, it must be true that
\begin{equation*}
\widetilde{S}_{np}(\widehat{\bm{c}}_i, \widehat{\bm{A}}) = o_p(1).
\end{equation*}
This implies that 
\begin{equation*}
 \frac{1}{np}\sum_{i=1}^n  \bm{F}_i^\top {\bm{A}^0}^\top \bm{M}_{\bm{A}} \bm{A}^0\bm{F}_i = \text{tr}\left[\frac{{\bm{A}^0}^\top \bm{M}_{\widehat{\bm{A}}}\bm{A}^0}{p}\frac{\bm{F}^\top \bm{F}}{n}\right] = o_p(1).
\end{equation*}
Since $\bm{F}^\top\bm{F}/n = O_p(1)$, it must be true that 
\begin{equation}\label{eq:17}
\frac{{\bm{A}^0}^\top \bm{M}_{\widehat{\bm{A}}}\bm{A}^0}{p} = \frac{{\bm{A}^0}^\top\bm{A}^0}{p} - \frac{{\bm{A}^0}^\top \widehat{\bm{A}}}{p}\frac{\widehat{\bm{A}}^\top\bm{A}^0}{p} = o_p(1).
\end{equation}
By Assumption~\ref{as:4}, ${\bm{A}^0}^\top\bm{A}^0/p$ is invertible. Thus ${\bm{A}^0}^\top \widehat{\bm{A}}/p$ is also invertible.
Next,
\begin{equation*}
\|\bm{P}_{\widehat{\bm{A}}}-\bm{P}_{\bm{A}^0}\|^2 = \text{tr}[(\bm{P}_{\widehat{\bm{A}}}-\bm{P}_{\bm{A}^0})^2] = 2\text{tr}(\bm{I}_r-\widehat{\bm{A}}^\top \bm{P}_{\bm{A}^0}\widehat{\bm{A}}/p).
\end{equation*}
But~\eqref{eq:17} implies $\widehat{\bm{A}}^\top\bm{P}_{\bm{A}^0}\widehat{\bm{A}}/p\rightarrow \bm{I}_r$, which means $\|\bm{P}_{\widehat{\bm{A}}}-\bm{P}_{\bm{A}}\| \rightarrow 0$.
\end{proof}

\subsubsection*{Proof of Theorem~\ref{th:2}}
\begin{proof}
Writing the first equation in~\eqref{eq:1} in matrix notation, we have
\begin{equation}\label{eq:24}
\widehat{\bm{C}} = \left(\bm{\Phi}^\top \bm{M}_{\widehat{\bm{A}}}\bm{\Phi} + \alpha \bm{R} \right)^{-1} \bm{\Phi}^\top \bm{M}_{\widehat{\bm{A}}}\bm{Y}.
\end{equation}
Substitute $\bm{Y} = \bm{\Phi} \bm{C}^0 + \bm{A}^0 \bm{f}^\top+ \bm{E}$ into~\eqref{eq:24} and subtract the matrix $\bm{C}^0$ on both sides, we get
\begin{align*}
\widehat{\bm{C}}-\bm{C}^0=  \left[\left(\bm{\Phi}^\top \bm{M}_{\widehat{\bm{A}}}\bm{\Phi} + \alpha \bm{R} \right)^{-1}\bm{\Phi}^\top \bm{M}_{\widehat{\bm{A}}} \bm{\Phi} - \bm{I}_K \right] \bm{C}^0 \\
+ \left(\bm{\Phi}^\top \bm{M}_{\widehat{\bm{A}}}\bm{\Phi} + \alpha \bm{R} \right)^{-1} \bm{\Phi}^\top \bm{M}_{\widehat{\bm{A}}}\bm{A}^0\bm{F}^\top \\
+ \left(\bm{\Phi}^\top \bm{M}_{\widehat{\bm{A}}}\bm{\Phi} + \alpha \bm{R} \right)^{-1} \bm{\Phi}^\top \bm{M}_{\widehat{\bm{A}}}\bm{E},
\end{align*}
or
\begin{equation}\label{eq:7}
\frac{1}{p}\left(\bm{\Phi}^\top \bm{M}_{\widehat{\bm{A}}}\bm{\Phi} + \alpha \bm{R} \right)\left(\widehat{\bm{C}}-\bm{C}^0\right) = \frac{1}{p}\alpha \bm{R} \bm{C}^0 + \frac{1}{p}\bm{\Phi}^\top \bm{M}_{\widehat{\bm{A}}}\bm{A}^0\bm{F}^\top + \frac{1}{p}\bm{\Phi}^\top \bm{M}_{\widehat{\bm{A}}}\bm{E}
\end{equation}
We first look at the second term on the right-hand side of~\eqref{eq:7}. Recall that $\bm{M}_{\widehat{\bm{A}}} = \bm{I}_p-\widehat{\bm{A}}\widehat{\bm{A}}^\top/p$. We have $\bm{M}_{\widehat{\bm{A}}}\widehat{\bm{A}} = \bm{0}$. Thus
\begin{equation*}
\bm{M}_{\widehat{\bm{A}}}\bm{A}^0 = \bm{M}_{\widehat{\bm{A}}} \left(\bm{A}^0-\widehat{\bm{A}}\bm{H}^{-1}+\widehat{\bm{A}}\bm{H}^{-1}\right) = \bm{M}_{\widehat{\bm{A}}}\left(\bm{A}^0-\widehat{\bm{A}}\bm{H}^{-1}\right),
\end{equation*}
where $\bm{H}$ is defined in~\eqref{eq:37}. Using~\eqref{eq:5}, it follows that
\begin{align}\label{eq:19}
\frac{1}{p}\bm{\Phi}^\top \bm{M}_{\widehat{\bm{A}}}\bm{A}^0\bm{F}^\top&= -\frac{1}{p}\bm{\Phi}^\top \bm{M}_{\widehat{\bm{A}}}\left(I1 + \dots +I8\right)\left(\frac{{\bm{A}^0}^\top\widehat{\bm{A}}}{p}\right)^{-1}\left(\frac{\bm{F}^\top \bm{F}}{n}\right)^{-1}\bm{F}^\top\\\nonumber
&\equiv J1 + \dots + J8.
\end{align}

In the following, we calculate the order for each from $J1$ to $J8$. Note that $I1$ to $I8$ are defined in~\eqref{eq:6}. Before we begin, for simplicity, denote
\begin{equation}\label{eq:21}
\bm{G} \equiv \left(\frac{{\bm{A}^0}^\top\widehat{\bm{A}}}{p}\right)^{-1}\left(\frac{\bm{F}^\top \bm{F}}{n}\right)^{-1}.
\end{equation}
We prove in Lemma~\ref{le:10} that $\bm{G} = O_p(1)$. We also use the fact that $\left\|\bm{M}_{\widehat{\bm{A}}}\right\| = O_p(1)$.
Now
\begin{equation}\label{eq:25}
J1 = -\frac{1}{p}\bm{\Phi}^\top \bm{M}_{\widehat{\bm{A}}}\left(I1\right)\bm{G}\bm{F}^\top.
\end{equation}
Since $I1 = O_p\left(\frac{\sqrt{p}}{n}\left\|\bm{C}^0-\widehat{\bm{C}}\right\|^2\right)$, using the result from Lemma~\ref{le:8} $(i)$, the term $J1$ is bounded in norm by $O_p\left(\frac{1}{\sqrt{n}}\left\|\bm{C}^0-\widehat{\bm{C}}\right\|^2\right)$. Thus it is also $o_p\left(\left\|\bm{C}^0-\widehat{\bm{C}}\right\|\right)$.
\begin{align}\label{eq:30}
J2 &= -\frac{1}{p}\bm{\Phi}^\top \bm{M}_{\widehat{\bm{A}}}\left(I2\right)\left(\frac{{\bm{A}^0}^\top\widehat{\bm{A}}}{p}\right)^{-1}\left(\frac{\bm{F}^\top \bm{F}}{n}\right)^{-1}\bm{F}^\top\nonumber\\
&= \frac{1}{p}\bm{\Phi}^\top \bm{M}_{\widehat{\bm{A}}}\bm{\Phi}\left(\widehat{\bm{C}}-\bm{C}^0\right)\bm{F}\left(\bm{F}^\top\bm{F}\right)^{-1}\bm{F}^\top. 
\end{align}
For the term $J2$, since it is not a small order term, we keep it as what it is.

Now consider 
\begin{align}\label{eq:26}
J3 &= -\frac{1}{p}\bm{\Phi}^\top \bm{M}_{\widehat{\bm{A}}}\left(I3\right)\left(\frac{{\bm{A}^0}^\top\widehat{\bm{A}}}{p}\right)^{-1}\left(\frac{\bm{F}^\top \bm{F}}{n}\right)^{-1}\bm{F}^\top\nonumber\\
&= \frac{1}{np^2}\bm{\Phi}^\top \bm{M}_{\widehat{\bm{A}}}\bm{\Phi} \left(\widehat{\bm{C}}-\bm{C}^0\right)\bm{E}^\top\widehat{\bm{A}}\bm{G}\bm{F}^\top.
\end{align}
We take $\bm{E}^\top \widehat{\bm{A}} = \bm{E}^\top \left(\widehat{\bm{A}}-\bm{A}^0\bm{H}\right) + \bm{E}^\top \bm{A}^0\bm{H}$, where the order of each term can be found in Lemma~\ref{le:2} $(i)$. Again using the result of Lemma~\ref{le:8} $(i)$ and $(iii)$, it can be shown that $J3$ is $o_p\left(\left\|\bm{C}^0-\widehat{\bm{C}}\right\|\right)$.

Next
\begin{align}\label{eq:27}
\|J4\| &= \left\|-\frac{1}{p}\bm{\Phi} \bm{M}_{\widehat{\bm{A}}}I4\left(\frac{{\bm{A}^0}^\top\widehat{\bm{A}}}{p}\right)^{-1}\left(\frac{\bm{F}^\top \bm{F}}{n}\right)^{-1}\bm{F}^\top\right\|\nonumber\\
&= O_p\left(\frac{\bm{M}_{\widehat{\bm{A}}}\bm{A}^0}{\sqrt{p}}\left\|\bm{C}^0-\widehat{\bm{C}}\right\|\right).
\end{align}
Using Proposition~\ref{pro}, we have $\frac{1}{\sqrt{p}} \bm{M}_{\widehat{\bm{A}}}\bm{A}^0 = \bm{M}_{\widehat{\bm{A}}}\frac{1}{\sqrt{p}}\left(\bm{A}^0-\widehat{\bm{A}}\bm{H}^{-1}\right) = o_p\left(1\right)$, Thus, $\|J4\| = o_p\left(\left\|\bm{C}^0-\widehat{\bm{C}}\right\|\right)$.

It can also be proven that $\|J5\| = o_p\left(\left\|\bm{C}^0-\widehat{\bm{C}}\right\|\right)$.

Then we consider 
\begin{align*}
J6 &= -\frac{1}{p}\bm{\Phi}^\top \bm{M}_{\widehat{\bm{A}}}I6 \bm{G}\bm{F}^\top \\
&= -\frac{1}{np^2}\bm{\Phi}^\top \bm{M}_{\widehat{\bm{A}}}\bm{A}^0\bm{F}^\top \bm{E}^\top \widehat{\bm{A}} \bm{G} \bm{F}^\top\\
&= -\frac{1}{np^2} \bm{\Phi}^\top \bm{M}_{\widehat{\bm{A}}}\left(\bm{A}^0-\widehat{\bm{A}}\bm{H}^{-1}\right)\bm{F}^\top \bm{E}^\top \widehat{\bm{A}} \bm{G} \bm{F}^\top,
\end{align*}
where the last equation comes from $\bm{M}_{\widehat{\bm{A}}}\widehat{\bm{A}}\bm{H}^{-1} = \bm{0}$. Now 
\begin{align*}
\left\|\bm{F}^\top \bm{E}^\top \widehat{\bm{A}}\right\| &= \left\|\bm{E}^\top\widehat{\bm{A}} \right\|\\
&\le \left\|\ \bm{E}^\top \left(\widehat{\bm{A}}-\bm{A}^0\bm{H}\right)\right\| + \left\|\bm{E}^\top \bm{A}^0\bm{H}\right\|\\
&= O_p\left(\frac{p}{\min \left(\sqrt{n},\sqrt{p}\right)}\left\|\bm{C}^0-\widehat{\bm{C}}\right\|\right) + O_p\left(\sqrt{n}\right) + O_p\left(\frac{p}{\sqrt{n}}\right) + O_p\left(\sqrt{np}\right)\\
&= O_p\left( \frac{p}{\sqrt{n}}\left\|\bm{C}^0-\widehat{\bm{C}}\right\|\right)  + O_p\left(\frac{p}{\sqrt{n}}\right) + O_p\left(\sqrt{np}\right),
\end{align*}
using Lemma~\ref{le:2}. Thus,
\begin{align}\label{eq:28}
\|J6\| &\le \left\|\frac{1}{np^2} \bm{\Phi}^\top \bm{M}_{\widehat{\bm{A}}}\left(\bm{A}^0-\widehat{\bm{A}}\bm{H}^{-1}\right)\right\| \left\|\bm{F}^\top \bm{E}^\top \widehat{\bm{A}}\right\| \left\|\bm{G}\right\|\left\|\bm{F}^\top\right\|\nonumber\\
&= -\frac{1}{np^2}\times \left[ O_p\left(\frac{p}{\sqrt{n}}\left\|\bm{C}^0-\widehat{\bm{C}}\right\|\right) + O_p\left(\frac{p}{\min \left(n,p\right)}\right)\right] \nonumber\\
&\times \left[ O_p\left(  \frac{p}{\sqrt{n}}\left\|\bm{C}^0-\widehat{\bm{C}}\right\|\right) + O_p\left(\frac{p}{\sqrt{n}}\right) + O_p\left(\sqrt{np} \right)\right] \times O_p\left(\sqrt{n}\right)\nonumber\\
& = o_p\left(\left\|\bm{C}^0-\widehat{\bm{C}}\right\|\right) + O_p\left(\frac{1}{n^2}\right) + O_p\left(\frac{1}{n\sqrt{p}}\right) + O_p\left(\frac{1}{np}\right) + O_p\left( \frac{1}{p\sqrt{p}}\right),
\end{align}
where Proposition~\ref{pro} is used in the first equation, and the second equation is a result of the calculation on the orders.
Next
\begin{align}\label{eq:31}
J7 &=  -\frac{1}{p}\bm{\Phi}^\top \bm{M}_{\widehat{\bm{A}}}I7 \bm{G}\bm{F}^\top \nonumber\\
&= -\frac{1}{np}\bm{\Phi}^\top \bm{M}_{\widehat{\bm{A}}}\bm{E}\bm{F} \left(\frac{\bm{F}^\top \bm{F}}{n}\right)^{-1}\bm{F}^\top.
\end{align}
{This term is not a small order term, so we keep it as what it is.
And lastly, the proof of order for the term $J8$ is too long, so we show in Lemma~\ref{le:11} that  
\begin{equation}\label{eq:29}
J8 = o_p\left(\left\|\bm{C}^0-\widehat{\bm{C}}\right\|\right) + O_p\left(\frac{1}{\min \left(n,p\right)} \right) + O_p\left(\frac{\sqrt{n}}{\sqrt{p}}\frac{1}{\min \left(n,p\right)} \right).
\end{equation}

Collecting terms from $J1$ to $J8$, we can write~\eqref{eq:7} as
\begin{equation*}
\left(\frac{1}{p}\bm{\Phi}^\top \bm{M}_{\widehat{\bm{A}}}\bm{\Phi} + \frac{1}{p}\alpha \bm{R} \right) \left(\widehat{\bm{C}}-\bm{C}^0\right) = \frac{1}{p}\alpha \bm{R} \bm{C}^0 + J1 + \dots + J8 + \frac{1}{p}\bm{\Phi}^\top \bm{M}_{\widehat{\bm{A}}}\bm{E}.
\end{equation*}
Combining the results we have found for $J1, J3, J4, J5, J6$ and $ J8$ in \cref{eq:26,eq:27,eq:28,eq:29},
\begin{align}\label{eq:33}
\left(\frac{1}{p}\bm{\Phi}^\top \bm{M}_{\widehat{\bm{A}}}\bm{\Phi} + o_p(1) \right) \left(\widehat{\bm{C}}-\bm{C}^0\right) - J2 =&  \frac{1}{p}\alpha \bm{R} \bm{C}^0 + \frac{1}{p}\bm{\Phi}^\top \bm{M}_{\widehat{\bm{A}}}\bm{E} + J7 \nonumber\\
&+ O_p\left(\frac{1}{\min (n,p)}\right) + O_p\left(\frac{\sqrt{n}}{p\sqrt{p}} \right)+ O_p\left(\frac{1}{\sqrt{np}}\right). 
\end{align}
Substitute $J2 $ and $J7$ from~\cref{eq:30,eq:31} into~\eqref{eq:33}, we have
\begin{align}\label{eq:32}
&\left(\frac{1}{p}\bm{\Phi}^\top \bm{M}_{\widehat{\bm{A}}}\bm{\Phi} + o_p(1) \right) \left(\widehat{\bm{C}}-\bm{C}^0\right) - \frac{1}{p}\bm{\Phi}^\top \bm{M}_{\widehat{\bm{A}}}\bm{\Phi}\left(\widehat{\bm{C}}-\bm{C}^0\right)\bm{F}\left(\bm{F}^\top\bm{F}\right)^{-1}\bm{F}^\top \nonumber\\
=&  \frac{1}{p}\alpha \bm{R} \bm{C}^0 + \frac{1}{p}\bm{\Phi}^\top \bm{M}_{\widehat{\bm{A}}}\bm{E} -\frac{1}{p}\bm{\Phi}^\top \bm{M}_{\widehat{\bm{A}}}\bm{E}\bm{F} \left(\bm{F}^\top \bm{F}\right)^{-1}\bm{F}^\top\nonumber  \\
&+ O_p\left(\frac{1}{\min (n,p)}\right) + O_p\left(\frac{\sqrt{n}}{p\sqrt{p}}\right) + O_p\left(\frac{1}{\sqrt{np}}\right).
\end{align}
We combine the two terms on the left-hand side of~\eqref{eq:32} and also combine the second and third term on the right-hand side of~\eqref{eq:32}, then we get
\begin{align*}
\frac{1}{p}\bm{\Phi}^\top \bm{M}_{\widehat{\bm{A}}}\bm{\Phi} (\widehat{\bm{C}}-\bm{C}^0) \left(\bm{I}_n -\bm{F} \left(\bm{F}^\top\bm{F}\right)^{-1}\bm{F}^\top\right) =& \frac{1}{p}\alpha \bm{R} \bm{C}^0 +\frac{1}{p}\bm{\Phi}^\top \bm{M}_{\widehat{\bm{A}}}\bm{E}\left(\bm{I}_n- \bm{F} (\bm{F}^\top \bm{F})^{-1}\bm{F}^\top\right) \\
&+  O_p\left(\frac{1}{\min (n,p)}\right) + O_p\left(\frac{\sqrt{n}}{p\sqrt{p}} \right)+ O_p\left(\frac{1}{\sqrt{np}}\right).
\end{align*}
Let $\bm{Q}(\widehat{\bm{A}}) \equiv \frac{1}{p}\bm{\Phi}^\top \bm{M}_{\widehat{\bm{A}}}\bm{\Phi}$, and $\bm{M}_{\bm{F}} \equiv \bm{I}_n- \bm{F} (\bm{F}^\top \bm{F})^{-1}\bm{F}^\top$. Left multiplying $\bm{Q}^{-1}(\widehat{\bm{A}})$ to both sides of the equation above, we have
\begin{align*}
 (\widehat{\bm{C}}-\bm{C}^0)\bm{M}_{\bm{F}} =& \bm{Q}(\widehat{\bm{A}})^{-1}\frac{1}{p}\alpha \bm{R} \bm{C}^0 + \bm{Q}(\widehat{\bm{A}})^{-1}\frac{1}{p}\bm{\Phi}^\top \bm{M}_{\widehat{\bm{A}}}\bm{E}\bm{M}_{\bm{F}}\\
&+ O_p\left(\frac{1}{\min (n,p)}\right) + O_p\left(\frac{\sqrt{n}}{p\sqrt{p}} \right)+ O_p\left(\frac{1}{\sqrt{np}}\right)\\
 =& \bm{Q}^{-1}(\bm{A}^0)\frac{1}{p}\alpha \bm{R} \bm{C}^0 + \bm{Q}^{-1}(\bm{A}^0)\frac{1}{p}\bm{\Phi}^\top \bm{M}_{\bm{A}^0}\bm{E}\bm{M}_{\bm{F}}\\
&+ O_p\left(\frac{1}{\min (n,p)}\right) + O_p\left(\frac{\sqrt{n}}{p\sqrt{p}} \right)+ O_p\left(\frac{1}{\sqrt{np}}\right),
\end{align*}
where in the last equation, we substitute $\bm{Q}(\widehat{\bm{A}})$ with $\bm{Q}(\bm{A}^0)$ using Lemma~\ref{le:6} and substitute $\frac{1}{\sqrt{np}}\bm{\Phi}^\top \bm{M}_{\widehat{\bm{A}}}\bm{E}$ with $\frac{1}{\sqrt{np}}\bm{\Phi}^\top \bm{M}_{\bm{A}^0}\bm{E}$ using Lemma~\ref{le:5}. Note that $\sqrt{p}O_p\left(\frac{\left\|\bm{C}^0-\widehat{\bm{C}}\right\|^2}{n}\right)$ in the result of Lemma~\ref{le:5} is dominated by $\sqrt{p}\frac{\left\|\bm{C}^0-\widehat{\bm{C}}\right\|}{\sqrt{n}}$. Next by multiplying a scale of $\frac{\sqrt{p}}{\sqrt{n}}$,
\begin{align}\label{eq:35}
 \frac{\sqrt{p}}{\sqrt{n}}(\widehat{\bm{C}}-\bm{C}^0)\bm{M}_{\bm{F}} =& \bm{Q}(\bm{A}^0)^{-1}\frac{1}{p}\alpha \bm{R} \bm{C}^0 + \bm{Q}(\bm{A}^0)^{-1}\frac{1}{\sqrt{np}}\bm{\Phi}^\top \bm{M}_{\bm{A}^0}\bm{E}\bm{M}_{\bm{F}} \nonumber\\
 &+ \frac{\sqrt{p}}{\sqrt{n}} \times O_p\left(\frac{1}{\min (n,p)}\right) +  \frac{\sqrt{p}}{\sqrt{n}} \times O_p\left(\frac{\sqrt{n}}{p\sqrt{p}}\right) + O_p\left(\frac{1}{\sqrt{np}}\right)\nonumber\\
 =& O_p(1),
\end{align}
when $n$ and $p$ are of the same order, that is $p/n \rightarrow \rho >0$.
}
\end{proof}

\subsubsection*{Proof of Theorem~\ref{th:3}} 

From~\eqref{eq:35}, we have when $p/n \rightarrow \rho >0$,
\begin{equation*}
 \frac{\sqrt{p}}{\sqrt{n}}(\widehat{\bm{C}}-\bm{C}^0)\bm{M}_{\bm{F}} = \bm{Q}(\bm{A}^0)^{-1}\frac{1}{\sqrt{np}}\bm{\Phi}^\top \bm{M}_{\bm{A}^0}\bm{E}\bm{M}_{\bm{F}} + o_p(1).
\end{equation*}

Using Lemma~\ref{le:4} we have, for any vector $\bm{b} = (b_1, \dots,  b_n)^\top$,
\begin{equation}\label{eq:38}
\frac{1}{\sqrt{np}}\bm{\Phi}^\top \bm{M}_{\bm{A}^0}\bm{E}\bm{M}_{\bm{F}}\bm{b} \overset{d}{\to} \mathcal{N}(0,\bm{L}),
\end{equation}
where $\bm{L}$ is defined in~\eqref{eq:45}.

Multiplying the constant matrix $\bm{Q}(\bm{A}^0)^{-1}$ to~\eqref{eq:38}, we have the result
\begin{equation*}
\bm{Q}(\bm{A}^0)^{-1}\frac{1}{\sqrt{np}}\bm{\Phi}^\top \bm{M}_{\bm{A}^0}\bm{E}\bm{M}_{\bm{F}}\bm{b} \overset{d}{\to} \mathcal{N}\left(0,\bm{Q}(\bm{A}^0)^{-1}\bm{L}\bm{Q}(\bm{A}^0)^{-1}\right).
\end{equation*}
The theorem is thus proved.

\subsection*{Appendix B}\label{app:b}

In this section, we provide the proposition used in Appendix A, along with its proof.
\begin{proposition}\label{pro}
Under Assumptions~\ref{as:1} to~\ref{as:4}, we have the following statements:
\begin{enumerate}
\item[(i)]
The matrix $\bm{V}_{np}$ defined in~\eqref{eq:9} is invertible and $\bm{V}_{np}\overset{p}{\to} \bm{V}$, where the $r \times r$ matrix $\bm{V}$ is a diagonal matrix consisting of the eigenvalues of  $\bm{\Sigma}_{\bm{F}}\bm{\Sigma}_{\bm{A}}$;
\item[(ii)]
Let 
\begin{equation}\label{eq:37}
\bm{H} = (\bm{F}^\top\bm{F}/n)^{-1}({\bm{A}^0}^\top\widehat{\bm{A}}/p)^{-1}\bm{V}_{np},
\end{equation}
 then $\bm{H}$ is $r\times r$ invertible matrix and
\begin{equation*}
\frac{1}{p}\|\widehat{\bm{A}}-\bm{A}^0\bm{H}\|^2 = O_p\left(\frac{1}{n}\left\|\bm{C}^0-\widehat{\bm{C}}\right\|^2\right) + O_p\left(\frac{1}{\min (n,p)}\right).
\end{equation*}
\end{enumerate} 
\end{proposition}
\begin{proof}
Write the second equation in~\eqref{eq:1} in a matrix form, we have
\begin{equation*}
\frac{1}{np}(\bm{Y}-\bm{\Phi} \widehat{\bm{C}})(\bm{Y}-\bm{\Phi} \widehat{\bm{C}})^\top\widehat{\bm{A}} = \widehat{\bm{A}}\bm{V}_{np}.
\end{equation*}
By~\eqref{eq:16}, we also have 
\begin{equation}\label{eq:43}
\bm{Y}-\bm{\Phi} \widehat{\bm{C}} = \bm{\Phi}(\bm{C}^0-\widehat{\bm{C}}) + \bm{A}^0\bm{F}^\top + \bm{E}.
\end{equation}
Plugging it in~\eqref{eq:43} and by expanding terms, we obtain
\begin{align}\label{eq:6}
\widehat{\bm{A}}\bm{V}_{np} =& \frac{1}{np}\left[ \bm{\Phi}(\bm{C}^0-\widehat{\bm{C}}) + \bm{A}^0\bm{F}^\top + \bm{E}\right]\left[ \bm{\Phi}(\bm{C}^0-\widehat{\bm{C}}) + \bm{A}^0\bm{F}^\top + \bm{E}\right]^\top\widehat{\bm{A}}\nonumber\\
=&  \frac{1}{np} \bm{\Phi} (\bm{C}^0-\widehat{\bm{C}})(\bm{C}^0-\widehat{\bm{C}})^\top\bm{\Phi}^\top\widehat{\bm{A}} + \frac{1}{np}\bm{\Phi}(\bm{C}^0-\widehat{\bm{C}})\bm{F} {\bm{A}^0}^\top\widehat{\bm{A}}  \nonumber\\
&+ \frac{1}{np} \bm{\Phi} (\bm{C}^0-\widehat{\bm{C}})\bm{E}^\top \widehat{\bm{A}} +\frac{1}{np}\bm{A}^0\bm{F}^\top (\bm{C}^0-\widehat{\bm{C}})^\top\bm{\Phi}^\top \widehat{\bm{A}},\nonumber\\
&+ \frac{1}{np} \bm{E}(\bm{C}^0-\widehat{\bm{C}})^\top\bm{\Phi}^\top \widehat{\bm{A}} +\frac{1}{np}\bm{A}^0\bm{F}^\top \bm{E}^\top \widehat{\bm{A}}  \nonumber\\
&+\frac{1}{np}\bm{E}\bm{F} {\bm{A}^0}^\top \widehat{\bm{A}} + \frac{1}{np}\bm{EE}^\top \widehat{\bm{A}}  \nonumber\\
&+ \frac{1}{np}\bm{A}^0\bm{F}^\top\bm{F}{\bm{A}^0}^\top\widehat{\bm{A}}\nonumber\\
\equiv& I1 +\dots + I9.
\end{align}
The above can be rewritten as
\begin{equation}\label{eq:40}
\widehat{\bm{A}}\bm{V}_{np} - \bm{A}^0(\bm{F}^\top\bm{F}/n)({\bm{A}^0}^\top\widehat{\bm{A}}/p) = I1 + \dots + I8.
\end{equation}
Right multiplying $(\bm{F}^\top\bm{F}/n)^{-1}({\bm{A}^0}^\top\widehat{\bm{A}}/p)^{-1} $ on each side, we obtain
\begin{equation}\label{eq:5}
\widehat{\bm{A}}\left[\bm{V}_{np}({\bm{A}^0}^\top\widehat{\bm{A}}/p)^{-1} (\bm{F}^\top\bm{F}/n)^{-1}\right] - \bm{A}^0 = (I1 +\dots + I8)({\bm{A}^0}^\top\widehat{\bm{A}}/p)^{-1} (\bm{F}^\top\bm{F}/n)^{-1}.
\end{equation}
Note that the matrix in the square brackets is $\bm{H}^{-1}$, but the invertibility of $\bm{V}_{np}$ hasn't been proved yet. We can write
\begin{equation}\label{eq:39}
\frac{1}{\sqrt{p}}\left\|\widehat{\bm{A}}\left[\bm{V}_{np}({\bm{A}^0}^\top\widehat{\bm{A}}/p)^{-1} (\bm{F}^\top\bm{F}/n)^{-1}\right] - \bm{A}^0\right\|\le \frac{1}{\sqrt{p}} (\|I1\| +\dots + \|I8\|)\|\bm{G}\|,
\end{equation}
where $\bm{G}$ is defined in~\eqref{eq:21} and $\|\bm{G}\|$ is proved to be $O_p(1)$ in Lemma~\ref{le:10}. In the following, we find the order for each term on the right-hand side of~\eqref{eq:39}. We repeatedly use results from Lemma~\ref{le:8}, where the orders of the matrices $\bm{\bm{\Phi}}, \bm{A}$ and $\bm{F}$ are given. The first term
\begin{align*}
\frac{1}{\sqrt{p}}\|I1\| &\le \frac{1}{\sqrt{p}}\frac{1}{np}\|\bm{\Phi}\|\|(\bm{C}^0-\widehat{\bm{C}})(\bm{C}^0-\widehat{\bm{C}})^\top\|\|\bm{\Phi}^\top\|\|\widehat{\bm{A}}\|\\
&= O_p\left(\frac{1}{n}\left\|\bm{C}^0-\widehat{\bm{C}}\right\|^2\right) = o_p\left(\frac{1}{\sqrt{n}}\left\|\bm{C}^0-\widehat{\bm{C}}\right\|\right).
\end{align*}
For the second term
\begin{align*}
\frac{1}{\sqrt{p}}\|I2\| &\le \frac{1}{\sqrt{p}} \frac{1}{np} \|\bm{\Phi}\| \|(\bm{C}^0-\widehat{\bm{C}})\| \|\bm{F}\| \|{\bm{A}^0}^\top\| \|\widehat{\bm{A}}\|\\
&= O_p\left(\frac{1}{\sqrt{n}}\left\|\bm{C}^0-\widehat{\bm{C}}\right\|\right).
\end{align*}
The terms $I3$ to $I5$ are all $O_p(\frac{1}{\sqrt{n}}\left\|\bm{C}^0-\widehat{\bm{C}}\right\|)$. The proofs are similar to the proof for $I2$ since they are only a switch in the order of the matrices. For the sixth term
\begin{equation*}
\frac{1}{\sqrt{p}}\|I6\| \le \frac{1}{\sqrt{p}} \frac{1}{np} \|\bm{A}^0\|\|\bm{F}^\top \bm{E}^\top\|\|\widehat{\bm{A}}\| = O_p\left(\frac{1}{\sqrt{n}}\right),
\end{equation*}
by Lemma~\ref{le:1} $(i)$. Similarly, for the next term
\begin{equation*}
\frac{1}{\sqrt{p}}\|I7\| \le \frac{1}{\sqrt{p}}\frac{1}{np}\|\bm{E}\bm{F}\| \|{\bm{A}^0}^\top\|\|\widehat{\bm{A}}\| = O_p\left(\frac{1}{\sqrt{n}}\right).
\end{equation*} 
For the last term
\begin{equation*}
\frac{1}{\sqrt{p}} \|I8\| \le \frac{1}{\sqrt{p}}\frac{1}{np}\|\bm{EE}^\top\| \|\widehat{\bm{A}}\| = O_p\left(\frac{1}{\sqrt{n}}\right) + \left(\frac{1}{\sqrt{p}}\right),
\end{equation*}
where Lemma~\ref{le:1} $(iv)$ is used.

Putting all the above together, we have
\begin{align}\label{eq:41}
\frac{1}{\sqrt{p}}\left\|\widehat{\bm{A}}\left[\bm{V}_{np}({\bm{A}^0}^\top\widehat{\bm{A}}/p)^{-1} (\bm{F}^\top\bm{F}/n)^{-1}\right] - \bm{A}^0\right\| =& O_p\left(\frac{1}{\sqrt{n}}\left\|\bm{C}^0-\widehat{\bm{C}}\right\|\right) \nonumber\\
&+ O_p\left(\frac{1}{\min (\sqrt{n}, \sqrt{p})}\right).
\end{align}

To show $(i)$, left multiply~\eqref{eq:40} by $\frac{1}{p}\widehat{\bm{A}}^\top$. Using $\widehat{\bm{A}}^\top\widehat{\bm{A}}/p = \bm{I}_r$, we  have
\begin{equation*}
\bm{V}_{np} - (\widehat{\bm{A}}^\top\bm{A}^0/p)(\bm{F}^\top\bm{F}/n)({\bm{A}^0}^\top\widehat{\bm{A}}/p) = \frac{1}{p}\widehat{\bm{A}}^\top ( I1 + \dots + I8) = o_p(1),
\end{equation*}
where the last equality is using Lemma~\ref{le:8} $(v)$ and that $p^{-1/2}(\|I1\|+\dots+\|I8\|) = o_p(1)$ from~\eqref{eq:41}. Thus,
\begin{equation*}
\bm{V}_{np} =  (\widehat{\bm{A}}^\top\bm{A}^0/p)(\bm{F}^\top\bm{F}/n)({\bm{A}^0}^\top\widehat{\bm{A}}/p)  + o_p(1).
\end{equation*}
We have shown in~\eqref{eq:17} that $\widehat{\bm{A}}^\top\widehat{\bm{A}}^0$ is invertible, thus $\bm{V}_{np}$ is invertible. To obtain the limit of $\bm{V}_{np}$, left multiply~\eqref{eq:40} by $\frac{1}{p}{\bm{A}^0}^\top$ to yield
\begin{equation*}
({\bm{A}^0}^\top\widehat{\bm{A}}/p)\bm{V}_{np} - ({\bm{A}^0}^\top\bm{A}^0/p)(\bm{F}^\top\bm{F}/n)({\bm{A}^0}^\top\widehat{\bm{A}}/p) = o_p(1),
\end{equation*}
or
\begin{equation}\label{eq:42}
({\bm{A}^0}^\top\bm{A}^0/p)(\bm{F}^\top\bm{F}/n)({\bm{A}^0}^\top\widehat{\bm{A}}/p) + o_p(1) = ({\bm{A}^0}^\top\widehat{\bm{A}}/p)\bm{V}_{np}
\end{equation}
because $p^{-1}{\bm{A}^0}^\top (\|I1\| +\dots \|I8\|) = o_p(1)$. Equation~\eqref{eq:42} shows that the columns of $({\bm{A}^0}^\top\widehat{\bm{A}}/p)$ are the eigenvectors of the matrix $({\bm{A}^0}^\top\bm{A}^0/p)(\bm{F}^\top\bm{F}/n)$, and that $\bm{V}_{np}$ consists of the eigenvalues of the same matrix in the limit. Thus, $\bm{V}_{np}\overset{p}{\to} \bm{V}$, where the $r \times r$ matrix $\bm{V}$ is a diagonal matrix consisting of the eigenvalues of  $\bm{\Sigma}_{\bm{F}}\bm{\Sigma}_{\bm{A}}$.

For $(ii)$, since $\bm{V}_{np}$ is invertible, $\bm{H}$ is also invertible we can write~\eqref{eq:41} as 
\begin{equation*}
\frac{1}{\sqrt{p}}\left\|\widehat{\bm{A}}\bm{H}^{-1}- \bm{A}^0\right\| = O_p\left(\frac{1}{\sqrt{n}}\left\|\bm{C}^0-\widehat{\bm{C}}\right\|\right) + O_p\left(\frac{1}{\min (\sqrt{n}, \sqrt{p})}\right).
\end{equation*}

By right multiplying the matrix $\bm{H}$, we obtain $(ii)$.

\end{proof}

\subsection*{Appendix C}\label{app:c}

In this section, we state all the lemmas used for previous theorems and propositions, along with the proofs of the lemmas.

Lemma~\ref{le:4} is stated in Section~\ref{se:4.2}. We provide the proof here.
\begin{proof}
For any vector $\bm{b} = (b_1,\dots, b_n)^\top$,
\begin{equation*}
\frac{1}{\sqrt{np}}\bm{\Phi}^\top \bm{M}_{\bm{A}^0}\bm{E}\bm{b} = \frac{1}{\sqrt{np}}\sum_i^n\sum_j^p\bm{\omega}_j\epsilon_{ji}b_i \equiv  \frac{1}{\sqrt{np}}\sum_i^n\sum_j^p\bm{x}_{ij},
\end{equation*}
where $\bm{\omega}_j$ is the $j$th column in the matrix $\bm{\Phi}^\top \bm{M}_{\bm{A}^0}$. Since we assume $\epsilon_{ji}$ are i.i.d., the variance of the above quantity is given by
\begin{equation*}
\text{var}\left(\frac{1}{\sqrt{np}}\bm{\Phi}^\top \bm{M}_{\bm{A}^0}\bm{E} \bm{b} \right) = \text{var}\left(\frac{1}{\sqrt{np}}\sum_i^n\sum_j^p \bm{x}_{ij}\right) = \frac{1}{np}\sum_i^p\sum_j^n b_j b_i \sigma^2 E\left(\bm{\omega}_i\bm{\omega}_j^\top\right).
\end{equation*}
The Lindeberg condition is assumed to hold in Assumption~\ref{as:7}. Thus we have a central limit theorem result
\begin{equation*}
\frac{1}{\sqrt{np}}\bm{\Phi}^\top \bm{M}_{\bm{A}^0}\bm{E}\bm{b} = \frac{1}{\sqrt{np}}\sum_i^n\sum_j^p\bm{x}_{ij}\overset{d}{\to} \mathcal{N}(0,\bm{L}),
\end{equation*}
where $\bm{L}$ is defined in~\eqref{eq:45}.

\end{proof}

\begin{lemma}\label{le:8}
Under Assumptions~\ref{as:1}-\ref{as:3}, we have
\begin{enumerate}[(i)]
\item
$\frac{1}{\sqrt{p}}\|\bm{\Phi}\| = O_p(1)$
\item
$\frac{1}{\sqrt{p}}\|\bm{A}\| = O_p(1)$
\item
$ \frac{1}{\sqrt{n}}\|\bm{F}\| = O_p(1)$
\item
$ \frac{1}{\sqrt{np}}\|\bm{E}\| = O_p(1)$
\item
$\frac{1}{\sqrt{p}}\|\widehat{\bm{A}}\| = O_p(1)$
\end{enumerate}

\end{lemma}
\begin{proof}
In Assumption~\ref{as:1}, we assume the basis functions $\phi_k(u), \ k = 1,\dots, K$ are bounded. The $p\times K$ basis matrix $\bm{\Phi}$ contains discrete evaluations on the basis functions so each element is $O_p(1)$, thus $\bm{\Phi}$ is of order $\sqrt{p}$.
Similarly, using Assumption~\ref{as:3}, we have results $(ii)$ and $(iii)$. Using Assumption~\ref{as:4}, we have result $(iv)$.
Lastly, $(v)$ is directly from the restriction $\widehat{\bm{A}}^\top \widehat{\bm{A}}/p = \bm{I}_r$. 
\end{proof}

\begin{lemma} \label{le:1}
Under Assumptions~\ref{as:1} to~\ref{as:5}, we have
\begin{enumerate}[(i)]
\item
$
\frac{1}{np}\|\bm{EF}\|^2 = O_p(1)
$
\item
$\frac{1}{np}\|\bm{E}^\top\bm{\Phi}\|^2 = O_p(1)
$ and $\frac{1}{np}\|\bm{E}^\top \bm{A}^0\|^2 = O_p(1)
$ 
\item
$\frac{1}{np}\|\bm{F}^\top \bm{E}^\top \bm{A}^0\|^2 = O_p(1)
$ and $\frac{1}{np}\|\bm{\Phi}^\top \bm{E}^\top \bm{F}\|^2 = O_p(1)
$
\item
$\|\bm{E}^\top \bm{E}\|^2 = O_p(n^2p) + O_p(p^2n);\\ \|\bm{EE}^\top \|^2  =O_p(n^2p) + O_p(p^2n);\\
 \|\bm{F}^\top \bm{E}^\top \bm{E}\|^2  = O_p(n^2p) + O_p(p^2n); \\
  \|\bm{\Phi}^\top \bm{E}^\top \bm{E}\|^2  = O_p(n^2p) + O_p(p^2n); \\
   \|\bm{\Phi}^\top \bm{E}^\top \bm{E}\bm{A}^0\|^2  = O_p(n^2p) + O_p(p^2n); \\
 \|\bm{F}^\top \bm{E}^\top \bm{E} \bm{F}\|^2  = O_p(n^2p) + O_p(p^2n)$.
\end{enumerate}
\end{lemma}

\begin{proof}
For $(i)$
\begin{align*}
\mathbb{E}\left(\frac{1}{np}\|\bm{EF}\|^2\right) &= \mathbb{E}\left(\frac{1}{np}\sum_{k=1}^p\sum_{i=1}^n\sum_{j=1}^n\epsilon_{ki}\epsilon_{kj}\bm{f}_i^\top\bm{f}_j\right) \\
& =  \frac{1}{np}\sum_{k=1}^p\sum_{i=1}^n\sum_{j=1}^n\mathbb{E}(\epsilon_{ki}\epsilon_{kj})\mathbb{E}(\bm{f}_i^\top\bm{f}_j) =  O(1),
\end{align*}
where the second equation uses the independence between $\epsilon_{ki}$ and $\bm{f}_j$ assumed in Assumption~\ref{as:5}.

The proof of $(ii)$ and $(iii)$ is similar to $(i$). For $(iv)$,
\begin{align*}
\mathbb{E}\left(\|\bm{E}^\top \bm{E}\|^2 \right) &= \mathbb{E}\left( \sum_{ij}^n\sum_{kl}^p \epsilon_{kj}\epsilon_{lj}\epsilon_{ki}\epsilon_{li}\right) \\
&= \sum_{i\neq j}^n \sum_{k=l}^p \mathbb{E}
(\epsilon^2_{kj}) \mathbb{E}(\epsilon_{ki}^2) + \sum_{i= j}^n \sum_{k\neq l}^p  \mathbb{E}
(\epsilon^2_{kj}) \mathbb{E}(\epsilon_{lj}^2)+  \sum_{i= j}^n \sum_{k=l}^p \mathbb{E}(\epsilon_{kj}^4) \\
&= O(n^2p) + O(p^2n) + O(np)\\
&= O(n^2p) + O(p^2n),
\end{align*}
where Assumption~\ref{as:4} is used. The proof of $\|\bm{EE}^\top\|$ is the same. The orders of  $\|\bm{F}^\top\bm{E}^\top\bm{E}\|$ and $\|\bm{F}^\top \bm{E}^\top \bm{E} \bm{F}\|^2$ are the same since 
\begin{equation*}
\mathbb{E}\left(\|\bm{F}^\top\bm{E}^\top \bm{E}\|^2 \right) = \mathbb{E}\left( \sum_{ij}^n\sum_{kl}^p \epsilon_{kj}\epsilon_{lj}\epsilon_{ki}\epsilon_{li}\|\bm{f}_i\|^2\right),
\end{equation*}
and
\begin{equation*}
\mathbb{E}\left(\|\bm{F}^\top\bm{E}^\top \bm{E}\bm{F}\|^2 \right) = \mathbb{E}\left( \sum_{ij}^n\sum_{kl}^p \epsilon_{kj}\epsilon_{lj}\epsilon_{ki}\epsilon_{li}\|\bm{f}_i\|^4\right),
\end{equation*}
where the order of $\bm{f}_i$ is assumed to be $O_p(1)$ in Assumption~\ref{as:3}.

\end{proof}

\begin{lemma}\label{le:9}
Under Assumptions~\ref{as:1}-\ref{as:6},
\begin{enumerate}[(i)]
\item
$\frac{1}{np}\sum_{i=1}^n  \bm{\epsilon}_i^\top \bm{M}_{\bm{A}} \bm{A}^0\bm{F}_i = o_p(1)$
\item
 $\frac{1}{np}\sum_{i=1}^n  \bm{\epsilon}_i^\top \bm{M}_{\bm{A}} \bm{\Phi}\bm{c}_i = o_p(1)$
 \item
 $\frac{1}{np}\sum_{i=1}^n  \bm{\epsilon}_i^\top (\bm{M}_{\bm{A}}-\bm{M}_{\bm{A}^0}) \bm{\epsilon}_i = o_p(1)$
 \item
 $ \frac{\alpha}{np}\sum_{i=1}^n  \bm{c}_i^\top \bm{R}\bm{c}_i = o_p(1)$
\end{enumerate}
\end{lemma}

\begin{proof}
We prove $(ii)$. First, we have
\begin{equation*}
\mathbb{E}\left(\left\|\sum_{i=1}^n\bm{\epsilon}_i\right\|^2\right) = \mathbb{E}\left(\sum_{i=1}^n\sum_{j=1}^n\sum_{k=1}^p\epsilon_{ik}\epsilon_{jk}\right) = \sum_{i=j}^n\sum_{k=1}^p\mathbb{E}\left(\epsilon_{ik}^2\right) = O(np).
\end{equation*}

Since $\bm{M}_{\bm{A}} = \bm{I}_{p}- \bm{A}\bm{A}^\top/p$, we have
\begin{align}\label{eq:46}
\frac{1}{np}\sum_{i=1}^n  \bm{\epsilon}_i^\top \bm{M}_{\bm{A}} \bm{\Phi}\bm{c}_i = \frac{1}{np}\sum_{i=1}^n  \bm{\epsilon}_i^\top \bm{\Phi}\bm{c}_i - \frac{1}{np^2}\sum_{i=1}^n  \bm{\epsilon}_i^\top \bm{A}{\bm{A}}^\top \bm{\Phi}\bm{c}_i.
\end{align}

The first term on the right of~\eqref{eq:46} is $o_p(1)$ since
\begin{align*}
\mathbb{E}\left(\left\|\sum_{i=1}^n\bm{\epsilon}^\top_i\bm{\Phi}\bm{c}_i\right\|^2\right) &= \mathbb{E}\left(\left\|\sum_{j=1}^p\sum_{i=1}^n\epsilon_{ji}\bm{\phi}_j^\top \bm{c}_i\right\|^2\right) \\
&= \mathbb{E}\left(\sum_t^p\sum_s^n\sum_j^p\sum_i^n\epsilon_{ji}\epsilon_{ts}\bm{c}_i^\top\bm{\phi}_j\bm{\phi}_t^\top\bm{c}_s\right)\\
&= \sum_t^p\sum_s^n\sum_j^p\sum_i^n\mathbb{E}(\epsilon_{ji}\epsilon_{ts})\mathbb{E}(\bm{c}_i^\top\bm{\phi}_j\bm{\phi}_t^\top\bm{c}_s)\\
&= \sum_j^p\sum_i^n\sigma^2\mathbb{E}(\bm{c}_i^\top\bm{\phi}_j\bm{\phi}_ j^\top\bm{c}_i)\\
&= O(np),
\end{align*}
where the third equation uses Assumption~\ref{as:5}; the fourth equation uses the assumption that $\epsilon_{ji}$ are independent in both directions.   

The second term on the right-hand side of~\eqref{eq:46} is also $o_p(1)$ since
\begin{align*}
\mathbb{E}\left(\left\|\sum_{i=1}^n\bm{\epsilon}^\top_i\bm{AA}^\top\bm{\Phi}\bm{c}_i\right\|^2\right) &= \mathbb{E}\left(\left\|\sum_{j=1}^p\sum_{i=1}^n\epsilon_{ji}\bm{a}_j^\top\bm{A}^\top\bm{\Phi}\bm{c}_i\right\|^2\right) \\
&= \mathbb{E}\left(\sum_t^p\sum_s^n\sum_j^p\sum_i^n\epsilon_{ji}\epsilon_{ts}\bm{c}_i^\top\bm{\Phi}^\top\bm{Aa}_j\bm{a}_t^\top\bm{A}^\top\bm{\Phi}\bm{c}_s\right)\\
&= \sum_t^p\sum_s^n\sum_j^p\sum_i^n\mathbb{E}(\epsilon_{ji}\epsilon_{ts})\mathbb{E}(\bm{c}_i^\top\bm{\Phi}^\top\bm{Aa}_j\bm{a}_t^\top\bm{A}^\top\bm{\Phi}\bm{c}_s)\\
&= \sum_j^p\sum_i^n\sigma^2\bm{c}_i^\top\mathbb{E}(\bm{\Phi}_i^\top\bm{Aa}_j\bm{a}_j^\top\bm{A}^\top\bm{\Phi})\bm{c}_i\\
&= O(np^3),
\end{align*}
where the third equality uses the independence in Assumption~\ref{as:5} and the last equality uses the results in Lemma~\ref{le:8}, where $\bm{\Phi}$ and $\bm{A}$ are both $O_p(\sqrt{p})$.

The proofs for $(i)$ and $(iii)$ are similar. And $(iv)$ is a direct result from Assumption~\ref{as:6}. 
\end{proof}

\begin{lemma}\label{le:10}
Under Assumptions~\ref{as:1}-\ref{as:5} , we have
 \begin{equation*}
 \bm{G} \equiv \left({\bm{A}^0}^\top\widehat{\bm{A}}/p\right)^{-1}\left(\bm{F}^\top\bm{F}/n\right)^{-1} = O_p(1).
 \end{equation*}
\end{lemma}
\begin{proof}
The matrix $\bm{F}^\top\bm{F}/n$ is positive definite by Assumption~\ref{as:3}. We have shown in the proof of Theorem~\ref{th:1} in~\eqref{eq:17} that the matrix ${\bm{A}^0}^\top\widehat{\bm{A}}/p$ is invertible, thus is also positive definite. Therefore, $\lambda_{\min} \left({\bm{A}^0}^\top\widehat{\bm{A}}/p\right) >  0$, and $\lambda_{\min} \left(\bm{F}^\top\bm{F}/n\right)>  0$, where $\lambda_{\min} (\cdot)$ denotes the smallest eigenvalue of a matrix. So we have
\begin{equation*}
\left({\bm{A}^0}^\top\widehat{\bm{A}}/p\right)^{-1} = O_p(1), \quad \left(\bm{F}^\top\bm{F}/n\right)^{-1} = O_p(1).
\end{equation*}

\end{proof}

\begin{lemma}\label{le:2}
We have the following
\begin{enumerate}[(i)]
\item
\begin{equation*}
\left\|\bm{E}^\top (\widehat{\bm{A}}-\bm{A}^0\bm{H})\right\| = O_p\left(\frac{p}{\min (\sqrt{n}, \sqrt{p})}\left\|\bm{C}^0-\widehat{\bm{C}}\right\|\right) + O_p(\sqrt{n}) + O_p\left( \frac{p}{\sqrt{n}}\right).
\end{equation*}
\item
\begin{equation*}
\left\|\bm{F}^\top \bm{E}^\top (\widehat{\bm{A}}-\bm{A}^0\bm{H})\right\| = O_p\left(\frac{p}{\min (\sqrt{n}, \sqrt{p})}\left\|\bm{C}^0-\widehat{\bm{C}}\right\|\right) + O_p(\sqrt{n}) + O_p\left( \frac{p}{\sqrt{n}}\right).
\end{equation*}
\end{enumerate}

\begin{proof}
For $(i)$, from Proposition~\ref{pro}, we can write
\begin{align*}
\|\bm{E}^\top (\widehat{\bm{A}}-\bm{A}^0\bm{H})\| &= \| \bm{E}^\top (I1 + \dots, I8)\bm{G}\| \\
& \le \| \bm{E}^\top I1\bm{G}\| + \dots + \|\bm{E}^\top I8 \bm{G}\| \\
&= \|a1\|+ \dots + \|a8\|.
\end{align*}

To find the order for each term, the results from Lemma~\ref{le:8} are repeatedly used where the order of the matrices $\bm{\Phi}, \bm{A}, \bm{\widehat{\bm{A}}}$ and $\bm{F}$ are given. 
\begin{align*}
\|a1\| &= \left\|\bm{E}^\top \frac{1}{np} \bm{\Phi} (\bm{C}^0-\widehat{\bm{C}})(\bm{C}^0-\widehat{\bm{C}})^\top\bm{\Phi}^\top\widehat{\bm{A}} \bm{G}\right\|\\
&\le \frac{1}{np} \|\bm{E}^\top \bm{\Phi}\| \left\|\bm{C}^0-\widehat{\bm{C}}\right\|^2\|\bm{\Phi}\|\|\widehat{\bm{A}}\|\|\bm{G}\| \\
&= O_p\left(\frac{\sqrt{p}}{\sqrt{n}}\left\|\bm{C}^0-\widehat{\bm{C}}\right\|^2\right) = o_p\left(\sqrt{p}\left\|\bm{C}^0-\widehat{\bm{C}}\right\|\right),
\end{align*}
where the order of $\|\bm{E}^\top \bm{\Phi}\|$ is from Lemma~\ref{le:1} $(ii)$. The orders of $\|\bm{\Phi}\|$, $\|\widehat{\bm{A}}\|$ and $\|\bm{G}\|$can be found from Lemmas~\ref{le:8} and~\ref{le:10}. Similarly,
\begin{align*}
\|a2\| &= \left\|\bm{E}^\top \frac{1}{n}\bm{\Phi}(\bm{C}^0-\widehat{\bm{C}})\bm{F} \left(\frac{\bm{F}^\top\bm{F}}{n}\right)^{-1}\right\|\\
&\le \frac{1}{n}\|\bm{E}^\top\bm{\Phi}\|\left\|\bm{C}^0-\widehat{\bm{C}}\right\|\|\bm{F}\| \left\|\left(\frac{\bm{F}^\top\bm{F}}{n}\right)^{-1}\right\|\\
&= O_p\left(\sqrt{p}\left\|\bm{C}^0-\widehat{\bm{C}}\right\|\right).
\end{align*}
\begin{align*}
\|a3\| &= \left\|\bm{E}^\top  \frac{1}{np} \bm{\Phi} (\bm{C}^0-\widehat{\bm{C}})\bm{E}^\top \widehat{\bm{A}}\bm{G}\right\| \\
&\le \frac{1}{np}\|\bm{E}^\top\bm{\Phi}\|\left\|\bm{C}^0-\widehat{\bm{C}}\right\|\|\bm{E}^\top\|\|\widehat{\bm{A}}\|\|\bm{G}\| \\
&= O_p\left(\sqrt{p}\left\|\bm{C}^0-\widehat{\bm{C}}\right\|\right).
\end{align*}
\begin{align*}
\|a4\| &= \left\|\bm{E}^\top  \frac{1}{np}\bm{A}^0\bm{F}^\top (\bm{C}^0-\widehat{\bm{C}})^\top\bm{\Phi}^\top \widehat{\bm{A}}\bm{G}\right\| \\
&\le \frac{1}{np}\|\bm{E}^\top \bm{A}^0\|\|\bm{F}^\top\|\left\|\bm{C}^0-\widehat{\bm{C}}\right\|\|\bm{\Phi}\|\|\widehat{\bm{A}}\|\|\bm{G}\| \\
&= O_p\left(\sqrt{p}\left\|\bm{C}^0-\widehat{\bm{C}}\right\|\right),
\end{align*}
where Lemma~\ref{le:1} $(ii)$ is used.
\begin{align*}
\|a5\| &= \|\bm{E}^\top\frac{1}{np} \bm{E}(\bm{C}^0-\widehat{\bm{C}})^\top\bm{\Phi}^\top \widehat{\bm{A}}\bm{G}\| \\
&\le \frac{1}{np} \|\bm{E}^\top \bm{E}\|\left\|\bm{C}^0-\widehat{\bm{C}}\right\|\|\bm{\Phi}\|\|\widehat{\bm{A}}\|\|\bm{G}\| \\
&= O_p\left(\sqrt{p}\left\|\bm{C}^0-\widehat{\bm{C}}\right\|\right) + O_p\left(\frac{p}{\sqrt{n}}\left\|\bm{C}^0-\widehat{\bm{C}}\right\|\right),
\end{align*}
where Lemma~\ref{le:1} $(iv)$ is used.
\begin{align*}
\|a6\| &= \|\bm{E}^\top \frac{1}{np} \bm{A}^0\bm{F}^\top \bm{E}^\top  \widehat{\bm{A}}\bm{G}\| \\
&\le \frac{1}{np} \|\bm{E}^\top \bm{A}^0\bm{F}^\top \bm{E}^\top (\widehat{\bm{A}}-\bm{A}^0\bm{H})\bm{G}\| + \frac{1}{np} \|\bm{E}^\top \bm{A}^0\bm{F}^\top\bm{E}^\top \bm{A}^0\bm{HG}\| \\
&\le  \frac{1}{np} \|\bm{E}^\top \bm{A}^0\|\|\bm{F}^\top \bm{E}^\top\|\| \widehat{\bm{A}}-\bm{A}^0\bm{H}\|\|\bm{G}\| + \frac{1}{np} \|\bm{E}^\top \bm{A}^0\|\|\bm{F}^\top \bm{E}^\top \bm{A}^0\|\|\bm{HG}\| \\
&= O_p \left(\sqrt{p} (\frac{\left\|\bm{C}^0-\widehat{\bm{C}}\right\|}{\sqrt{n}} \right)+ O_p\left( \frac{1}{\text{min}(\sqrt{n}, \sqrt{p})} \right) + O_p(1)\\
&= O_p \left(\sqrt{p} (\frac{\left\|\bm{C}^0-\widehat{\bm{C}}\right\|}{\sqrt{n}} \right)+ O_p\left( \frac{1}{\text{min}(\sqrt{n}, \sqrt{p})} \right),
\end{align*} 
where the order of $\widehat{\bm{A}}-\bm{A}^0\bm{H}$ is proved in Proposition~\ref{pro} and the order of other matrix norms can be found in Lemma~\ref{le:1} $(i), (ii)$ and $(iii)$.
\begin{align*}
\|a7\| &= \|\bm{E}^\top \frac{1}{np}\bm{EF} {\bm{A}^0}^\top  \widehat{\bm{A}} \bm{G}\|\\
&= \left\| \frac{1}{n}\bm{E}^\top \bm{EF} \left(\frac{\bm{F}^\top\bm{F}}{n}\right)^{-1}\right\|\\
&\le \frac{1}{n}\left\|\bm{E}^\top \bm{EF}\right\| \left\|\left(\frac{\bm{F}^\top\bm{F}}{n}\right)^{-1} \right\| \\
&= O_p\left(\sqrt{p}\right) + O_p\left(\frac{p}{\sqrt{n}}\right),
\end{align*}
where Lemma~\ref{le:1} $(iv)$ is used.
\begin{align*}
\|a8\| =& \frac{1}{np}\left\| \bm{E}^\top \bm{E E}^\top \widehat{\bm{A}} \bm{G}\right\|\\
\le& \frac{1}{np} \left\|\bm{E}^\top \bm{E E}^\top \bm{A}^0\bm{HG}\right\| + \frac{1}{np} \left\|\bm{E}^\top \bm{E E}^\top (\widehat{\bm{A}}-\bm{A}^0\bm{H})\bm{G}\right\|\\
\le& \frac{1}{np}\|\bm{E}^\top \bm{E}\|\|\bm{E}^\top \bm{A}^0\|\|\bm{H}\|\|\bm{G}\| + \frac{1}{np}\|\bm{E}^\top\bm{E}\|\|\bm{E}^\top\|\|\widehat{\bm{A}}-\bm{A}^0\bm{H}\|\|\bm{G}\|\\
=& \frac{1}{np}\left[O_p(n\sqrt{p}) + O_p(p\sqrt{n})\right]O_p(\sqrt{np}) \\
&+ \frac{1}{np}\left[O_p(n\sqrt{p}) + O_p(p\sqrt{n})\right]O_p(\sqrt{np})\left[O_p\left(\frac{\sqrt{p}\left\|\bm{C}^0-\widehat{\bm{C}}\right\|}{\sqrt{n}} \right)+ O_p\left(\frac{\sqrt{p}}{\min (\sqrt{n}, \sqrt{p})}\right)\right]\\
=&  O_p(\sqrt{n})+ O_p(\sqrt{p}) + O_p\left( \frac{p}{\sqrt{n}}\left\|\bm{C}^0-\widehat{\bm{C}}\right\|\right) + O_p\left(\sqrt{p}\left\|\bm{C}^0-\widehat{\bm{C}}\right\|\right) + O_p\left(\sqrt{n}\right) + O_p\left(\frac{p}{\sqrt{n}}\right)\\
=&  O_p\left( \frac{p}{\sqrt{n}}\left\|\bm{C}^0-\widehat{\bm{C}}\right\|\right) + O_p\left(\sqrt{p}\left\|\bm{C}^0-\widehat{\bm{C}}\right\|\right) + O_p\left(\sqrt{n}
\right) + O_p\left(\frac{p}{\sqrt{n}}\right),
\end{align*}
where the order of $\widehat{\bm{A}}-\bm{A}^0\bm{H}$ is proved in Proposition~\ref{pro} and the order of other matrix norms can be found in Lemma~\ref{le:1} $(i), (ii)$ and $(iii)$.

Combining all the terms, we have
\begin{equation*}
\|\bm{E}^\top (\widehat{\bm{A}}-\bm{A}^0\bm{H})\| = O_p\left(\frac{p}{\min (\sqrt{n}, \sqrt{p})}\left\|\bm{C}^0-\widehat{\bm{C}}\right\|\right) + O_p(\sqrt{n}) + O_p\left(\frac{p}{\sqrt{n}}\right).
\end{equation*}
For $(ii)$, multiplying the matrix $\bm{F}^\top$ in the front does not change the order, using the fact that $\|\bm{F}^\top \bm{E}^\top\bm{ \Phi}\|$ is of the same order as $\|\bm{E}^\top \bm{\Phi}\|$ and that $\|\bm{F}^\top \bm{E}^\top \bm{E}\|$ and $\|\bm{F}^\top \bm{E}^\top \bm{EF}\|$ are of the same order as $\|\bm{E}^\top\bm{ E}\|$, as proved in Lemma~\ref{le:1}.
\end{proof}
\end{lemma}

\begin{lemma}\label{le:3}
Under Assumptions~\ref{as:1}-\ref{as:5}, we have the following
 \renewcommand{\labelenumi}{\roman{enumi}}
\begin{enumerate}[(i)]
\item
\begin{equation*}
\frac{1}{p} \bm{\Phi}^\top (\widehat{\bm{A}}-\bm{A}^0\bm{H}) = O_p\left(\frac{1}{\sqrt{n}}\left\|\bm{C}^0-\widehat{\bm{C}}\right\|\right) + O_p\left(\frac{1}{\min (n, p)}\right)
\end{equation*}
\item
\begin{equation*}
\frac{1}{p} {\bm{A}^0}^\top (\widehat{\bm{A}}-\bm{A}^0\bm{H}) = O_p\left(\frac{1}{\sqrt{n}}\left\|\bm{C}^0-\widehat{\bm{C}}\right\|\right) + O_p\left(\frac{1}{\min (n, p)}\right)
\end{equation*}
\item
\begin{equation*}
\frac{1}{p} \widehat{\bm{A}}^\top (\widehat{\bm{A}}-\bm{A}^0\bm{H}) = O_p\left(\frac{1}{\sqrt{n}}\left\|\bm{C}^0-\widehat{\bm{C}}\right\|\right) + O_p\left(\frac{1}{\min (n, p)}\right)
\end{equation*}
\item
\begin{equation*}
\frac{1}{p} \bm{\Phi}^\top \bm{M}_{\widehat{\bm{A}}} (\widehat{\bm{A}}-\bm{A}^0\bm{H}) = O_p\left(\frac{1}{\sqrt{n}}\left\|\bm{C}^0-\widehat{\bm{C}}\right\|\right) + O_p\left(\frac{1}{\min (n, p)}\right)
\end{equation*}
\end{enumerate}

\end{lemma}
\begin{proof}
For $(i)$, using~\eqref{eq:5}
\begin{align}\label{eq:47}
\bm{\Phi}^\top (\widehat{\bm{A}}-\bm{A}^0\bm{H}) = \bm{\Phi}^\top (I1 + I2 + \dots + I8)\bm{G}.
\end{align}
It can be easily proved that the first five terms in~\eqref{eq:47} are $O_p\left(\frac{p}{\sqrt{n}}\left\|\bm{C}^0-\widehat{\bm{C}}\right\|\right)$ using the results from Lemma~\ref{le:8} and~\ref{le:1}. Recall that $\bm{G} = ({\bm{A}^0}^\top\widehat{\bm{A}}/p)^{-1} (\bm{F}^\top\bm{F}/n)^{-1} $, and that $\bm{G} = O_p(1)$ from Lemma~\ref{le:10}. For the sixth term,
\begin{align*}
\bm{\Phi}^\top I6 \bm{G} &= \frac{1}{np}\bm{\Phi}^\top \bm{A}^0 \bm{F}^\top \bm{E}^\top \widehat{\bm{A}}\bm{G}\\
&= \frac{1}{np} \bm{\Phi}^\top \bm{A}^0 \bm{F}^\top \bm{E}^\top(\widehat{\bm{A}}-\bm{A}^0\bm{H})\bm{G} +  \frac{1}{np} \bm{\Phi}^\top \bm{A}^0 \bm{F}^\top \bm{E}^\top \bm{A}^0\bm{H}\bm{G}\\
&\le \frac{1}{np} \|\bm{\Phi}^\top \|\|\bm{A}^0 \|\|\bm{F}^\top \bm{E}^\top(\widehat{\bm{A}}-\bm{A}^0\bm{H})\|\|\bm{G}\| +  \frac{1}{np} \|\bm{\Phi}^\top \|\|\bm{A}^0\|\| \bm{F}^\top \bm{E}^\top \bm{A}^0\|\|\bm{H}\|\|\bm{G}\|\\
&=  O_p\left(\frac{\sqrt{p}}{\sqrt{n}}\left\|\bm{C}^0-\widehat{\bm{C}}\right\|\right) + O_p\left( \frac{\sqrt{p}}{\min (\sqrt{n}, \sqrt{p})}\right) + O_p\left(\frac{\sqrt{p}}{\sqrt{n}}\right),
\end{align*}
using the results from Lemma~\ref{le:2}. Next,
\begin{align*}
\bm{\Phi}^\top I7 \bm{G} &= \frac{1}{np}\bm{\Phi}^\top \bm{E}\bm{F} {\bm{A}^0}^\top\widehat{\bm{A}}\bm{G}\\
&= \frac{1}{n}\bm{\Phi}^\top \bm{E}\bm{F}\left(\frac{\bm{F}^\top \bm{F}}{n}\right)^{-1}  = O_p\left(\frac{\sqrt{p}}{\sqrt{n}}\right),
\end{align*}
where the order of $\bm{\Phi}^\top \bm{EF}$ is found in Lemma~\ref{le:1} $(iii)$.
\begin{align*}
\bm{\Phi}^\top I8 \bm{G} &= \frac{1}{np}\bm{\Phi}^\top \bm{E} \bm{E}^\top \widehat{\bm{A}}\bm{G}\\
&=  \frac{1}{np}\bm{\Phi}^\top \bm{E} \bm{E}^\top (\widehat{\bm{A}}-\bm{A}^0\bm{H})\bm{G} +  \frac{1}{np}\bm{\Phi}^\top \bm{E} \bm{E}^\top \bm{A}^0\bm{H}\bm{G}\\
&\le  \frac{1}{np}\|\bm{\Phi}^\top \bm{E}\| \|\bm{E}^\top (\widehat{\bm{A}}-\bm{A}^0\bm{H})\|\|\bm{G}\| +  \frac{1}{np}\|\bm{\Phi}^\top \bm{E}\| \|\bm{E}^\top \bm{A}^0\|\|\bm{H}\|\|\bm{G}\|\\
&= O_p\left(\frac{1}{\sqrt{n}}\left\|\bm{C}^0-\widehat{\bm{C}}\right\|\right) + O_p\left(\frac{1}{\min (\sqrt{n}, \sqrt{p})}\right) + O_p(1) = O_p(1),
\end{align*}
where Lemma~\ref{le:1} $(ii)$ and Lemma~\ref{le:2} $(i)$ are used.
Combining the terms, we have proved $(i)$. The proof for $(ii)$ is the same. 

For $(iii)$, we can write
\begin{equation*}
\widehat{\bm{A}}^\top (\widehat{\bm{A}}- \bm{A}^0\bm{H}) =  (\widehat{\bm{A}}- \bm{A}^0\bm{H})^\top  (\widehat{\bm{A}}- \bm{A}^0\bm{H}) + (\bm{A}^0\bm{H})^\top (\widehat{\bm{A}}- \bm{A}^0\bm{H}).
\end{equation*}
The order of the first term on the right can be found in Propostion~\ref{pro}. The order of the second term on the right is proved in $(ii)$. For $(iv)$, we have
\begin{equation*}
 \bm{\Phi}^\top \bm{M}_{\widehat{\bm{A}}} (\widehat{\bm{A}}-\bm{A}^0\bm{H}) = \bm{\Phi}^\top (\widehat{\bm{A}}-\bm{A}^0\bm{H}) + \frac{1}{p} \bm{\Phi}^\top \widehat{\bm{A}} \widehat{\bm{A}}^\top (\widehat{\bm{A}}-\bm{A}^0\bm{H}),
\end{equation*}
where the orders of the two terms are proved in $(i)$ and $(iii)$.
\end{proof}

\begin{lemma}\label{le:6}
Define the matrix 
\begin{equation*}
\bm{Q}(\bm{A}) = \frac{1}{p}\bm{\Phi}^\top \bm{M}_{\bm{A}}\bm{\Phi}.
\end{equation*}
Under Assumptions~\ref{as:1}-\ref{as:4}, it holds
\begin{equation*}
\bm{Q}\left(\widehat{\bm{A}}\right)^{-1}-\bm{Q}\left(\bm{A}^0\right)^{-1} = o_p(1).
\end{equation*}
\end{lemma}
\begin{proof}

\begin{align*}
\bm{Q}\left(\widehat{\bm{A}}\right)-\bm{Q}\left(\bm{A}^0\right) &= \frac{1}{p}\bm{\Phi}^\top \bm{M}_{\widehat{\bm{A}}}\bm{\Phi} - \frac{1}{p}\bm{\Phi}^\top \bm{M}_{\bm{A}^0}\bm{\Phi} = \frac{1}{p}\bm{\Phi}^\top \left(\bm{M}_{\widehat{\bm{A}}}- \bm{M}_{\bm{A}^0}\right)\bm{\Phi} \\
&= \frac{1}{p}\bm{\Phi}^\top \left(\bm{P}_{\bm{A}^0}- \bm{P}_{\widehat{\bm{A}}}\right)\bm{\Phi} = O_p\left(\left\| \bm{P}_{\bm{A}^0}- \bm{P}_{\widehat{\bm{A}}}\right\|\right) = o_p(1),
\end{align*}
using Theorem~\ref{th:1} $(ii)$. In Assumption~\ref{as:2}, we have assumed $\inf_{\bm{A}}\bm{D}(\bm{A}) > 0$, since the second term in $\bm{D}(\bm{A})$ is nonnegative, we have $\inf_{\bm{A}}\bm{Q}(\bm{A}) > 0$, so the matrix $\bm{Q}(\bm{A}^0)$ is invertible. Therefore,
\begin{equation*}
\bm{Q}\left(\widehat{\bm{A}}\right)^{-1} = \left[ \bm{Q}\left(\bm{A}^0\right)^{-1} + o_p(1) \right]^{-1} = \bm{Q}\left(\bm{A}^0\right) + o_p(1).
\end{equation*}
\end{proof}

\begin{lemma}\label{le:7}
Recall $\bm{H}$ defined in Proposition~\ref{pro}, then
\begin{equation*}
\bm{H}\bm{H}^\top =  \left(\frac{{\bm{A}^0}^\top \bm{A}^0}{p}\right)^{-1} + O_p\left(\frac{1}{\sqrt{n}}\left\|\bm{C}^0-\widehat{\bm{C}}\right\|\right) + O_p\left(\frac{1}{\min (n,p)}\right)
\end{equation*}

\end{lemma}
\begin{proof}
We have from Lemma~\ref{le:3}
\begin{align}\label{eq:18}
\frac{1}{p}{\bm{A}^0}^\top (\widehat{\bm{A}}-\bm{A}^0\bm{H}) = O_p\left(\frac{1}{\sqrt{n}}\left\|\bm{C}^0-\widehat{\bm{C}}\right\|\right) + O_p\left(\frac{1}{\min (n,p)}\right),
\end{align}
and 
\begin{align}\label{eq:20}
\frac{1}{p}{\widehat{\bm{A}}}^\top (\widehat{\bm{A}}-\bm{A}^0\bm{H}) = \bm{I}_r- \frac{1}{p}\widehat{\bm{A}}^\top \bm{A}^0\bm{H} = O_p\left(\frac{1}{\sqrt{n}}\left\|\bm{C}^0-\widehat{\bm{C}}\right\|\right) + O_p\left(\frac{1}{\min (n,p)}\right).
\end{align}
Left multiply~\eqref{eq:18} by $\bm{H}^\top$ and sum with the transpose of~\eqref{eq:20} to obtain
\begin{equation*}
\bm{I}_r-\frac{1}{p}\bm{H}^\top {\bm{A}^0}^\top \bm{A}^0\bm{H} = O_p\left(\frac{1}{n}\left\|\bm{C}^0-\widehat{\bm{C}}\right\|\right) + O_p\left(\frac{1}{\min (n,p)}\right).
\end{equation*}
Right multiplying by $\bm{H}^\top$ and left multiplying by ${\bm{H}^\top}^{-1}$, we obtain
\begin{equation*}
\bm{I}_r-\frac{1}{p}{\bm{A}^0}^\top \bm{A}^0 \bm{H}\bm{H}^\top  = O_p\left(\frac{1}{\sqrt{n}}\left\|\bm{C}^0-\widehat{\bm{C}}\right\|\right) + O_p\left(\frac{1}{\min (n,p)}\right).
\end{equation*}
Then left multiplying $\left({\bm{A}^0}^\top \bm{A}^0/p\right)^{-1}$, we have 
\begin{equation*}
\bm{H}\bm{H}^\top = \left(\frac{{\bm{A}^0}^\top \bm{A}^0}{p}\right)^{-1} + O_p\left(\frac{1}{\sqrt{n}}\left\|\bm{C}^0-\widehat{\bm{C}}\right\|\right) + O_p\left(\frac{1}{\min (n,p)}\right).
\end{equation*}
\end{proof}

\begin{lemma}\label{le:5}
Under Assumptions~\ref{as:1}-\ref{as:5}, when $p/n \rightarrow \rho >0$,
\begin{equation*}
 \left\|\frac{1}{\sqrt{np}}\bm{\Phi}^\top \bm{M}_{\widehat{\bm{A}}}\bm{E} -\frac{1}{\sqrt{np}}\bm{\Phi}^\top \bm{M}_{\bm{A}^0} \bm{E} \right\| = \sqrt{p}\times O_p\left(\frac{\left\|\bm{C}^0-\widehat{\bm{C}}\right\|^2}{n}\right) + o_p(1).
\end{equation*}
\begin{proof}
Using 
\begin{equation*}
\bm{M}_{\bm{A}^0} = \bm{I}_p - \bm{A}^0\left({\bm{A}^0}^\top\bm{A}^0\right)^{-1}{\bm{A}^0}^\top,\qquad \bm{M}_{\widehat{\bm{A}}} = \bm{I}_p - \left(\widehat{\bm{A}}\widehat{\bm{A}}^\top\right)/p,
\end{equation*}
we calculate
\begin{align*}
\frac{1}{\sqrt{np}}\bm{\Phi}^\top \bm{M}_{\bm{A}^0} \bm{E} - \frac{1}{\sqrt{np}}\bm{\Phi}^\top \bm{M}_{\widehat{\bm{A}}}\bm{E} =& \frac{1}{p\sqrt{np}}\bm{\Phi}^\top \widehat{\bm{A}}\widehat{\bm{A}}^\top \bm{E} -\frac{1}{p\sqrt{np}}\bm{\Phi}^\top \bm{A}^0 \left(\frac{{\bm{A}^0}^\top \bm{A}^0}{p}\right)^{-1} {\bm{A}^0}^\top \bm{E} \\
=& \frac{1}{p\sqrt{np}}\Bigg\{ \bm{\Phi}^\top (\widehat{\bm{A}}-\bm{A}^0\bm{H})\bm{H}^\top {\bm{A}^0}^\top \bm{E} \\
&+ \bm{\Phi}^\top (\widehat{\bm{A}}-\bm{A}^0\bm{H})(\widehat{\bm{A}}-\bm{A}^0\bm{H})^\top \bm{E}\\
 &+  \bm{\Phi}^\top \bm{A}^0\bm{H}(\widehat{\bm{A}}-\bm{A}^0\bm{H})^\top \bm{E} \\
 &+ \bm{\Phi}^\top \bm{A}^0 \left[\bm{H}\bm{H}^\top -\left(\frac{{\bm{A}^0}^\top \bm{A}^0}{p}\right)^{-1}  \right] {\bm{A}^0}^\top \bm{E} \Bigg\}\\
 \equiv& a + b + c + d,
\end{align*}
where we substitute $\widehat{\bm{A}}$ with $\widehat{\bm{A}}-\bm{A}^0\bm{H} + \bm{A}^0\bm{H}$ in the second equality. So the first term on the right-hand side of the first equality is broken down into four terms, one of which is combined with the second term in the right-hand side of the first equality.

For notation simplicity, we denote 
\begin{align}\label{eq:10}
q = \frac{1}{\sqrt{n}}\left\|\bm{C}^0-\widehat{\bm{C}}\right\| + \frac{1}{\min (\sqrt{n}, \sqrt{p})},
\end{align}
which is used to represent the order in the result of Proposition~\ref{pro}.

We calculate each term:
\begin{align*}
\|a\| &= \left\|\frac{1}{p\sqrt{np}} \bm{\Phi}^\top (\widehat{\bm{A}}-\bm{A}^0\bm{H})\bm{H}^\top \bm{A}^\top \bm{E}\right\|\\
&= \frac{1}{p\sqrt{np}} \times O_p(\sqrt{p})\times O_p(\sqrt{p}q)\times O_p(\sqrt{np})\\
&= O_p\left(\frac{\left\|\bm{C}^0-\widehat{\bm{C}}\right\|}{\sqrt{n}}\right) + \left(\frac{1}{\min (\sqrt{n}, \sqrt{p})}\right) = o_p(1),
\end{align*}
where the order of $\widehat{\bm{A}}-\bm{A}^0\bm{H}$ is $\sqrt{p}q$ as proved in Proposition~\ref{pro} and the order of $\|\bm{\Phi}\|$ $\|\bm{A}^\top\bm{E}\|$ can be found in Lemma~\ref{le:8} and \ref{le:1} $(ii)$ respectively.
And when $p/n \rightarrow \rho >0$,
\begin{align*}
\|b\| &= \left\|\frac{1}{p\sqrt{np}}\bm{\Phi}^\top (\widehat{\bm{A}}-\bm{A}^0\bm{H})(\widehat{\bm{A}}-\bm{A}^0\bm{H})^\top \bm{E}\right\|\\
&\le \frac{1}{p\sqrt{np}} \times O_p(\sqrt{p}) \times O_p(pq^2)\times O_p(\sqrt{np})\\
&= \sqrt{p}\times \left[ O_p\left(\frac{\left\|\bm{C}^0-\widehat{\bm{C}}\right\|^2}{n}\right) + O_p\left(\frac{1}{\min (n, p)}\right)\right]\\
&= \sqrt{p}\times O_p\left(\frac{\left\|\bm{C}^0-\widehat{\bm{C}}\right\|^2}{n}\right) + o_p(1),
\end{align*}
where again Proposition~\ref{pro} and Lemma~\ref{le:8} are used.
And
\begin{align*}
c =& \frac{1}{p\sqrt{np}}\bm{\Phi}^\top \bm{A}^0\bm{H}(\widehat{\bm{A}}-\bm{A}^0\bm{H})^\top \bm{E} \\
=&  \frac{1}{p\sqrt{np}}\bm{\Phi}^\top \bm{A}^0\bm{H}\bm{H}^\top (\widehat{\bm{A}}\bm{H}^{-1}-\bm{A}^0)^\top \bm{E}\\
=& \frac{1}{p\sqrt{np}}\bm{\Phi}^\top \bm{A}^0\left[\bm{H}\bm{H}^\top- \left(\frac{{\bm{A}^0}^\top \bm{A}^0}{p}\right)^{-1}\right] (\widehat{\bm{A}}\bm{H}^{-1}-\bm{A}^0)^\top \bm{E} \\
&+  \frac{1}{p\sqrt{np}}\bm{\Phi}^\top \bm{A}^0 \left(\frac{{\bm{A}^0}^\top \bm{A}^0}{p}\right)^{-1}(\widehat{\bm{A}}\bm{H}^{-1}-\bm{A}^0)^\top \bm{E}\\
\equiv& c1 + c2,
\end{align*}
where the second equality is using $\widehat{\bm{A}}-\bm{A}^0\bm{H} = (\widehat{\bm{A}}\bm{H}^{-1}-\bm{A}^0)\bm{H}$. In the third equality, we subtract $({\bm{A}^0}^\top \bm{A}^0/p )^{-1}$ from $\bm{HH}^\top$ and then add it back.

For $c1$, when $p/n \rightarrow \rho > 0$
\begin{align*}
\|c1\| =& \left\|\frac{1}{p\sqrt{np}}\bm{\Phi}^\top \bm{A}^0\left[\bm{H}\bm{H}^\top- \left(\frac{{\bm{A}^0}^\top \bm{A}^0}{p}\right)^{-1}\right] (\widehat{\bm{A}}\bm{H}^{-1}-\bm{A}^0)^\top \bm{E}\right\|\\
\le& \frac{1}{p\sqrt{np}} \times O_p(\sqrt{p})\times O_p(\sqrt{p})\times \left[O_p\left(\frac{1}{\sqrt{n}}\left\|\bm{C}^0-\widehat{\bm{C}}\right\|\right) + O_p\left(\frac{1}{\min (n,p)}\right)\right]\\
&\times \left[O_p\left(\frac{p}{\min (\sqrt{n}, \sqrt{p})}\left\|\bm{C}^0-\widehat{\bm{C}}\right\|\right) + O_p(\sqrt{n}) + O_p\left(\frac{p}{\sqrt{n}}\right)\right]\\
=& O_p\left(\frac{\sqrt{p}}{\min (\sqrt{n},\sqrt{p})}\frac{\left\|\bm{C}^0-\widehat{\bm{C}}\right\|}{n}\right) + O_p\left(\frac{1}{\sqrt{p}}\frac{\left\|\bm{C}^0-\widehat{\bm{C}}\right\|}{\sqrt{n}}\right) +  O_p\left(\frac{\sqrt{p}}{n}\frac{\left\|\bm{C}^0-\widehat{\bm{C}}\right\|}{\sqrt{n}}\right) \\
&+ O_p\left(\frac{1}{\sqrt{p}}\frac{1}{\min (n,p)}\right) +  O_p\left(\frac{\sqrt{p}}{n}\frac{1}{\min (n,p)}\right) \\
=& o_p(1),
\end{align*}
where the order of $\bm{\Phi}$, $\bm{A}^0$ and $\bm{E}$ are found in Lemma~\ref{le:8}; the order of $\bm{H}\bm{H}^\top- \left({\bm{A}^0}^\top \bm{A}^0/p\right)^{-1}$ is found in Lemma~\ref{le:7}; and the order of $\widehat{\bm{A}}\bm{H}^{-1}-\bm{A}^0$ is found in Proposition~\ref{pro}.
Now for $c2$, using the same lemmas and proposition,
\begin{align*}
\|c2\| &=  \left\|\frac{1}{p\sqrt{np}}\bm{\Phi}^\top \bm{A}^0 \left(\frac{{\bm{A}^0}^\top \bm{A}^0}{p}\right)^{-1}(\widehat{\bm{A}}\bm{H}^{-1}-\bm{A}^0)^\top \bm{E}\right\|\\
&\le \frac{1}{p\sqrt{np}} O_p(\sqrt{p})\times O_p(\sqrt{p}) \times \left[O_p\left(\frac{p}{\min (\sqrt{n}, \sqrt{p})}\left\|\bm{C}^0-\widehat{\bm{C}}\right\|\right) + O_p(\sqrt{n}) + O_p\left(\frac{p}{\sqrt{n}}\right)\right]\\
&= O_p\left(\frac{\sqrt{p}}{\sqrt{n}} \frac{\left\|\bm{C}^0-\widehat{\bm{C}}\right\|}{\sqrt{n}}\right) + O_p\left(\frac{\sqrt{p}}{n}\right),
\end{align*}
which is $o_p(1)$ when $p/n \rightarrow \rho > 0$.

And lastly, we have
\begin{align*}
\|d\| &= \left\|\frac{1}{p\sqrt{np}}\bm{\Phi}^\top \bm{A}^0 \left[\bm{H}\bm{H}^\top -\left(\frac{{\bm{A}^0}^\top \bm{A}^0}{p}\right)^{-1}  \right] {\bm{A}^0}^\top \bm{E}\right\|\\
&\le \frac{1}{p\sqrt{np}} O_p(\sqrt{p})\times O_p(\sqrt{p}) \times  O_p\left(\frac{1}{\sqrt{n}}\left\|\bm{C}^0-\widehat{\bm{C}}\right\| + \frac{1}{\min (n,p)}\right)\times O_p(\sqrt{np})\\
&=O_p\left(\frac{1}{\sqrt{n}}\left\|\bm{C}^0-\widehat{\bm{C}}\right\|\right) + O_p\left(\frac{1}{\min (n,p)}\right) = o_p(1),
\end{align*}
where again Lemma~\ref{le:7} is used.

Thus combining the above terms, we have 
\begin{equation*}
 \left\|\frac{1}{\sqrt{np}}\bm{\Phi}^\top \bm{M}_{\widehat{\bm{A}}}\bm{E} - \frac{1}{\sqrt{np}}\bm{\Phi}^\top \bm{M}_{\bm{A}^0} \bm{E}\right\| = \sqrt{p}\times O_p\left(\frac{\left\|\bm{C}^0-\widehat{\bm{C}}\right\|^2}{n}\right) + o_p(1),
\end{equation*}
when $p/n \rightarrow \rho > 0$.
\end{proof}
\end{lemma}

\begin{lemma}\label{le:11}
Recall $J8$ defined in~\eqref{eq:19}, we have
\begin{equation*}
\|J8\| = o_p\left(\left\|\bm{C}^0-\widehat{\bm{C}}\right\|\right) + O_p\left(\frac{1}{\min \left(n,p\right)} \right) +  O_p\left( \frac{\sqrt{n}}{\sqrt{p}}\frac{1}{\min \left(n,p\right)}\right).
\end{equation*}
\end{lemma}
\begin{proof}
\begin{align*}
J8 &=  -\frac{1}{p}\bm{\Phi}^\top \bm{M}_{\widehat{\bm{A}}}I8 \bm{G}\bm{F}^\top \\
&= -\frac{1}{np^2}\bm{\Phi}^\top \bm{M}_{\widehat{\bm{A}}} \bm{E}\bm{E}^\top \widehat{\bm{A}}\bm{G}\bm{F}^\top \\
&= -\frac{1}{np^2}\bm{\Phi}^\top \bm{E}\bm{E}^\top \widehat{\bm{A}}\bm{G}\bm{F}^\top + \frac{1}{np^3}\bm{\Phi}^\top \widehat{\bm{A}}\widehat{\bm{A}}^\top \bm{E}\bm{E}^\top \widehat{\bm{A}}\bm{G}\bm{F}^\top \\
&\equiv I + II,
\end{align*}
where we use $\bm{M}_{\widehat{\bm{A}}} = \bm{I}_p-\widehat{\bm{A}}\widehat{\bm{A}}^\top/p$.
For $I$,
\begin{equation*}
I = -\frac{1}{np^2}\bm{\Phi}^\top \bm{E}\bm{E}^\top \left(\widehat{\bm{A}}-\bm{A}^0\bm{H}\right)\bm{G}\bm{F}^\top -\frac{1}{np^2}\bm{\Phi}^\top \bm{E}\bm{E}^\top \bm{A}^0\bm{H}\bm{G}\bm{F}^\top,
\end{equation*}
then
\begin{align*}
\|I\| \le & \frac{1}{np^2}\left\|\bm{\Phi}^\top \bm{E}\bm{E}^\top\right\|\|\widehat{\bm{A}}-\bm{A}^0\bm{H}\|\|\bm{G}\|\|\bm{F}\| + \frac{1}{np^2}\left\|\bm{\Phi}^\top \bm{E}\bm{E}^\top \bm{A}^0\right\|\|\bm{H}\|\|\bm{G}\|\|\bm{F}\|\\
=& \frac{1}{np^2} \times \left[ O_p\left(p\sqrt{n}\right) + O_p\left(n\sqrt{p}\right)\right] \times O_p\left(\sqrt{p}q\right) \times O_p\left(\sqrt{n}\right) \\
&+ \frac{1}{np^2}\times \left[O_p\left(p\sqrt{n} \right)+O_p\left( n\sqrt{p}\right)\right]\times O_p\left(\sqrt{n}\right)\\
=& O_p\left(\frac{1}{\sqrt{p}}q\right) + O_p\left(\frac{\sqrt{n}}{p}q\right),
\end{align*}
where the order of $\left\|\bm{\Phi}^\top \bm{E}\bm{E}^\top\right\|$ and $\left\|\bm{\Phi}^\top \bm{E}\bm{E}^\top \bm{A}^0\right\|$ are found in Lemma~\ref{le:1}; the order of $\|\widehat{\bm{A}}-\bm{A}^0\bm{H}\|$ is from Proposition~\ref{pro}; and the orders of $\|\bm{F}\|$ and $\|\bm{G}\|$ are found in Lemma~\ref{le:8} $(iii)$ and Lemma~\ref{le:10} respectively.

For $II$,
\begin{align*}
II =& \frac{1}{np^3}\bm{\Phi}^\top \widehat{\bm{A}}\left(\widehat{\bm{A}} -\bm{A}^0\bm{H} + \bm{A}^0\bm{H}\right)^\top \bm{E}\bm{E}^\top \left(\widehat{\bm{A}} -\bm{A}^0\bm{H} + \bm{A}^0\bm{H}\right)\bm{G}\bm{F}^\top \\
=& \frac{1}{np^3}\bm{\Phi}^\top \widehat{\bm{A}}\left[ \left(\widehat{\bm{A}} -\bm{A}^0\bm{H}\right)^\top \bm{E}\bm{E}^\top \left(\widehat{\bm{A}} -\bm{A}^0\bm{H}\right)+ \left(\bm{A}^0\bm{H}\right)^\top \bm{E}\bm{E}^\top \left(\widehat{\bm{A}} -\bm{A}^0\bm{H}\right) \right.\\
&+ \left. \left(\widehat{\bm{A}} -\bm{A}^0\bm{H}\right)^\top \bm{E}\bm{E}^\top \bm{A}^0\bm{H} + \left(\bm{A}^0\bm{H}\right)^\top \bm{E}\bm{E}^\top \bm{A}^0\bm{H}\right]\bm{G}\bm{F}^\top,
\end{align*}

then
\begin{align*}
\|II\| &\le \frac{1}{np^3}\|\bm{\Phi}^\top \|\|\widehat{\bm{A}}\|\left[ \|\widehat{\bm{A}} -\bm{A}^0\bm{H}\|^2 \|\bm{E}\bm{E}^\top\| + \|\widehat{\bm{A}} -\bm{A}^0\bm{H}\| \|\bm{E}\bm{E}^\top \bm{A}^0 \|+ \|{\bm{A}^0}^\top \bm{E}\bm{E}^\top \bm{A}^0\|\right]\|\bm{G}\|\|\bm{F}\|\\
&= \frac{1}{np^2}\left\{\left[O_p\left(p\sqrt{n} \right)+O_p\left( n\sqrt{p}\right)\right]\times \left(pq^2 + \sqrt{p}q + 1\right)\right\}\times O_p\left(\sqrt{n}\right)\\
&= O_p\left[\left(\frac{1}{p} + \frac{\sqrt{n}}{p\sqrt{p}}\right)\left(pq^2 + \sqrt{p}q\right)\right],
\end{align*}
where the order of $\left\| \bm{E}\bm{E}^\top\bm{A}^0\right\|$ and $\left\|{\bm{A}^0}^\top \bm{E}\bm{E}^\top \bm{A}^0\right\|$ are found in Lemma~\ref{le:1}; the order of $\|\widehat{\bm{A}}-\bm{A}^0\bm{H}\|$ is from Proposition~\ref{pro}; and the orders of $\|\bm{\Phi}\|$, $\|\bm{F}\|$ and $\|\bm{G}\|$ are found in Lemma~\ref{le:8} $(i)$, $(iii)$ and Lemma~\ref{le:10} respectively.

Combining $I$ and $II$, we have 
\begin{equation*}
\|J8\| = O_p\left[\left(\frac{1}{p} + \frac{\sqrt{n}}{p\sqrt{p}}\right)\left(pq^2 + \sqrt{p}q\right)\right].
\end{equation*}
Since $1 = O\left(\sqrt{p}q\right) $, the term $\sqrt{p}q$ is dominated by $pq^2$, thus
\begin{align*}
\|J8\| &= O_p\left[\left(\frac{1}{p} + \frac{\sqrt{n}}{p\sqrt{p}}\right)pq^2 \right]\\
&= O_p\left[\left(1 + \frac{\sqrt{n}}{\sqrt{p}}\right)\left(\frac{\left\|\bm{C}^0-\widehat{\bm{C}}\right\|^2}{n} + \frac{1}{\min \left(n,p\right)}\right)\right]\\
&= o_p\left(\left\|\bm{C}^0-\widehat{\bm{C}}\right\|\right) + O_p\left(\frac{1}{\min \left(n,p\right)} \right) +  O_p\left( \frac{\sqrt{n}}{\sqrt{p}}\frac{1}{\min \left(n,p\right)}\right).
\end{align*}

\end{proof}
\end{document}